\pdfoutput=1 
\documentclass[11pt]{article}
\usepackage{fullpage}
\usepackage[pagebackref, colorlinks = true, linkcolor = blue, urlcolor  = blue, citecolor = purple, bookmarks, hypertexnames=false]{hyperref}
\usepackage{amssymb}
\usepackage{amsmath}
\usepackage{amsthm}
\usepackage{graphicx,color,colordvi}
\usepackage{bbm}
\usepackage{varioref}
\usepackage[capitalise]{cleveref}
\usepackage{stmaryrd}
\usepackage[utf8]{inputenc}
\usepackage[blocks]{authblk}
\usepackage{dsfont}
\usepackage{mathtools}
\usepackage{authblk}
\usepackage{wrapfig}
\usepackage[english]{babel}
\usepackage{physics}
\usepackage{braket}
\usepackage[table]{xcolor}
\usepackage[center]{caption}
\usepackage{tikz}
\usetikzlibrary{shapes,arrows,patterns,positioning,shapes.geometric,arrows.meta,decorations.markings}
\usepackage{ifthen}
\usepackage{pgfplots}
\pgfplotsset{compat=1.9}
\usetikzlibrary{shapes,arrows.meta}
\usetikzlibrary{positioning}
\usetikzlibrary{shapes.geometric}
\RequirePackage[framemethod=default]{mdframed}
\usepackage[margin=1in]{geometry}
\usepackage{comment}
\usepackage{url}
\usepackage{ upgreek }
\usepackage{makecell}
\usepackage{fancyref} 
\usepackage{float}
\usepackage{enumitem}
\usepackage{autonum}
\usepackage{subcaption}
\newfloat{algorithm}{t}{lop}
\usepackage{array,rotating ,makecell, multirow, tabularx}
\usepackage{scrextend}
\usepackage{mathrsfs}
\usepackage{enumitem}
\usepackage{soul}

\usepackage{IEEEtrantools}
\usepackage[normalem]{ulem}
\usepackage{cancel}

\DeclareMathOperator{\opt}{opt}

\DeclareMathOperator{\id}{id}
\DeclareMathOperator{\co}{co}
\DeclareMathOperator{\cb}{cb}

\newcommand{\1}{\ensuremath{\mathbbm{1}}}

\newcommand{\SWAP}{\mathsf{SWAP}}

\newtheoremstyle{newdefinition}{}{}{\normalfont}{}{\bfseries}{}
{ }
{\thmname{#1} \thmnumber{#2}\thmnote{ (#3)}}

\newtheoremstyle{newplain}{}{}{\itshape}{}{\bfseries}{}{1em}
{\thmname{#1} \thmnumber{#2}\thmnote{ (#3)}}

\newtheoremstyle{newremark}{}{}{\normalfont}{}{\bfseries}{}{1em}
{\thmname{#1}}

\theoremstyle{newdefinition}
\newtheorem{definition}{Definition}[]

\theoremstyle{newplain}
\newtheorem{theorem}{Theorem}[section]
\newtheorem{lemma}{Lemma}[section]
\newtheorem{proposition}{Proposition}[section]
\newtheorem{corollary}{Corollary}[section]
\newtheorem{remark}{Remark}[section]

\newtheoremstyle{myplain}{5pt}{5pt}{\itshape}{0pt}{\bfseries}{}{5pt plus 1pt minus 1pt}{}
\theoremstyle{myplain}
\newtheorem*{theorem*}{Theorem}
\newtheorem*{corollary*}{Corollary}
\newtheorem*{lemma*}{Lemma}
\newtheorem*{definition*}{Definition}
\newtheorem*{remark*}{Remark}
\newtheorem*{proposition*}{Proposition}

\DeclareMathOperator{\R}{\mathbb{R}}

\DeclareMathOperator{\C}{\mathbb{C}}

\DeclareMathOperator{\cH}{\mathcal{H}}

\DeclareMathOperator{\cS}{\mathcal{S}}
\DeclareMathOperator{\cT}{\mathcal{T}}
\DeclareMathOperator{\cK}{\mathcal{K}}
\DeclareMathOperator{\cC}{\mathcal{C}}
\DeclareMathOperator{\cB}{\mathcal{B}}

\DeclareMathOperator{\cE}{\mathcal{E}}

\DeclareMathOperator{\cX}{\mathcal{X}}
\DeclareMathOperator{\cY}{\mathcal{Y}}

\DeclareMathOperator{\cD}{\mathcal{D}}

\DeclareMathOperator{\cQ}{\mathcal{Q}}
\DeclareMathOperator{\cP}{\mathcal{P}}
\DeclareMathOperator{\cR}{\mathcal{R}}

\title{Multi-Indexed Schatten Quasi- and Anti-Norms \\ and Additivity of Conditional Rényi Entropies}
\author[]{Jan Kochanowski\thanks{Email: jan.kochanowski@inria.fr}}

\affil[]{Inria, Télécom Paris - LTCI, Institut Polytechnique de Paris, 91120 Palaiseau, France}

\date{\today}

\begin{document}

\maketitle

\begin{abstract}
\noindent We introduce and study multi-indexed Schatten spaces which lift the commutative $\ell_{q_1}\![\ell_{q_2}\![\ell_{q_3}\![\dots]]]$ -spaces to the quantum toolbox. They extend Pisier's operator-valued Schatten norms and the recently introduced two-indexed Schatten quasi-norms. For every ordered tuple $\cQ=(q_1,\dots,q_n)$ of non-zero indices we define a quasi-norm (for positive indices) or an anti-norm (for negative indices) on operators on an $n$-fold tensor product of Hilbert spaces. We do this for indices of the same sign satisfying the mutual compatibility condition $|\frac{1}{q_i}-\frac{1}{q_j}|\leq1$. Beyond establishing desirable properties of these quasi- and anti-norms via factorization formulas that combine operator space and matrix-analytic techniques, we lift several central tools and results of the seminal [Devetak, Junge, King, Ruskai, CMP 2006] to the multi-indexed and quasi-normed setting. These include a general non-commutative Minkowski inequality relating swapped systems and two identity-removal equalities. As a consequence, complete hypercontractivity and complete reverse hypercontractivity tensorize for all compatible indices, including negative ones.

Concerning quantum information theory we prove that the completely bounded minimal and the maximal output conditional entropies are additive under tensor products of quantum channels for all $\alpha\geq\frac12$. This extends the case $\alpha>1$ of [Fawzi, Kochanowski, Rouzé, Van Himbeeck, CMP 2026] and the unconditional results of [Devetak, Junge, King, Ruskai, CMP 2006]. Beyond tensor products, we establish chain rules for $\tilde H^\uparrow_\alpha$, $\alpha\in[\frac12,1)$, for compositions of possibly signaling channels, extending upon the smooth max-entropy chain rule used in device-independent entanglement certification to the whole interval $\alpha\in[1/2,1)$.
\end{abstract}

\newpage  
\setcounter{tocdepth}{3} 
\tableofcontents  
\newpage 

\section{Introduction} 

Entropy is a central concept of information theory, with applications throughout mathematics, physics and computer science. It quantifies disorder in statistical physics \cite{Boltzmann.1877, Jaynes.1957, LiebYngvason.1999}, uncertainty and compressibility in information theory \cite{Shannon.1948}, and, through von Neumann's quantum extension \cite{vonNeumann.1927, book:vonNeumann.1932}, correlations and entanglement in quantum systems \cite{Horodecki.2009}, where it also captures the structure of ground states of many-body systems and of quantum field theories \cite{Eisert.2010, book:Hollands.2018}. In the form of relative entropies, conditional entropies and mutual informations, it determines the optimal rates of operational tasks such as data compression, communication, hypothesis testing, key distribution and entanglement manipulation \cite{book:Watrous.2018, book:KhatriWilde.2024}. The von Neumann entropy and the Umegaki relative entropy \cite{Umegaki.1962}, however, describe these tasks only in the asymptotic limit of infinitely many independent repetitions. Rényi's one-parameter family of entropies \cite{Renyi.1961} and its quantum generalizations \cite{Petz.1986, Lennert.2013, Wilde.2014} refine this picture to settings with finite resources: they appear in one-shot and finite-block-length bounds, in error and strong converse exponents, and in security proofs with finitely many rounds \cite{Tomamichel.2013, book:Tomamichel.2016, Mosonyi.2011, Mosonyi.2015, Dupuis.2022}, and they have direct operational meaning in single step quantum thermodynamics \cite{Brandao.2015.PNAS}. Understanding how such quantities behave under physical processes and protocols is therefore a basic and relevant structural question.

In classical information theory, Rényi quantities are conveniently described by so called mixed norms.

\subsection{Classical Mixed Norms}

\emph{Mixed norms} measure functions of or sequences in several variables with a different exponent in each variable. For a doubly indexed sequence $a=(a_{xy})$, the $\ell_q[\ell_p]$ norm $\|a\|_{\ell_q[\ell_p]}=\big(\sum_y(\sum_x|a_{xy}|^p)^{q/p}\big)^{1/q}$ first takes the $\ell_p$ norm along $x$ and then the $\ell_q$ norm along $y$. Iterating this procedure yields the nested spaces $\ell_{q_1}[\ell_{q_2}[\ell_{q_3}[\cdots]]]$. Since the definition is pointwise, it applies verbatim to arbitrarily many non-zero exponents: if all exponents lie in $[1,\infty]$ one obtains norms, if all are positive one obtains quasi-norms, and if all lie in $(-\infty,1]\setminus\{0\}$ the resulting functionals are anti-norms. Such spaces were introduced systematically by Benedek and Panzone \cite{Benedek.1961} and have since become a standard tool in harmonic analysis, Banach space theory, and classical information theory.

he Rényi entropy of a distribution is a logarithm of an $\ell_\alpha$ (quasi-)norm \cite{Campbell.1965, Erven.2014}, and many other Rényi quantities are logarithms of mixed (quasi-)norms of joint distributions. Arimoto's conditional entropy, for instance, is an $\ell_1[\ell_\alpha]$ (quasi-)norm of $P_{XY}$ \cite{Arimoto.1977, Fehr.2014},
\begin{align}\label{equ:intro.arimoto}
        H_\alpha(X|Y)_P=\frac{\alpha}{1-\alpha}\log\sum_y\Big(\sum_xP(x,y)^\alpha\Big)^{\frac1\alpha}=\frac{\alpha}{1-\alpha}\log\|P_{YX}\|_{\ell_1[\ell_\alpha]}.
\end{align}
Gallager's exponent function and Ar{\i}kan's guessing moments are of the same $\ell_1[\ell_\alpha]$ form with $\alpha=\frac{1}{1+\rho}<1$ \cite{Gallager.1965, Arikan.1996, Bunte.2014} and Sibson's $\alpha$-mutual information is the mutual information counterpart, given as an $\ell_1[\ell_\alpha]$ quasi-norm, of $\big(P_X(x)^{1/\alpha}P_{Y|X}(y|x)\big)$ \cite{Sibson.1969, Verdu.2015}. Multi-system settings likewise lead to operational quantities with more than two levels nesting \cite{Scarlett.2016}.

\subsection{Non-Commutative Mixed Norms}

Quantum information theory inherits these entropic quantities, for instance through the Petz \cite{Petz.1986} and the sandwiched \cite{Lennert.2013, Wilde.2014} Rényi divergences and the conditional entropies and mutual informations they induce. It does not, however, inherit the convenient coordinate structure that makes mixed norms easy to define. Non-commutativity, the existence of entanglement, and the non-existence of a canonical basis means there is no direct way lift the above mixed norm toolbox to the quantum regime. For instance, the naïve expression $\tr_1\big[\tr_2[|X|^p]^{q/p}\big]^{1/q}$ is not even convex \cite{Carlen.2008} for $p>2$ and hence generally cannot be a norm. 
A systematic answer for indices in $[1,\infty]$ is Pisier's theory of operator-valued Schatten spaces \cite{book:Pisier.1998}, see also \cite{Beigi.2023} for a review in quantum information theory. It provides, among others, two-indexed norms $\|\cdot\|_{(q,p)}$ on $\cB(\cH_1\otimes\cH_2)$ which act locally as Schatten-$q$ and Schatten-$p$ norms, satisfy $\|X\otimes Y\|_{(q,p)}=\|X\|_q\|Y\|_p$, collapse to $\ell_q[\cS_p]$ on block-diagonal operators, and come with Hölder duality, a non-commutative Minkowski inequality \cite[Corollary 1.10]{book:Pisier.1998} and, most importantly, Pisier's factorization formula \cite[Theorem 1.5]{book:Pisier.1998}, which makes them tractable. They are thus non-commutative versions of the $\ell_q[\ell_p]$ spaces for $q,p\geq1$. These norms have found applications in quantum information theory, among others to the additivity of completely bounded minimal output entropies \cite{Devetak.2006} and to complete hypercontractivity and complete (modified) logarithmic Sobolev inequalities \cite{Beigi.2016, Gao.2020, Bardet.2022, Bardet.2024}. They also provide the quantum analogue of \eqref{equ:intro.arimoto}: for $\alpha>1$, the optimized sandwiched conditional R\'enyi entropy and the R\'enyi entropy of a marginal satisfy \cite{Devetak.2006, Beigi.2023}
\begin{align}\label{equ:Intro.opt.cond.entropy}
    \tilde H^\uparrow_\alpha(A|B)_{\rho_{AB}}=\frac{\alpha}{1-\alpha}\log\|\rho_{BA}\|_{(1,\alpha)}
    \qquad\text{and}\qquad
    H_\alpha(A)_{\tr_B[\rho_{AB}]}=\frac{\alpha}{1-\alpha}\log\|\rho_{AB}\|_{(\alpha,1)},
\end{align}
where for a pure state $\rho_{AB}$ the latter is its Rényi entanglement entropy.

Pisier's theory has two limitations. First, being built on operator space theory, it only gives rise to norms, and hence indices $q,p\geq1$. Second, although his construction can in principle be iterated to arbitrarily many indices, the resulting norms are not sufficiently analytically tractable for applications beyond two, or with additional effort three, indices, since Pisier's factorization formula \cite[Theorem 1.5]{book:Pisier.1998} does not suitably apply to these \cite{Fawzi.2026}. The first limitation was recently lifted for two indices in joint work with Fawzi and Rouz\'e \cite{Kochanowski.2026}, where we put Pisier's two-indexed theory on a matrix-analytic foundation. This allows one to construct $\cS_q[\cS_p]$ spaces for all $q,p>0$ satisfying the \emph{compatibility condition} $|\frac1q-\frac1p|\leq1$, which is automatically satisfied for $q,p\geq1$. This condition is sufficient, and essentially necessary, for these objects to behave like the commutative $\ell_q[\ell_p]$ spaces. For the pair $(1,\alpha)$ in \eqref{equ:Intro.opt.cond.entropy} it holds exactly when $\alpha\geq\frac12$, which is precisely the range in which the sandwiched R\'enyi divergence satisfies the data-processing inequality \cite{FrankLieb.2013}. Consequently, \eqref{equ:Intro.opt.cond.entropy} holds for all $\alpha\geq\frac12$ \cite{Kochanowski.2026}, and the extension to quasi-norms reaches an information-theoretically meaningful regime: at $\alpha=\frac12$ the optimized conditional entropy is the max-entropy, which has a direct operational meaning \cite{Koenig.2009} and underlies device-independent entanglement certification \cite{Arnon-Friedman.2019}, and sandwiched R\'enyi divergences of order $\alpha\in[\frac12,1)$ are strong converse exponents \cite{Li.2022}. The construction of \cite{Kochanowski.2026} was, however, restricted to two indices.

Multi-Indexed quantities with more than two indices arise classically, as we saw \cite{Scarlett.2016}, but they are in some sense even more natural in the quantum setting, in particular for quantities defined through quantum channels, such as channel uncertainty relations \cite{Gao.2018.1} and completely bounded minimal output (conditional) entropies \cite{Devetak.2006, Fawzi.2026, Kochanowski.2026}. Every system carries its own index, so that complete boundedness, which is a uniquely quantum phenomenon, and tensor products of channels each add indices to the input and output systems. Accordingly, the proofs in all of these works require intermediate norms with three or four nested indices. In \cite{Devetak.2006}, three-indexed $\cS_q[\cS_p[\cS_q]]$ norms appear, which in their use-case, Pisier's theory can still handle; the four-indexed $\cS_1[\cS_\alpha[\cS_1[\cS_\alpha]]]$ norms with $\alpha>1$ arising in \cite{Fawzi.2026} for completely bounded minimal output conditional entropies already required a generalization of Pisier's factorization formula \cite[Lemma 3.1, Theorem 3.4]{Fawzi.2026}, however, and in \cite{Kochanowski.2026} a three-indexed quasi-norm had to be introduced as an auxiliary object.

\medskip

In this work we lift both restrictions. We allow arbitrarily many indices, as well as negative ones, while retaining the tractability that made Pisier's norms and the two-indexed quasi-norms useful in applications. In this way we do not only extend \cite{book:Pisier.1998, Devetak.2006, Fawzi.2026, Kochanowski.2026}, but lift the classical theory of mixed $\ell_{q_1}[\ell_{q_2}[\ell_{q_3}[\cdots]]]$ spaces to the non-commutative setting, under a compatibility condition on the indices involved. The key is a manifestly multi-indexed factorization formula through a single Schatten quasi- or anti-norm, from which the structural properties follow.

As a consequence, we can extend the multiplicativity tools of \cite{book:Pisier.1998, Devetak.2006} to multi-indexed Schatten spaces and apply them to conditional entropies with $\alpha<1$ under tensor product and more complex signaling channel structures. 

Operationally, all of these statements address the same basic question of whether worse-case bounds on an adversary having access to environments of many rounds of a parallel or sequential strategy are worse then the sum of the single round worst case bounds. Or dually whether entangled inputs to many rounds do provide advantages over product inputs. Additivity and chain rules show that, they cannot, for the quantities considered here and the analysis of a multi-round protocol can be reduced to single-round quantities. Specifically we prove, for all $\alpha\geq\frac12$, additivity of the completely bounded minimal and of the maximal output conditional entropies, as well as chain rules for compositions of signaling channels for $\alpha\in[\frac12,1)$. The latter complement the smooth max-entropy chain rule used in device-independent entanglement certification \cite{Vitanov.2013, Arnon-Friedman.2019}, which corresponds to $\alpha=\frac12$. These applications follow in just a few lines from the tools we develop in this work, but do require a genuine understanding of multi-indexed Schatten spaces which is novel as of this work. We summarize our contributions next. 

As a consequence we cannot just extend all the multiplicativity tools of \cite{Devetak.2006, book:Pisier.1998} to multi-indexed Schatten spaces, but also consider $\alpha<1$ conditional entropies under more complex channel structures. Since the conditional output entropies of channels are $(1,\alpha)$-type mixed norms of channels, tensor products of two such channels involve four indices $(1,\alpha,1,\alpha)$ and sequential, possibly signaling compositions of channels, as they arise e.g. in device-independent settings, require these improved tools, like swapping systems deep inside a multi-indexed norm. In particular do we show additivity of completely bounded minimal output and `regular' maximal output conditional entropies, whose proofs require going beyond positive and two-indexed formulations, as well as derive chain rules for $\alpha\in[\frac{1}{2},1)$-conditional entropies under non-product, possibly signaling setups, which are relevant for device independent entanglement certification. Given the developed toolbox, the entropic consequences  follow in a few lines. See the next subsection for a more detailed summary of the results.

\subsection{Contributions}

We define multi-indexed Schatten spaces via an inductive construction, inspired by \cite[Theorem 1.3]{Kochanowski.2026}. We call a tuple $(\cQ)=(q_1,\cdots, q_n)$ of indices $q_i>0$ or $q_i<0$ \emph{compatible} if $|\frac{1}{q_i}-\frac{1}{q_j}|\leq 1\, \forall i,j$.

\begin{definition*}[Multi-indexed Schatten spaces, \Cref{def:GeneralMultiIndexed}, informal]
Let $\cQ=(q_1,\dots,q_n)$ be a compatible collection of positive indices. For an operator $X$ on $\cH_1\otimes\cdots\otimes\cH_n$ we define $\|X\|_{(\cQ)}$ recursively, by absorbing the first system into the second one: with $\frac1r=\big|\frac1{q_1}-\frac1{q_2}\big|$,
\begin{align}
    \|X\|_{(q_1,q_2,\dots,q_n)}=\begin{cases}
        \inf_{a,b>0}\ \|a\|_{2r}\|b\|_{2r}\,\big\|(a^{-1}\otimes\1)X(b^{-1}\otimes\1)\big\|_{(q_2,q_3,\dots,q_n)} & q_1\leq q_2,\\[2pt]
        \sup_{a,b>0}\ \|a\|^{-1}_{2r}\|b\|^{-1}_{2r}\,\big\|(a\otimes\1)X(b\otimes\1)\big\|_{(q_2,q_3,\dots,q_n)} & q_1\geq q_2,
    \end{cases}
\end{align}
where $a,b$ act on $\cH_1$ only, and on the right-hand side $\cH_1\otimes\cH_2$ is treated as a single system carrying the index $q_2$. The recursion ends with the Schatten (quasi-)norm $\|X\|_{(q)}=\|X\|_q$. For negative indices we use the same formula, restricted to $X>0$ and $a=b^*>0$. 
\end{definition*}

\noindent In words, going from the outside in, each system is weighted by an optimal operator, either divided out (infimum) or multiplied in (supremum), and then merged with its neighbor, until a single Schatten (quasi-)norm remains. 
We will show that the resulting objects are quasi- and anti-norms for $q_i>0$ and $q_i<0$, respectively. They are manifestly multi-indexed extensions of the $n=2$ quasi-norms of \cite{Kochanowski.2026} or the iterated Pisier norms \cite{Fawzi.2026, Beigi.2023, book:Pisier.1998}. For details see \cref{sec:GenMultiIndex}. 

These quasi- or anti-norms should be thought of as `norms' on $n$-partite quantum systems which assign to each system $\cH_i$ an index $q_i$ and measure this system locally in that Schatten-$q_i$-(quasi- or anti-)norm.

\subsubsection{Properties of and Factorization Formulas for Multi-Indexed Schatten Spaces}

In order to make these quasi-norms tractable we require a manifestly multi-indexed factorization formula, which we supply with \Cref{thm:Multi.Thm.Pisier.like}. It is the technical main result of this work and a factorization formula from which we will be able to derive almost all structural properties of these multi-indexed Schatten quasi-norms and anti-norms. It can be thought of as a multi-indexed Hölder inequality; effectively we prove that the multi-indexed Schatten quasi-norm $\|\cdot\|_{(q_1,\dots, q_n)}$ can be factored through the Schatten-$s$-quasi-norm, while the weights carry the indices $s_i$, $\frac{1}{s_i}=\frac{1}{q_i}-\frac{1}{s}$.\footnote{Compare to the regular Hölder inequality, which relates the $q$-quasi-norm to the $s$ and $r, \frac{1}{r}=\frac{1}{q}-\frac{1}{s}$ quasi-norm for $s\geq q$.} 

\begin{theorem*}[Factorization through Schatten spaces, \Cref{thm:Multi.Thm.Pisier.like} and \Cref{cor:multi.Pisier.like.negative}, informal]
Let $\cQ=(q_1,\dots,q_n)$ be a compatible collection of indices, either all positive or all negative, and let $\underline{q}\leq q_i \leq \overline{q}$ $\forall i$ be indices of the same sign, where in addition $\underline{q}\leq1$ in the positive and $\overline{q}\geq-1$ in the negative case. Then
\begin{align}
    \|X\|_{(\cQ)}=\inf_{\underset{}{A,B}}\|AA^*\|^\frac{1}{2}_{(\cS)}\|B^*B\|^\frac{1}{2}_{(\cS)}\|A^{-1}XB^{-1}\|_{\overline{q}}=\sup_{\underset{}{A,B}}\|A^*A\|^{-\frac{1}{2}}_{(\cT)}\|BB^*\|^{-\frac{1}{2}}_{(\cT)}\|AXB\|_{\underline{q}} 
\end{align}
where $A,B$ act on all $n$ systems and $\cS$, $\cT$ are the collections given by $\frac{1}{s_i}=\frac{1}{q_i}-\frac{1}{\overline{q}}$ and $\frac{1}{t_i}=\frac{1}{\underline{q}}-\frac{1}{q_i}$. For negative indices this applies only to $X>0$ with $A=B^*>0$.
Moreover, for self-adjoint $X$ one may take $A=B^*\geq0$. 
\end{theorem*}

\begin{remark*}
    For a single index $\cQ=(q)$ these formulas reduce to the well-known variational formulas for Schatten (quasi-)norms, \Cref{prop:SpVariationalFormulas}, and in the two-indexed case they recover \cite[Lemma 1.1]{Kochanowski.2026} for $s=1$ and \cite[Lemma 3.4]{Kochanowski.2026} for $t=\min\{1,q,p\}$. Moreover, self-adjointness $X=X^*$ is sufficient for symmetric optimizers $A=B^*$, in particular also in the supremum formula, where previously $X\geq0$ was required \cite{Devetak.2006, Kochanowski.2026}.
\end{remark*}

The proof uses the matrix analytic tools established by, among others, the author in \cite{Kochanowski.2026} together with some operator space prerequisites and a strengthening of \cite[Lemma 3.1]{Fawzi.2026} in \cref{thm:VariationalSpOperatorspace}. On the way, in \Cref{sec:Multi.Pisier}, we also prove a manifestly multi-indexed extension of Pisier's seminal formula \cite[Theorem 1.5]{book:Pisier.1998} with \Cref{thm:main.multi.pisier}, which informally states that 
\begin{align}\label{equ:multi.Pisier.intro}
    \|X\|_{(\cQ;\cX)} =\inf_{X=(A\otimes\1)Y(B\otimes\1)}\|AA^*\|_{(\cQ)}^{\frac{1}{2}}\|B^*B\|_{(\cQ)}^{\frac{1}{2}}\|Y\|_{(\infty,\dots,\infty;\cX)}. \tag{\cref{thm:main.multi.pisier}}
\end{align} for any operator space $\cX$ and normed multi-indexed Schatten space $\cQ=(q_1,\dots, q_n)$ with $q_i\in[1,\infty]$, stemming from iteration of Pisier's construction \cite{book:Pisier.1998}. Even though Pisier's original construction is iterable \cite{book:Pisier.1998, Beigi.2023} and thus allows for definition of multi-indexed Schatten quasi-norms, these are not analytically very tractable beyond two indices without the above mentioned \cref{thm:main.multi.pisier}. In particular, this generalizes the two-indexed formula \cite[Theorem 3.4]{Fawzi.2026}, which was used to prove the additivity of completely bounded minimal output conditional entropies and, through it, improved security proofs for time-dependent quantum key distribution.

One of the consequences of \Cref{thm:Multi.Thm.Pisier.like} is a multi-indexed version of the relational consistency result of \cite[Theorem 1.3]{Kochanowski.2026} with \Cref{thm:multi.quasi.relational.consistency}, which relates the $(q,\cQ)$-quasi-norm with the $(p,\cQ)$-quasi-norm in the same manner as its two-indexed \cite[Theorem 1.3]{Kochanowski.2026} and operator space \cite[Lemma 3.1]{Fawzi.2026} versions. Informally for compatible $q\leq p$ it states
\begin{align}
        \|X\|_{(q,\cQ)} &= \inf_{\underset{\Pi_a=L_X,\, \Pi_b=R_X}{a,b\geq0}}\|a\|_{2r}\|b\|_{2r}\|a_1^{-1}Xb_1^{-1}\|_{(p,\cQ)}, \tag{\cref{thm:multi.quasi.relational.consistency}} \\
         \|X\|_{(p,\cQ)} &= \sup_{\underset{\Pi_a=L_X,\, \Pi_b=R_X}{a,b\geq0}}\|a\|^{-1}_{2r}\|b\|^{-1}_{2r}\|a_1Xb_1\|_{(q,\cQ)}, 
\end{align} where $\frac{1}{r}=\frac{1}{q}-\frac{1}{p}$.
From these formulas we prove in \cref{cor:multi.block.collapse} the expected collapse on block-diagonal operators
\begin{align}
    \|\oplus_iX_i\|_{(q,\cQ)}= \Big\|\sum_i|i\rangle\langle i|\otimes X_i\Big\|_{(q,\cQ)}= \left(\sum_i\|X_i\|^q_{(\cQ)}\right)^\frac{1}{q}, \tag{\cref{cor:multi.block.collapse}}
\end{align}
which extends \cite[Theorem 1.4]{Kochanowski.2026} to many indices and \cite[Corollary 1.3]{book:Pisier.1998} to quasi- and anti-norms; the case $(\cQ)=(1,\alpha)$ with $\alpha\geq\frac{1}{2}$ appears in \cite[Proposition 5.4]{book:Tomamichel.2016}. We also prove multiplicativity over product operators in \cref{cor.multi.tensor.multiplicativity}
\begin{align}
    \|X_Q\otimes Y_P\|_{(\cQ, \cP)} = \|X_Q\|_{(\cQ)}\|Y_P\|_{(\cP)} \tag{\cref{cor.multi.tensor.multiplicativity}}
\end{align} along arbitrary ordered splittings. This is the expected property; the point is that it holds for arbitrary cuts compatible with the ordering of the indices.

We prove, as a consequence of \Cref{thm:Multi.Thm.Pisier.like}, that for $q_i>0$ our construction gives rise to $\min\{1,\min_i\{q_i\}\}$-normable quasi-norms, \Cref{thm:Multi.Quasi.Norms}, and in particular that they satisfy the triangle inequality when $q_i\geq 1,\ \forall i$. They also satisfy a multi-indexed Hölder duality when $q_i\geq 1$ as per \Cref{prop:MultiHölderDual} and we further show that when $0< q_i\leq 1$ are compatible indices, then on positive operators the associated quasi-norms satisfy the reverse triangle (Minkowski) inequality in \Cref{thm:Multi.Rev.Triangle} and the reverse Hölder duality in \Cref{cor:reverse.multi.hölder}. All of these are the natural multi-indexed generalizations of the corresponding properties of the two-indexed Schatten quasi-norms. These properties can be viewed as natural extensions of \cite{Kochanowski.2026} to the multi-indexed setting and of Pisier's theory \cite{book:Pisier.1998} to the quasi- and anti-norm setting.

\noindent Thus we have practically extended all of the classical mixed-norm framework to the quantum setting, under a compatibility assumption on the involved indices, and effectively established a complete theory of multi-indexed Schatten spaces.

\subsubsection{Some Properties for Practitioners}

Some further properties of these quasi-norms that we prove, and which may be helpful for quantum information practitioners, are the following.

\begin{proposition*}[Useful properties of multi-indexed quasi-norms]
Assume that $\cQ=(q_1,\dots, q_n)$, $\cP=(p_1,\dots, p_m)$, and all multi-indexed Schatten spaces below are compatible. Then it holds that
\begin{itemize}
    \item \Cref{cor:partial.trace.1}: A $1$ as the last index acts as a partial trace: For any positive $X\geq0$,
    \begin{align}
        \|X\|_{(q_1,\dots, q_n,1)}= \|\tr_{n+1}[X]\|_{(q_1,\dots, q_n)}
    \end{align}
    \item \Cref{thm:QuasiNormMinkowskiInequality}: Swapping adjacent decreasing indices is non-contractive: Let $a\geq b$; then\footnote{Here we associate to each subsystem $Q_i,A,B,P_j$ an index $q_i,a,b,p_j$ and the order is dictated by the order in the quasi-norm. For a detailed explanation of this notation see \Cref{sec:notation}.}
    \begin{align}
        \|X\|_{(Q:\cQ,\, A:a,\,B:b,\,P:\cP)} \leq \|\SWAP_{A\leftrightarrow B}(X)\|_{(Q:\cQ,\, B:b,\,A:a,\,P:\cP)}
    \end{align}
    \item \Cref{cor.multi.tensor.multiplicativity}: Multi-indexed quasi-norms are multiplicative under tensor products of operators along ordered cuts: Let $X,Y$ live on, respectively, $Q=Q_1\dots Q_n$ and $P=P_1\dots P_m$; then
    \begin{align}
        \|X\otimes Y\|_{(Q:\cQ,P:\cP)}= \|X\|_{(Q:\cQ)} \|Y\|_{(P:\cP)},
    \end{align}
    \item \Cref{lem:basic.4.index.multiplicativity}: Mixed tensor products also split:
    Let $X_{Q_1P_1}, Y_{Q_2P_2}$ be operators and $q,p>0$ compatible indices; then it holds that
    \begin{align}
        \|X_{Q_1P_1}\otimes Y_{Q_2P_2}\|_{(Q_1:q,Q_2:q,P_1:p,P_2:p)} = \|X\|_{(Q_1:q,P_1:p)}\|Y\|_{(Q_2:q,P_2:p)}.
    \end{align}
    \item \Cref{prop:negative.multi.indexed}: For any $X>0$ it holds that 
    \begin{align}
    \|X\|_{(q_1,\dots, q_n)} = \|X^{-1}\|^{-1}_{(-q_1,\dots, -q_n)}.
    \end{align}
\end{itemize}
\end{proposition*}

\begin{remark*}
    \Cref{cor:partial.trace.1} generalizes \cite[Corollary 3.2]{Kochanowski.2026} to multi-indexed spaces, the swap inequality is the non-commutative Minkowski inequality discussed below, and the special case $q=1$, $p=\alpha\in[\frac{1}{2},\infty]$ of \Cref{lem:basic.4.index.multiplicativity} is the additivity of the optimized sandwiched Rényi conditional entropy under tensor-product states \cite[Corollary 5.9]{book:Tomamichel.2016}. As a further direct consequence, the sandwiched umlaut information of \cite{Kochanowski.2026} is additive under tensor-product states, \Cref{prop.add.sand.umlaut}. For its Petz-Rényi, Umegaki, and Kullback--Leibler versions see \cite{Girardi.2025, Girardi.2025.2}.
\end{remark*}

\subsubsection{Linear Maps and Tools for Multiplicativity}

In \Cref{sec:LinerMaps} we consider the induced quasi-norms and co-quasi-norms on the space of linear maps between multi-indexed Schatten spaces, indexed by some $(\cQ)=(q_1,\dots, q_n)$ and $(\cP)=(p_1,\dots, p_m)$. That is, given a linear map $\Phi$ we consider
\begin{align}
    \|\Phi\|^{(+)}_{(Q:\cQ)\to (P:\cP)}&:=\sup_{\underset{(X\geq0)}{X\in{\cS_{\cQ}}}}\frac{\|\Phi(X)\|_{(P:\cP)}}{\|X\|_{(Q:\cQ)}}, \\
\|\Phi\|^+_{\co,(Q:\cQ)\to (P:\cP)}&:=\inf_{\underset{X\geq0}{X\in\mathcal{S}_{\cQ}}}\frac{\|\Phi(X)\|_{(P:\cP)}}{\|X\|_{(Q:\cQ)}},
\end{align} and their completely bounded counterparts
\begin{align}
    \|\Phi\|^{(+)}_{\cb,(Q:\cQ)\to (P:\cP)}&:=\sup_E\|\id_E\otimes\Phi\|^{(+)}_{(E:t,Q:\cQ)\to (E:t,P:\cP)}\,, \tag{\cref{def:cb.quasi.norm}} \\
    \|\Phi\|^+_{\cb,\co,(Q:\cQ)\to (P:\cP)}&:=\inf_E\|\id_E\otimes\Phi\|^+_{\co,(E:t,Q:\cQ)\to (E:t,P:\cP)}\,. \tag{\cref{def:cb.co.quasi.norm}}
\end{align}
For these we show in \Cref{lem:cb.norm.via.t} and \Cref{lem:cb.co.norm.via.t} that just as in Pisier's theory, the index placed on the completely bounding system, $E$ above, is irrelevant, as long as it is compatible with the input and output indices.

Importantly, this allows us to define completely bounded (co-)quasi-norms between multi-indexed quasi-normed spaces, extending the completely bounded (co-)quasi-norms between single-indexed Schatten spaces introduced in \cite[Section 4.2]{Kochanowski.2026}, see also \cite{Li.2026}. 

The rest of \cref{sec:LinerMaps} is focused on deriving multiplicativity statements for such mixed norms under compositions and particularly tensor products of (CP) linear maps. To this end we require two important tools.

First we prove with \cref{thm:QuasiNormMinkowskiInequality} a statement about the contractivity of swapping adjacent systems in multi-indexed Schatten quasi- (and, by extension, anti-)norms.
\begin{theorem*}[The non-commutative Minkowski inequality, \cref{thm:QuasiNormMinkowskiInequality}, informal]
For any compatible $\cE,p,q,\cR$ with $q\leq p$ it holds for any $X$ that
\begin{align}
    \|X_{EQPR}\|_{(E:\cE,Q:q,P:p,R:\cR)} \geq \|X_{EPQR}\|_{(E:\cE,P:p,Q:q,R:\cR)} 
\end{align} 
\end{theorem*}
Such an inequality is also known as a \emph{non-commutative Minkowski inequality} and the above extends \cite[Corollary 1.10]{book:Pisier.1998}, a seminal result in Pisier's theory, to arbitrary compatible indices. 
In particular, it allows for inner indices $\cR$, which, to our knowledge, had not been considered even in the normed, operator space setting. Equivalently, the $\SWAP$ of two adjacent systems is completely contractive from $(Q:q,P:p,R:\cR)$ to $(P:p,Q:q,R:\cR)$, and in the reverse direction completely expansive for the corresponding co-quasi-norm.
Just as in the normed case \cite{Devetak.2006, Fawzi.2026}, this inequality is the key tool behind all following multiplicativity and additivity results of this work, and the ability to swap systems deep inside a multi-indexed norm, i.e.\ with non-trivial inner systems $R$, is crucial for more advanced applications.

The second type of tool consists of two ``chain-rule'' type equalities, developed in \Cref{sec:3.removing.indices}, concerning the removal of identity channels in mixed (quasi-)norms. The first instance of an equality of this type was \cite[Lemma 5]{Devetak.2006}, which the authors dubbed ``an important Lemma'', and which was instrumental in the proof of the multiplicativity of $\cb,q\to p$ norms, finding applications in \cite{Devetak.2006, Beigi.2016, Bardet.2022, Beigi.2023, JungePalazuelos.2015} as well as the additivity of the completely bounded minimal output entropy \cite{Jencova.2006}. 

\begin{theorem*}[Identity removal, \Cref{lem:reductionLemma} and \Cref{thm:id.right.quasi}, informal]
Let $\Phi:Q\to P$ be a completely positive map and $\cQ,\cP,\cT$ multi-indexed Schatten spaces on systems $Q,P,T$. If all indices lie in $[1,\infty]$, then
\begin{align}
    \|\Phi\otimes\id_T\|_{(Q:\cQ,T:\cT)\to(P:\cP,T:\cT)}=\|\Phi\|_{(Q:\cQ)\to(P:\cP)}, 
\end{align}
and this remains true with $\cT$ replaced by an arbitrary operator space $\cX$. If all indices lie in $[\frac12,1]$, then the analogous statement holds for co-quasi-norms,
\begin{align}
    \|\Phi\otimes\id_T\|^+_{\co,(Q:\cQ,T:\cT)\to(P:\cP,T:\cT)}=\|\Phi\|^+_{\co,(Q:\cQ)\to(P:\cP)}. 
\end{align}
\end{theorem*}

\noindent In words: an identity channel acting on the innermost systems can be removed without changing the (co-)norm. The normed statement generalizes \cite[Lemma 5]{Devetak.2006} and \cite[Theorem 4.1]{Fawzi.2026}, which were instrumental there in proving the additivity of minimal output entropies, respectively conditional entropies, to normed multi-indexed input and output spaces. The requirement that the inner space $\cT$ is normed is necessary, since otherwise the equality fails by \cite[Counterexample 4.1]{Kochanowski.2026}. 
The co-quasi-normed statement extends the simple \cite[Lemma 4.8]{Kochanowski.2026}, where the inner index is $1$, to inner multi-indexed spaces with indices in $[\frac12,1]$, based on a very different proof, which effectively relates the latter to the former via reverse Hölder duality. Entropically, both equalities can be interpreted as chain rules, and together with the non-commutative Minkowski inequality they are the main tools behind the multiplicativity and additivity results of this work.

As a consequence of these proofs, which are based on the construction of \cite{Devetak.2006}, we also show in \cref{cor:mixed.positivity.normed} that CP maps achieve their mixed norms between arbitrary multi-indexed Schatten spaces $\cQ,\cP$ on positive operators, i.e.
\begin{align}
        \|\Phi\|_{(Q:\cQ)\to (P:\cP)}=\|\Phi\|^+_{(Q:\cQ)\to (P:\cP)}. \tag{\cref{cor:mixed.positivity.normed}}
    \end{align}
This generalizes the corresponding results for mixed Schatten norms ($q\to p$, $q,p\geq1$) from \cite{Watrous.2004} and completely bounded mixed Schatten norm $(\cb,q\to p)$, $q,p\geq1$ from \cite{Devetak.2006}, as well as $\cb,q\to (r,s)$ norms, $q,r,s\geq1$ of \cite[Lemma A.2]{Fawzi.2026} to arbitrary numbers of positive indices. It is thus the most general statement of this kind, since for negatively indexed input spaces the quantity is trivially $0$ without the positivity restriction.

\subsubsection{Applications in Quantum Information Theory}

Concerning complete (reverse) hypercontractivity, two concepts relevant in the study of dissipative Markovian open system dynamics, see e.g.\ \cite{Beigi.2016, Bardet.2022}, we show in \Cref{sec:Mult.cb.reverse.hyper}  that the $\cb,q\to p$ quasi-norm is multiplicative under tensor products of linear maps for all compatible $q\leq p$, \Cref{thm:mixed.up.cb.additivity}, and the $\cb,\co,q\to p$ co-quasi-norm for all compatible $q\geq p$, \Cref{thm:cb.reverse.hypercontractive.bound}, both including negative indices. The former extends \cite[Theorem 11]{Devetak.2006}, which establishes the corresponding statement for $1\leq q\leq p$ and CP maps, and the latter extends \cite[Theorem 1.6]{Kochanowski.2026} which had only proven super-multiplicativity for $0<p\leq q$. These multiplicativity results are relevant in the context of dissipative open system dynamics where they imply tensorization of the notions of complete hypercontractivity, respectively complete reverse hypercontractivity. 

We do this by establishing a simple multiplicativity result for states in form of the previously mentioned \cref{lem:basic.4.index.multiplicativity}, which directly yields the achievability parts of these statements and gives a complete proof of the achievability claimed in \cite[Theorem 11]{Devetak.2006}. This also proves the additivity of the sandwiched umlaut information stated in \cite[Proposition 4.3]{Kochanowski.2026}.

In \cref{sec:Renyi.coherent.info} we first consider the $\alpha$-Rényi coherent information of a quantum channel \cite{Gao.2018.2, Beigi.2023, Berta.2026.1, Li.2026} which is effectively the negative of a minimal output conditional entropy, and show that its completely bounded version is additive under tensor products of channels, see \cref{thm:cb.coherent.add}.

Then in \Cref{sec:Additivity.Extremal.Output.Entropies} we prove, for all $\alpha\in[\frac{1}{2},1)\cup(1,\infty]$, additivity under tensor products of quantum channels of the completely bounded minimal output conditional sandwiched Rényi entropy, i.e.
\begin{align}
    \inf_{E,\rho_{EQ_1Q_2}}\tilde{H}^\uparrow_p(A_1A_2|B_1B_2E)_{(\id_E\otimes\Phi\otimes\Psi)(\rho_{EQ_1Q_2})} &= \tag{\cref{thm:cb.min.cond.entropy.add}}\inf_{E,\rho_{EQ_1}}\tilde{H}^\uparrow_p(A_1|B_1E)_{(\id_E\otimes\Phi)(\rho_{EQ_1})}\\&\qquad+\inf_{E,\rho_{EQ_2}}\tilde{H}^\uparrow_p(A_2|B_2E)_{(\id_E\otimes\Psi)(\rho_{EQ_2})}, 
\end{align} and of the maximal output conditional sandwiched Rényi entropy, i.e.
\begin{align}
    \sup_{\rho_{Q_1Q_2}}\tilde{H}^\uparrow_p(A_1A_2|B_1B_2)_{(\Phi\otimes\Psi)(\rho_{Q_1Q_2})} &= \tag{\cref{thm:max.cond.entropy.add}}\sup_{\rho_{Q_1}}\tilde{H}^\uparrow_p(A_1|B_1)_{\Phi(\rho_{Q_1})}\\&\qquad +\sup_{\rho_{Q_2}}\tilde H^\uparrow_p(A_2|B_2)_{\Psi(\rho_{Q_2})}. 
\end{align}
In the mixed-norm picture these are multiplicativity statements for $(\cb),1\to(1,\alpha)$ mixed norms and $(\cb),\co,1\to(1,\alpha)$ mixed co-quasi-norms, which follow in a few lines with the previously derived tools. For $\alpha>1$ the former was established in \cite{Fawzi.2026}, where it was applied to security proofs of quantum key distribution protocols with time-varying channels; we extend it to $\alpha\in[\frac{1}{2},1)$. The latter shows that, unlike the minimal one, the maximal output conditional entropy is additive without complete boundedness and coincides with its completely bounded version; to the best of our knowledge this has not appeared explicitly in the literature before. For trivial conditioning systems both statements reduce to the known additivity of the completely bounded minimal \cite{Devetak.2006, Kochanowski.2026, Li.2026} and of the maximal \cite[Corollary 1.2]{Kochanowski.2026} output Rényi entropies.

Finally, we turn to a setting that is not covered by tensor products. In
\Cref{sec:DI.ent.cert} we consider compositions of maps in which each map may act directly on outputs of the previous one. This describes in particular device-independent protocols in which one cannot assume a certain channel structure and so has to assume that later rounds may depend on earlier outcomes, see \cref{fig:sequential.intro}. We prove in \Cref{thm:signaling.sup.chain.rule} and \ref{thm:signaling.inf.chain.rule} chain rules bounding the conditional entropy of the composition of two such maps by the sum of the single-round quantities. Specifically we show the following chain rule and interpret its operational significance in the task of device-independent entanglement certification. 
\begin{theorem*}[Signaling chain rule, \cref{thm:signaling.sup.chain.rule}, informal]
  For quantum systems $Q, Q_1, Q_2, A_1, A_2,$ $B_1, B_2$ and quantum channels $\Phi:Q\to A_1B_1Q_1, \Psi:Q_1\to A_2B_2$, we have for any $\alpha\in[\frac{1}{2},1]$ that  
  \begin{align}
       \sup_{\rho\in\cD(\cH_Q)} \tilde H_\alpha^\uparrow(A_1A_2|B_1B_2)_{(\id_{A_1B_1}\otimes\Psi)\circ\Phi(\rho)} &\leq  \sup_{\rho\in\cD(\cH_Q)} \tilde H_\alpha^\uparrow(A_1|B_1)_{\tr_{Q_1}[\Phi(\rho)]} \notag  \\& \quad+\sup_{\rho\in\cD(\cH_{Q_1})} \tilde H_\alpha^\uparrow(A_2|B_2)_{\Psi(\rho)}.
  \end{align}
\end{theorem*}

\noindent\begin{minipage}[c]{0.49\textwidth}
This \Cref{thm:signaling.sup.chain.rule} generalizes the smoothed max-entropy chain rule \cite{Vitanov.2013} used in \cite{Arnon-Friedman.2019}, which effectively corresponds to the single point $\alpha=\frac12$, to the whole interval $\alpha\in[\frac12,1)$, so that the resulting bounds for DI entanglement distillation certification can in principle be optimized over $\alpha$. This is morally similar to the improved QKD rates which resulted from \cite{Fawzi.2026}'s $\alpha\in[1,\infty]$ chain rule, compared to the smoothed-min entropy version. To the best of our knowledge these are the first chain rules in a signaling setting for optimized conditional sandwiched Rényi entropies with $\alpha<1$ \cite{Dupuis.2015, Marwah.2025, ArqandHahnTan.2025}. 
\end{minipage}
\begin{minipage}[c]{0.49\textwidth}
\centering
\begin{tikzpicture}[x=1cm, y=1cm, >={Stealth[length=1.6mm]}, font=\footnotesize,
  box/.style={draw, rounded corners=1pt, fill=black!5, minimum width=6.5mm, minimum height=6.5mm, inner sep=1pt}]
\node[box] (P1) at (0.9,0) {$\Phi_1$};
\node[box] (P2) at (2.3,0) {$\Phi_2$};
\node (d) at (3.45,0) {$\cdots$};
\node[box] (Pn) at (4.6,0) {$\Phi_n$};
\draw[->] (0,0) -- (P1.west) node[midway, above] {$Q$};
\draw[->] (P1.east) -- (P2.west) node[midway, above] {$Q_1$};
\draw[->] (P2.east) -- (3.1,0);
\draw[->] (3.8,0) -- (Pn.west);
\draw[->] (Pn.east) -- (5.5,0) node[midway, above] {$Q_n$};
\foreach \n/\i in {P1/1, P2/2, Pn/n} {
  \draw[->] (\n.north) -- ++(0,0.5) node[above, inner sep=1pt] {$A_\i$};
  \draw[->] (\n.south) -- ++(0,-0.5) node[below, inner sep=1pt] {$B_\i$};
}
\end{tikzpicture}

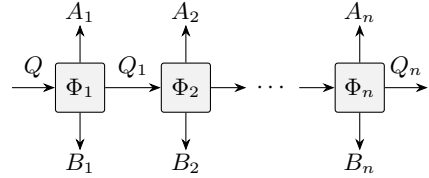
\captionof{figure}{\footnotesize A simplified view of the DI entanglement certification structure of \cite{Arnon-Friedman.2019}. In round $i$ the channel $\Phi_i:Q_{i-1}\to A_iB_iQ_i$ either stores the possibly entangled pair $A_iB_i$ or measures it, and passes the outcomes and remaining register in $Q_i$ on. 
The certified quantity is the maximal output conditional entropy $\tilde H^\uparrow_\alpha(A_1\cdots A_n\mid B_1\cdots B_n)$, whose negativity witnesses entanglement. Iterating \Cref{thm:signaling.sup.chain.rule} bounds it by the sum of the $n$ single-round quantities, for any $\alpha\in[\tfrac12,1)$.}
\label{fig:sequential.intro}
\end{minipage}

\medskip

\Cref{thm:signaling.inf.chain.rule} is a chain rule analogous to \cref{thm:signaling.sup.chain.rule} for minimal output conditional entropies and extends the completely bounded $\alpha>1$ version established in \cite[Theorem 4.13]{Fawzi.2026} to $\alpha\in[\frac{1}{2},1)$.

\subsection{Layout of the paper}

The remainder of this introduction fixes our notation in \Cref{sec:notation} and recalls Pisier's operator-valued Schatten spaces in \Cref{Prelim:Pisier}.

\Cref{sec:MultiIndex} develops the theory of multi-indexed Schatten spaces. In \Cref{sec:Multi.Pisier} we first extend Pisier's factorization formula to iterated operator-valued Schatten norms with indices in $[1,\infty]$, \Cref{thm:main.multi.pisier}.
In \Cref{sec:GenMultiIndex} with \Cref{def:GeneralMultiIndexed} we then define multi-indexed Schatten quasi- and anti-norms for arbitrary compatible indices and prove the key factorization formula of \Cref{thm:Multi.Thm.Pisier.like}. From these we derive the structural properties: relational consistency, multiplicativity under tensor products, (reverse) Hölder duality, the collapse on block-diagonal operators, quasi-normability, and the (reverse) triangle inequality, see \cref{sub:structural} and \cref{sub:quasi.anti}.

\Cref{sec:LinerMaps} then considers linear maps between multi-indexed Schatten spaces. After discussing completely bounded quasi-norms and co-quasi-norms between multi-indexed spaces in \Cref{Lin.1}, we prove the non-commutative Minkowski inequality in \Cref{Lin.2} and the two identity-removal equalities, one normed and one co-quasi-normed, in \Cref{sec:3.removing.indices}. Combining these tools, we establish in \Cref{sec:Mixed.Multiplicativity.Statements} the multiplicativity of complete hypercontractivity and complete reverse hypercontractivity for arbitrary ordered indices under tensor products and the additivity of the sandwiched umlaut information under tensor-product states.

\Cref{sec:ApplictoQI} contains the applications to quantum information theory related to entropies. After discussing the relevance of additivity statements and entropic chain rules, we recap in \Cref{sec:Entopic.Norms} the expressions of optimized sandwiched Rényi conditional entropies, and the extremal output entropies of quantum channels derived from them, as (mixed) multi-indexed Schatten quasi-norms. We then prove additivity of the completely bounded minimal and the maximal output conditional entropies in \Cref{sec:Additivity.Extremal.Output.Entropies} and additivity of the completely bounded Rényi coherent information in \Cref{sec:Renyi.coherent.info}. Finally, in \Cref{sec:DI.ent.cert} we go beyond tensor products and prove chain rules for compositions of signaling channels, with applications to device-independent entanglement certification.

We conclude in \Cref{sec:out.multi.indexed} with open questions and directions. \Cref{app:matrix.tools} collects the matrix-analytic tools, including Sion's minimax theorem and the single-indexed factorization formulas for Schatten spaces, require for the proofs in \cref{sec:MultiIndex}. \Cref{App:OpSpaces} recalls the required background on operator spaces, and \Cref{app:Generalized.Variational.Expressions} contains the proofs of the factorization formulas for the multi-indexed version of Pisier's operator-valued Schatten norms in \cref{sec:Multi.Pisier}, including that of \Cref{thm:main.multi.pisier}.

\subsection{Notation}\label{sec:notation}
Most of our notation is standard; here we collect the conventions specific to this work.

\medskip

For a vector $v=(v_i)_{i\in I}\in\C^{I}$ indexed by a finite set $I$ we write its $\ell_p$-quasi-norms, for any $p\neq0$, as
\begin{align}
    \|v\|_{\ell_p}:=\Big(\sum_{i\in I}|v_i|^p\Big)^\frac{1}{p}, \quad \|v\|_{\ell_\infty}:=\max_{i\in I}|v_i|, \quad \|v\|_{\ell_{-\infty}}:=\min_{i\in I}|v_i|
\end{align} and denote the associated sequence space by $\ell_p(I):=(\C^{I},\|\cdot\|_{\ell_p})$. Whenever $p<0$ we assume the sequences to contain only strictly positive elements. These are norms for $p\in[1,\infty]$ and $p$-normable quasi-norms for $p\in(0,1]$ and anti-norms for $p\in [-\infty,0)\cup (0,1]$.

\noindent Two indices $p,p^\prime\in\R\setminus\{0\}$ (or $\pm\infty$) are called \emph{Hölder dual} (or Hölder conjugate) if $\frac{1}{p}+\frac{1}{p^\prime}=1$, and we consistently write $p^\prime$ for the Hölder dual of $p$. 

\noindent All Hilbert spaces $\cH,\cK$ are finite dimensional unless stated otherwise, and $\cB(\cH,\cK)$ denotes the linear operators $\cH\to\cK$, $\cB(\cH):=\cB(\cH,\cH)$, with identity $\1$. The canonical (full) trace on $\cB(\cH)$ is denoted as $\Tr[\cdot]$. We write $\langle X,Y\rangle=\Tr[X^*Y]$ for the Hilbert--Schmidt inner product and $|X|:=\sqrt{X^*X}$, so that polar decompositions read $X=U|X|=|X^*|V$ with unitaries $U,V$. We denote the partial trace $\tr_2[\cdot]:\cB(\cH_1\otimes\cH_2)\to \cB(\cH_1)$ and the reduced operator $\tr_2[X_{12}]$ is at times denoted by $X_1$.  We write $X\geq0$ for positive semidefinite and $X>0$ for positive definite operators, and $\cB^+(\cH)$ for the positive semidefinite cone.

\medskip

For $0<\kappa\leq1$, a \emph{$\kappa$-normable quasi-norm} on a complex vector space $\mathcal V$ is a map $\|\cdot\|:\mathcal V\to[0,\infty)$ which is positively homogeneous, $\|\lambda v\|=|\lambda|\|v\|$, positive definite, and satisfies $\|v+w\|^\kappa\leq\|v\|^\kappa+\|w\|^\kappa$. It then satisfies $\|v+w\|\leq2^{\frac1\kappa-1}(\|v\|+\|w\|)$ \cite{book:Kalton.03}, and a norm is a $1$-normable quasi-norm. By the Aoki--Rolewicz theorem \cite{Aoki.1942, Rolewicz.1957} every quasi-norm is equivalent to a $\kappa$-normable one for some $\kappa\in(0,1]$. 

\noindent Dually, following \cite{Bourin.Hiai.2011, Bourin.2014}, we call a positively homogeneous, positive definite, functional on a cone of strictly positive operators that satisfies the reverse triangle inequality $\|X+Y\|\geq\|X\|+\|Y\|$ an \emph{anti-norm}.

\noindent The Schatten-$p$ functional for any $p\notin\{-\infty,0,\infty\}$
\begin{align}\label{equ:def.Schatten}
    \|X\|_p:=\Tr[|X|^p]^\frac{1}{p}, \quad |X|:=\sqrt{X^*X}\geq0.
\end{align} For $p\in\{-\infty,\infty\}$ we set $\|X\|_\infty$ to be the operator norm and $\|X\|_{-\infty}:=\|X^{-1}\|_\infty^{-1}$.
When $p>0$ these are $\min\{1,p\}$-normable quasi-norms and hence norms for $p\geq1$ \cite{McCarthy.1967, Rotfeld.1969, Simon.2005}. We write $\cS_p(\cH)=(\cB(\cH),\|\cdot\|_p)$
for the associated quasi-normed \emph{Schatten spaces} and denote its positive cone as
\begin{align}
    \cS_p^+(\cH):
=\{X>0\,|\, \|X\|_p<\infty\},\quad \cS^+_\infty(\cH):=\{X\in\cB(\cH)\,|\, X>0 \text{ compact}\}.
\end{align}
When $p<0$, these functionals are only defined on the cone of strictly positive operators $X>0$ and we have the identity
\begin{align}
    \|X\|_p= \|X^{-1}\|^{-1}_{-p}.
\end{align}

\noindent Since for $p\in(-\infty,1]\setminus\{0\}$ and $X,Y\geq0$ ($X,Y>0$ if $p<0$) the \emph{reverse Minkowski (triangle) inequality} 
\begin{align}
    \|X+Y\|_p\geq\|X\|_p+\|Y\|_p
\end{align}
holds, see e.g. \cite[Lemma 6]{Cubitt.2015}, they are anti-norms for such $0\neq p\leq 1$ on $\cS_p^+(\cH)$. The reverse Hölder duality
\begin{align}\label{equ:reverse.Hölder.Schatten}
    \|X\|_p=\inf_{Y>0}\Tr[YX]\,\|Y\|_{p^\prime}^{-1}
\end{align}
holds for $p\in(-\infty,1)\setminus\{0\}$ and $X\geq0$ ($X>0$ if $p<0$), where $p^\prime$ is the Hölder dual index of $p$ defined above. The \emph{generalized Hölder inequality} \cite{borell1982positivity, tomamichel2014relating} states that for $p,q,r>0$ with $\frac1p=\frac1q+\frac1r$
\begin{align}\label{equ:Schatten.gen.Hölder}
    \|XY\|_p\leq\|X\|_q\|Y\|_r,
\end{align}
with equality if and only if $|X|^q\propto|Y^*|^r$.

\medskip

Given some positive operator $a\geq0$ we denote its support projection as $\Pi_a$, that is the minimal projection s.t. $a=\Pi_aa\Pi_a$. Hence for any operator $X$ the projection onto its image (left support) is $\Pi_{XX^*}$ and the one onto its support (right support) is $\Pi_{X^*X}$. Given a bipartite operator $X\in\cB(\cH\otimes\cK)$, then its \emph{left-},  respectively \emph{right-marginal support projections} $L_X,R_X$ on $\cH$ are the minimal projections s.t.
\begin{align}\label{equ:full.to.marginal.support}
    X=(L_X\otimes\1)X, \quad X= X(R_X\otimes\1)
\end{align} hold, respectively. We will interchangeably call them \emph{projections onto the marginal image}, respectively, \emph{projections onto the marginal support} and the marginal in question will be clear from the context.
They are also equally characterized by being the minimal projections which satisfy $\Pi_{XX^*}\leq L_X\otimes\1, \quad \Pi_{X^*X}\leq R_X\otimes\1$. 

\noindent When inverting an operator $a\geq0$ we will generically be using the generalized Moore-Penrose inverse, which is a bounded operator $a^{-1}$ s.t. $a^{-1}a=aa^{-1}=\Pi_a$. Effectively it is the inverse on the support of that operator and $0$ elsewhere.

\medskip 

Operators acting on a subsystem are indexed by the subsystem, and identities are suppressed, e.g.\ for $a\in\cB(\cH_1)$ we write $a_1\equiv a\otimes\1_2$, and for system labels $A,B,E,Q,P,\dots$ we write $X_{AB}$, $\1_B$, etc. The identity map on $\cB(\cH_E)$ is $\id_E$ (or $\id_n$ if $\cH_E\simeq\C^n$), and $\SWAP_{Q\leftrightarrow P}$ denotes the unitary channel exchanging two systems (or, more generally, reordering tensor factors). A linear map $\Phi:\cB(\cH_Q)\to\cB(\cH_P)$ is also written $\Phi:Q\to P$ and its adjoint with respect to the Hilbert--Schmidt inner product is denoted by $\Phi^*$. We use the standard notions of completely positive (CP) and trace-preserving (TP) maps and a \emph{quantum channel} is a CPTP map. Every CP map admits a Kraus decomposition $\Phi(X)=\sum_iK_iXK_i^*$ which satisfies $\sum_iK^*_iK_i=\1$ if the map is also TP.

\medskip

An ordered collection $\cQ=(q_1,\dots,q_n)$ of indices, either all positive or all negative, is called \emph{(mutually) compatible} if $\max_{i,j}\big|\frac1{q_i}-\frac1{q_j}\big|\leq1$, and we write $|\cQ|=n$ for its length and commonly $-\cQ=(-q_1,\dots,-q_n)$. Two compatible collections $\cQ,\cP$ are compatible with each other if the union of their indices is compatible. We say that $\cQ$ is \emph{normed} if all $q_i\in[1,\infty]$.

\medskip

For multi-indexed (quasi-)norms we use two interchangeable notations. In the \emph{positional} notation $\|X\|_{(q_1,\dots,q_n)}$, $X\in\cB(\cH_1\otimes\cdots\otimes\cH_n)$, where the index $q_i$ refers to the $i$-th tensor factor. In the \emph{labeled} notation each index is tagged with the system carrying it, e.g.\ $\|X_{QP}\|_{(Q:q,P:p)}\equiv\|X\|_{(q,p)}$, and the order of the systems is the order of the labels in the subscript, so technically $\|X_{QP}\|_{(Q:q,P:p)}\equiv\|X_{PQ}\|_{(Q:q,P:p)}$, though we will avoid mixing the order of the indices in the operator and the index, wherever possible. If $Q=Q_1\cdots Q_n$ carries the collection $\cQ=(q_1,\dots, q_n)$ we write $\|X\|_{(Q:\cQ)}$, and e.g.\ $(E\!:\!t,Q\!:\!\cQ)$ denotes the collection obtained by prepending a system $E$ with index $t$. Inside a subscript three separators occur, each with a single meaning: a \emph{comma} separates indices of different tensor factors, a \emph{colon} attaches an index to a system label and a \emph{semicolon} separates a collection of indices from a trailing operator space, $(\cQ;\cX)$, see \Cref{Prelim:Pisier}.

\subsection{Preliminaries: Pisier's Operator-Valued Schatten Spaces}\label{Prelim:Pisier}

For completeness' sake we recall some basics of Pisier's \textit{operator-valued Schatten norms}, sometimes also referred to as \textit{Pisier norms}. For a proper exposition we refer to \cite{book:Pisier.2003, book:EffrosRuan.2000, book:Paulsen.2003} for operator spaces, to \cite{book:Pisier.1998} for operator-valued Schatten spaces, and to \cite{Beigi.2023} for an overview aimed at quantum information theory people. For a full account of Pisier's construction we refer the reader to Pisier's \cite{book:Pisier.1998} or the comprehensive \cite{book:Junge.2010}. We summarized some basics of operator spaces, including their definition via embeddings and dualization in \cref{App:OpSpaces}.

\noindent Let $\cX\subset\cB(\cK)$ be some operator space. Then Pisier constructs the operator-valued $q=\infty$ Schatten space via
\begin{align}
    \cS_\infty[\cH;\cX] := \mathcal{C}\otimes_h \cX \otimes_h \mathcal{R} = \cB_{\text{cpt}}(\cH)\otimes_{\min} \cX,
\end{align} where $\otimes_h$ is the Haagerup tensor product and $\cC, \cR$ are, respectively, column- and row-operator spaces, and $\otimes_{\min}$ the min tensor product, see \cite{book:Pisier.2003, Beigi.2023}.
Thus, effectively, $\cS_\infty[\cH;\cX]$ is the space of compact operators in $\cB(\cH\otimes\cK)$ which gain their operator space structure from exactly this embedding. If $\dim(\cH)=d$ then one can show that this implies $\cS_\infty[\mathbb{C}^d;\cX]=M_d(\cX)$ with equal operator space structures.

\noindent Next Pisier defines the operator-valued $q=1$ Schatten space 
\begin{align}
    \cS_1[\cH;\cX]:= (\cS_\infty[\cH;\cX^*])^*
\end{align} as the operator space dual of the operator-valued Schatten space $\cS_\infty[\cH;\cX^*]$ over the operator space dual $\cX^*$. Finally, for arbitrary indices $p\in[1,\infty]$, the corresponding operator-valued Schatten space is defined via complex interpolation
\begin{align}
    \cS_p[\cH;\cX]:= [\cS_\infty[\cH;\cX]\,,\, \cS_1[\cH;\cX]]_\frac{1}{p}.
\end{align}

\noindent As a direct consequence it follows that they satisfy the following, natural, operator space duality relation
\begin{align}
\label{eq:duality-op-space}
   \left( \mathcal{S}_p[\cH;\cX]\right)^* = \mathcal{S}_{p^\prime}[\cH;\cX^*],
\end{align} where $\frac{1}{p}+\frac{1}{p^\prime}=1$ are Hölder dual indices and $\mathcal{X}^*$ is the operator space dual of $\mathcal{X}$ \cite{book:Pisier.2003}.

\noindent When taking $\cX=\mathbb{C}$, it turns out that, as expected, all these spaces recover the usual $\cS_p(\cH)$ Schatten spaces. The neat and powerful point, as mentioned before, is that since $\cS_{p}(\cH)$ is a Banach space, it can be imbued with a natural operator space structure, i.e., we can also view it as an operator space. Thus the above construction gives rise to the spaces
\begin{align}
    \cS_q[\cH_1,\cS_p(\cH_2)], \, \cS_q[\cH_1,\cS_p[\cH_2;\cX]],\, \dots\,,
\end{align}
where the spaces most often considered in the quantum information literature, though not exclusively, are the two-indexed $\cS_q[\cH_1,\cS_p(\cH_2)]$, see e.g.\ \cite{Devetak.2006, Beigi.2016, Bardet.2024, Beigi.2023, Fawzi.2026}.

\noindent From these definitions Pisier also demonstrated the following property.
\begin{proposition}[Corollary 1.3 \cite{book:Pisier.1998}]\label{prelim:pisier.block.colaps}
Let $X_i\in\cX$, then for $\oplus_iX_i\in \cS_p[\cH;\cX]$ and $p\in[1,\infty)$ it holds that
\begin{align}
    \|\oplus_iX_i\|_{\cS_p[\cH;\cX]}&= \left(\sum_i\|X_i\|^p_{\cX}\right)^\frac{1}{p}, \\
    \|\oplus_iX_i\|_{\cS_\infty[\cH;\cX]}&= \sup_i\|X_i\|_{\cX}
\end{align}
\end{proposition}

A consequence of this is the following multiplicativity under tensor products.
\begin{corollary}
Let $\cX$ be an operator space and $Y\in \cX$. Let $p\in[1,\infty]$ and $X\in\cS_p(\cH)$, then
\begin{align}
    \|X\otimes Y\|_{\cS_p[\cH;\cX]} = \|X\|_{\cS_p(\cH)}\cdot \|Y\|_{\cX}.
\end{align}
\end{corollary}
\begin{proof}
Let $X=UDV$ be its singular value decomposition with $D=\text{diag}(x_i)$, then
\begin{align}
    \|X\otimes Y\|_{\cS_p[\cH;\cX]} &= \|(U\otimes\1)(D\otimes Y)(V\otimes\1)\|_{\cS_p[\cH;\cX]} = \|D\otimes Y\|_{\cS_p[\cH;\cX]} \\&= \|\oplus_ix_i Y\|_{\cS_p[\cH;\cX]} = \left(\sum_i\|x_iY\|^p_{\cX}\right)^\frac{1}{p}= \|X\|_{\cS_p(\cH)}\|Y\|_{\cX},
\end{align} where we used the local isometric invariance of these norms, which is a consequence of \cite[Lemma 1.6]{book:Pisier.1998}.
\end{proof}

\noindent So far these norms were defined somewhat abstractly, and the problem of computing them for explicit elements seems rather difficult. This is changed with a seminal result of Pisier's pertaining to the structure of his norms, which is sometimes known as \emph{Pisier's theorem} or \emph{Pisier's formula}.

\begin{theorem}[Pisier's formula; Theorem 1.5 in \cite{book:Pisier.1998}] \label{thm:PisierOriginal}
Let $\mathcal{H}$ be a separable Hilbert space and $\mathcal{X}$ an operator space. Then for $1\leq p\leq\infty$ the following variational formulas hold for any $X\in \mathcal{S}_p[\cH;\cX]$.
\begin{align}
    \|X\|_{\mathcal{S}_p[\cH;\cX]}= \inf_{\underset{X=aYb}{a,b\in \mathcal{S}_{2p}(\mathcal{H}), Y\in \mathcal{S}_\infty[\cH;\cX]}}\|a\|_{2p}\|b\|_{2p}\|Y\|_{\mathcal{S}_\infty[\cH;\cX]}.
\end{align}
\end{theorem}
\begin{remark}
    This is an incredibly important result. Not just because in theory it makes these norms tractable to compute, even numerically as a convex optimization problem, or because it also allows us to directly relate them to entropic quantities, but chiefly because it is the basis of all proofs of their important structural properties \cite{book:Pisier.1998, Beigi.2023, Fawzi.2026}. 
\end{remark}

\noindent One such consequence is the following about completely bounded norms.
\begin{proposition}[Lemma 1.7 \cite{book:Pisier.1998}]\label{prelim:Pisier.cb.via.q}
Let $\cX,\cY$ be operator spaces and $\phi:\cX\to \cY$ a completely bounded map, then
\begin{align}
    \|\phi\|_{\cb,\cX\to\cY} = \sup_{n\in\mathbb{N}}\|\id_n\otimes \phi\|_{\cS_p[\C^n;\cX]\to \cS_p[\C^n;\cY]}
\end{align} holds for any $p\in[1,\infty]$.\footnote{Note that for $p=\infty$ the right-hand side is exactly equal to the definition as $\cS_\infty[\C^n;\cX]\simeq M_n(\cX)$ completely isometrically, so the difficulty lies in proving this for $p\in[1,\infty)$.}
\end{proposition}

In practice, however, \Cref{thm:PisierOriginal} is not quite sufficient for certain applications, in particular when studying the two-indexed $\cS_q[\cH_1,\cS_p(\cH_2)]$ spaces where it would relate the $\cS_q[\cH_1,\cS_p(\cH_2)]$ on the LHS with the $\cS_\infty[\cH_1,\cS_p(\cH_2)]$-norm on the RHS. Instead, using \Cref{thm:PisierOriginal} and the duality relations, which hold by definition, one can derive more specialized factorization formulas in these cases.

\begin{theorem}[\cite{Devetak.2006, Beigi.2023}]\label{thm:2.pisier.variational.expressions}
\label{thm:expression-pq}
Given an element $X\in \mathcal{S}_q[\mathcal{H}_1,\mathcal{S}_p(\mathcal{H}_2)]$ acting on the Hilbert space $\mathcal{H}_{12}=\mathcal{H}_1\otimes\mathcal{H}_2$ and $q,p\in[1,\infty]$, then it holds that
\begin{align}
\|X_{12}\|_{\mathcal{S}_q[\mathcal{H}_1,\mathcal{S}_p(\mathcal{H}_2)]}
=
    \begin{cases}
        \inf_{\underset{X_{12}=a_1Y_{12}b_1}{a,b\in \mathcal{S}_{2r}(\mathcal{H}_1),Y\in \mathcal{S}_{p}(\mathcal{H}_{12})}} \|a\|_{2r}\|b\|_{2r}\|Y\|_{p}\, & \quad \textup{for } q< p, \\
        \|X_{12}\|_p\, & \quad \textup{for } q=p, \\
        \sup_{a,b\in\mathcal{S}_{2r}(\mathcal{H}_1)} \|a\|^{-1}_{2r}\|b\|^{-1}_{2r}\|a_1X_{12}b_1\|_p & \quad \textup{for } q> p.
    \end{cases}
\end{align}
where the infimum and supremum are over $a,b\in \mathcal{S}_{2r}(\mathcal{H}_1)$ acting only on the first Hilbert space and $Y\in \mathcal{S}_p(\mathcal{H}_{12})$ with $\frac{1}{r}:=\left|\frac{1}{q}-\frac{1}{p}\right|$.
Further if $X_{12}\geq0$, then one can choose $a=b\geq0$, see \cite[Proposition 3.1 (v,vi)]{Bardet.2022} and \cite[Proposition 5.2 iii)]{Bardet.2024}.
\end{theorem}
\begin{remark}
This formula is often taken as the starting point for these two-indexed norms in more quantum-information-focused papers \cite{Devetak.2006, Beigi.2023, Fawzi.2026, Bardet.2022, Beigi.2016}. Pisier's construction allows, though, for a more general version of this formula, in which the second system is kept a general operator space $\cX$ instead of $\cS_p(\cH_2)$. This was proven in \cite[Lemma 3.1]{Fawzi.2026} using only Pisier's formula and the operator space duality \eqref{eq:duality-op-space}. 
\end{remark}
We state it, in a slightly strengthened form, as \Cref{thm:VariationalSpOperatorspace} below, as it will be important for the proof of our technical main result \Cref{thm:Multi.Thm.Pisier.like}. For that we will require another property of operator-valued Schatten norms, namely the following concavity-convexity from \cite{Kochanowski.2026}.

\begin{lemma}[Lemma C.1 of \cite{Kochanowski.2026}]\label{lem.app.op.concavity}
Let $\cH$ be a finite-dimensional Hilbert space, $\cX$ be an operator space and fix $1\leq r\leq s$. Then it holds that the map 
\begin{align}
  (a,A)\mapsto  \|(a^{\frac{1}{2s}}\otimes\1)A(a^{\frac{1}{2s}}\otimes\1)\|_{\cS_r[\cH;\cX]}
\end{align} is concave and upper semi-continuous in $a\in\cS^+_1(\cH), \|a\|_1\leq1$ and continuous and convex in $0\leq A\in\cS_\infty[\cH;\cX]$.
\end{lemma}

\section{Multi-Indexed Schatten Spaces}\label{sec:MultiIndex}

It is a natural question how to extend these constructions, proof techniques, and properties to indeed arbitrarily many indices.
The goal of this section is to answer exactly this question by introducing and studying \emph{multi-indexed Schatten Spaces}. We will show that the naïve extension of the factorization idea from the two-indexed case to arbitrarily many, and even negative, indices leads to a theory of multi-indexed Schatten spaces, of which the two-indexed case is a special case. Unlike in the two-indexed construction of \cite{Kochanowski.2026}, here we do require some minimal results of Pisier's operator-valued Schatten norms theory to deal with certain technicalities. 

We will first extend Pisier's operator-valued Schatten space approach, in particular \Cref{thm:PisierOriginal}, to arbitrarily many indices using some of the non-commutative analysis tools already used in \cite{Kochanowski.2026}.

\subsection{A Multi-Indexed Extension of Pisier's Formula}\label{sec:Multi.Pisier}

Recall that Pisier's construction of operator-valued Schatten spaces $\cS_q[\cH;\cX]$, with $q\in[1,\infty]$ and $\cX$ any operator space, allows for iteration, since by construction these $\cS_q[\cH;\cX]$ are themselves again operator spaces. This leads to the following straightforward definition, used implicitly in the literature, for multi-indexed operator-valued Schatten spaces, with indices in $[1,\infty]$. 

\begin{definition}[Multi-indexed normed operator-valued Schatten spaces]\label{def:Multi.Pisier.Schatten}
Given an ordered list of indices $q_1,...,q_n\in[1,\infty]$ and Hilbert spaces $\{\cH_i\}_{i=1}^n$ we define the associated \emph{normed multi-indexed Schatten space} $\mathcal{Q}$ to indices $(q_1,...,q_n)$ as
\begin{align}
    \cS_\mathcal{Q}(\{\cH_i\}_{i=1}^n):=\cS_{q_1}[\cH_1,\cS_{q_2}[\cH_2,...,\cS_{q_n}(\cH_n)...]]
\end{align}
and we will in abuse of notation write $\mathcal{Q}\equiv(q_1,...,q_n)$. Intuitively it is the space of (suitably regular) operators acting on $\otimes_{i=1}^n\cH_i$ which assigns Schatten indices $q_i$ to system $i$.
Its associated norm is then written as 
\begin{align}
    \|X\|_{(\cQ)}\equiv\|X\|_{(q_1,...,q_n)}.
\end{align}
Likewise, given an operator space $\cX$ we define the associated \emph{multi-indexed operator-valued Schatten space} $\cS_\mathcal{Q}[\{\cH_i\}_{i=1}^n;\cX]$ to indices $(q_1,...,q_n)$ as
\begin{align}
    \cS_\mathcal{Q}[\{\cH_i\}_{i=1}^n;\cX]:=\cS_{q_1}[\cH_1;\cS_{q_2}[\cH_2,...,\cS_{q_n}[\cH_n,\cX]...]]
\end{align}
and we write $(\mathcal{Q};\cX)=(q_1,...,q_n;\cX)$. It is the space of (suitably regular) $\cX$-valued operators on $\otimes_{i=1}^n\cH_i$ which assigns Schatten indices $q_i$ to system $i$.
Its associated norm is then written as 
\begin{align}
    \|X\|_{(\cQ;\cX)}\equiv\|X\|_{(q_1,...,q_n;\cX)}.
\end{align}
\end{definition}

This definition alone is, however, arguably not too useful, as it does not tell us how to compute and work with these norms $\|\cdot\|_{(\cQ)}$ analytically. What made the (single-indexed) operator-valued Schatten norms $\cS_q[\cH;\cX]$ tractable was Pisier's theorem \Cref{thm:PisierOriginal}. We shall now derive a multi-indexed extension of it to make the above norms tractable. We do this in two parts. Firstly we effectively rewrite \cite[Lemma 3.1]{Fawzi.2026} using some of the ideas of \cite[Section 3.1]{Kochanowski.2026}.

\begin{theorem}[Generalized variational expressions, adaptation of Lemma 3.1 from \cite{Fawzi.2026}] \label{thm:VariationalSpOperatorspace}
Let $\mathcal{X}$ be an operator space. Then for $1\leq p\leq q\leq \infty$ and $\frac{1}{r}=\frac{1}{p}-\frac{1}{q}\,$, the following factorization formulas hold for any $X\in \mathcal{S}_p[\mathcal{H};\mathcal{X}],\mathcal{S}_q[\mathcal{H};\mathcal{X}]$ respectively:
\begin{align}
\operatorname{(i)}\qquad \|X\|_{\mathcal{S}_p[\mathcal{H};\mathcal{X}]}&=\inf_{\substack{a,b\in \mathcal{S}_{2r}(\mathcal{H}), Y\in\mathcal{S}_q[\mathcal{H}; \mathcal{X}] \\X=aYb}} \|a\|_{2r}\|b\|_{2r} \|Y\|_{\mathcal{S}_q[\cH;\mathcal{X}]}\\
&= \inf_{\underset{\Pi_a=L_X, \Pi_b=R_X}{a,b\geq0}}\|a\|_{2r}\|b\|_{2r}\|a_1^{-1}Xb_1^{-1}\|_{\mathcal{S}_q[\cH;\mathcal{X}]}
\\
\operatorname{(ii)}\qquad \|X\|_{\mathcal{S}_q[\mathcal{H};\mathcal{X}]} &= \sup_{a,b\in \mathcal{S}_{2r}(\mathcal{H})}\|a\|_{2r}^{-1}\|b\|_{2r}^{-1}\|a_1Xb_1\|_{\mathcal{S}_p[\mathcal{H};\mathcal{X}]} \\ 
&= \sup_{\underset{\Pi_a=L_X, \Pi_b=R_X}{a,b\geq0}}\|a\|^{-1}_{2r}\|b\|^{-1}_{2r}\|a_1Xb_1\|_{\mathcal{S}_p[\cH;\mathcal{X}]}\,.
\end{align}
Moreover, when $X=X^*$ is self-adjoint, one can restrict the optimization to $a=b$. 
\end{theorem}

Before we prove \Cref{thm:VariationalSpOperatorspace} some remarks are in order.
\begin{remark}
    This theorem combines in some sense three results in one. One is of course the way in which it allows the factorization of $\cS_q[\cH;\cX]$ through differently indexed $\cS_p[\cH;\cX]$ for any $q,p\in[1,\infty]$, which is the result of \cite[Lemma 3.1]{Fawzi.2026} (inspired by \cite[Theorem 4.5]{Beigi.2023}) and shall serve as inspiration for multi-indexed Schatten quasi-norms. Technically this represents a generalization of the well-known cases where $\cX=\mathbb{C}$, which is \Cref{prop:SpVariationalFormulas}, or $\cX=\cS_p(\cH_2)$, which is \Cref{thm:2.pisier.variational.expressions}, see also \cite{Devetak.2006, Beigi.2023} or even $\cX=\cS_\infty[\cH_2;\cY]$ \cite{book:Pisier.1998}. However, since $\cX$ can be any operator space, and thus in particular any normed multi-indexed Schatten space, the meaning is much more general. This will be important in our treatment of multi-indexed Schatten spaces in this work. 

    The second result, which goes beyond \cite[Lemma 3.1]{Fawzi.2026}, is that we may restrict the optimization to positive semidefinite $a,b$ whose supports are equal to the left-, respectively right-marginal support of $X$. These are the smallest orthogonal projections $L_X, R_X$ s.t. $X=(L_X\otimes\1)X(R_X\otimes\1)$. In that manner $a^{-1},b^{-1}$ are well defined as the generalized Moore-Penrose inverses of $a,b$. The proof of this is essentially an application of the proof ideas of \cite[Section 3.1]{Kochanowski.2026}.

    Finally, it shows that when $X=X^*$ we may restrict the optimization to $a=b$. This is somewhat surprising since statements like this were previously really only established for $X\geq0$. In that case it has allowed for major connections between these operator-valued Schatten spaces and quantum information theory \cite{Devetak.2006, Beigi.2023, Fawzi.2026}. In the case $\cX=\mathbb{C}$ this is almost trivially known, as the optimizers can explicitly be written down, and for the case $\cX=\cS_p(\cH_2)$ see e.g. \cite{Beigi.2023, Bardet.2024}. The proof we present here follows the same proof idea as the $\cX=\cS_\infty(\cH)$ case treated in \cite[Proposition 5.2 (iii)]{Bardet.2024} which we extend to any $\cX$.
\end{remark}

\begin{proof}[Proof of \Cref{thm:VariationalSpOperatorspace}]
Since this theorem combines several results, its proof is an amalgamation of the proofs of \cite[Lemma 3.1]{Fawzi.2026} (inspired by \cite{Beigi.2023}). It can be found in Appendix \ref{app:Generalized.Variational.Expressions}.
\end{proof}

By iteratively applying this theorem we have a concatenation of factorization formulas for any multi-indexed Schatten norm defined in \cref{def:Multi.Pisier.Schatten}, as 
\begin{align} 
    \|X\|_{(\cQ;\cX)} = \prod_{i=1}^{n-1}\left( \opt_{a^{(i)},b^{(i)}}\|a^{(i)}\|^{\mp1}_{2r_i}\|b^{(i)}\|^{\mp1}_{2r_i}\right)\Bigg\|\bigg(\prod_{i=n-1}^{1} a_{[i]}^{(i)\,\pm1}\bigg)X\bigg(\prod_{i=1}^{n-1} b_{[i]}^{(i)\,\pm1}\bigg)\Bigg\|_{(q_n;\cX)} \label{equ:iterated.pisier}
\end{align}
where $a^{(i)},b^{(i)}$ act on the first $i$ subsystems and $\frac{1}{r_i}=|\frac{1}{q_i}-\frac{1}{q_{i+1}}|$ for all $i\in[n-1]$. Here each $\opt$ is either an $\inf$ ($q_i\leq q_{i+1}$, lower sign) or a $\sup$ ($q_i\geq q_{i+1}$, upper sign). We dropped writing explicitly the support conditions on $a^{(i)},b^{(i)}$.

\noindent While, by Pisier's construction and \cref{thm:VariationalSpOperatorspace}, we know that for $q_i\in[1,\infty]$ the RHS of \eqref{equ:iterated.pisier} defines a norm, \eqref{equ:iterated.pisier} as a concatenation of optimizations, is, however, still not very analytically tractable. We change this in the following manifestly multi-indexed version of Pisier's formula. 

\begin{theorem}[A multi-indexed Pisier formula]\label{thm:main.multi.pisier}
Let $\mathcal{Q}$ be a normed multi-indexed Schatten space of \cref{def:Multi.Pisier.Schatten} with indices $q_i\in[1,\infty]$ and $\mathcal{X}$ some operator space, then it holds that
\begin{align}\label{equ:main.multi.pisier}
    \|X\|_{(\mathcal{Q};\cX)} &=\inf_{\underset{X=(A\otimes\1)Y(B\otimes\1)}{A,B,Y}}\|AA^*\|_{(\mathcal{Q})}^{\frac{1}{2}}\|B^*B\|_{(\mathcal{Q})}^{\frac{1}{2}}\|Y\|_{(\infty,...,\infty;\cX)} \\ &=\inf_{\underset{X=(\Pi_{AA^*}\otimes\1)X(\Pi_{B^*B}\otimes\1)}{A,B}}\!\!\!\!\!\|AA^*\|_{(\mathcal{Q})}^{\frac{1}{2}}\|B^*B\|_{(\mathcal{Q})}^{\frac{1}{2}}\|(A^{-1}\otimes \1)X(B^{-1}\otimes\1)\|_{(\infty,...,\infty;\cX)}.
\end{align} Moreover, when $X=X^*$, then we can restrict to $A=B^*$.
\end{theorem}
\begin{remark}
This variational expression is the natural extension of Pisier's seminal factorization expression for the norms of his operator-valued Schatten spaces \cite[Theorem 1.5]{book:Pisier.1998}, which is recovered in the simple case of just one index $\mathcal{Q}=(q)$. See e.g. \cite[equation (3.3)]{Devetak.2006} or \cite[Theorem 4.4]{Beigi.2016}.

We note that technically this result is not directly required for the general theory of multi-indexed Schatten spaces we develop in the next subsection; however, we present it since it is a structurally nice result demonstrating the use of matrix analytic tools in Pisier's theory and gives a particularly simple proof for the following \cref{cor:Pisier.multi.tensor.splitting}.
\end{remark}

\begin{proof}[Proof of \cref{thm:main.multi.pisier}]
    The proof of this theorem can be found in \Cref{app:main.multi.pisier}. It results from multiple uses of \cref{thm:VariationalSpOperatorspace} in which the appearing inf and sup can be exchanged using Sion's minimax theorem building on joint convexity and concavity statements from \cite{Zhang.2020}.
\end{proof}

One immediate consequence of this theorem is the following splitting of multi-indexed operator-valued Schatten norms under tensor products. 

\begin{corollary}[Splitting under arbitrary tensor products]\label{cor:Pisier.multi.tensor.splitting}
Let $\mathcal{Q}$ be a normed multi-indexed Schatten space and $\cX$ an operator space. Then for $X\otimes Y\in\cS_{\cQ}[\{\cH_i\}_{i=1}^n;\cX]$ it holds that
\begin{align}
    \|X\otimes Y\|_{(\mathcal{Q};\cX)}=\|X\|_{(\cQ)}\|Y\|_{\cX}.
\end{align}
\end{corollary}

This will be required in the proofs of \Cref{cor.multi.tensor.multiplicativity} and \Cref{lem:reductionLemma}.
\begin{proof}
The proof is very straightforward and goes like the single-indexed case when replacing Pisier's theorem by \Cref{thm:main.multi.pisier}.
\begin{align}
    \|X\otimes Y\|_{(\mathcal{Q};\cX)}&=\inf_{A,B}\|AA^*\|^\frac{1}{2}_{(\cQ)}\|B^*B\|^\frac{1}{2}_{(\cQ)}\|A^{-1}XB^{-1}\otimes Y\|_{(\infty,...,\infty;\cX)}\\ &=\inf_{A,B}\|AA^*\|^\frac{1}{2}_{(\cQ)}\|B^*B\|^\frac{1}{2}_{(\cQ)}\|A^{-1}XB^{-1}\|_\infty\| Y\|_{\cX} = \|X\|_{(\cQ)}\|Y\|_{\cX},
\end{align}
where we used that the $\cS_\infty[\cH;\cX]$ norm splits under tensor products \cite{book:Pisier.1998}.
\end{proof}

\subsection{General Multi-Indexed Schatten Spaces}\label{sec:GenMultiIndex}

We are finally at a point where we can introduce and study \emph{multi-indexed Schatten spaces}. This is the main mathematical object that this work introduces, analyzes, and applies within quantum information theory.

\noindent In this subsection we first give a definition for both positive and negative compatible indices and relate the two cases. Then, using the tools developed in \cite{Kochanowski.2026} and \cref{sec:Multi.Pisier}, we derive powerful factorization formulas that make these (quasi-)norms tractable. These are
\begin{itemize}
    \item \Cref{thm:Multi.Thm.Pisier.like}: Factoring through Schatten spaces,
    \item \Cref{thm:multi.quasi.relational.consistency}: Multi-indexed relational consistency.
\end{itemize}
They should be seen as the technical main results of this work, on which all subsequent results of this section and all of the next \Cref{sec:LinerMaps} build. In particular, we use them to show that these multi-indexed objects satisfy several properties one would expect. These are, among others,
\begin{itemize}
    \item \Cref{cor:partial.trace.1}: Partial trace via $1$ as last index on positive operators,
    \item \Cref{prop:MultiHölderDual}: The Hölder duality,
    \item \Cref{cor:reverse.multi.hölder}: The reverse Hölder duality,
    \item \Cref{cor.multi.tensor.multiplicativity}: Multiplicativity under tensor products,
    \item \Cref{cor:multi.block.collapse}: Collapse under block-diagonal operators,
    \item \Cref{thm:Multi.Quasi.Norms}: The triangle and quasi-triangle inequality, and
    \item \Cref{thm:Multi.Rev.Triangle}: The reverse triangle inequality.
\end{itemize}
These should be seen as natural generalizations of the ones we, with Fawzi and Rouzé, proved for the two-indexed versions in \cite{Kochanowski.2026}, and thus also of Pisier's ones \cite{book:Pisier.1998}. The proofs we present here are independent of \cite{Kochanowski.2026}, hence the structural results there can be seen as special cases of the ones in this work.

\medskip

In order to specify a multi-indexed Schatten space, we require an ordered list of, in this work, finite-dimensional Hilbert spaces $\cH_i$ and associated indices $q_i\neq0$.
\begin{definition}[Compatible indices]
    We will call a collection of indices $\{q_i\}_{i=1}^n$ with either $q_i>0$ for all $i$ or $q_i<0$ for all $i$ \emph{(mutually) compatible} if $\max_{ij}\bigl|\frac{1}{q_i}-\frac{1}{q_j}\bigr|\leq 1$.
    We will write $\cQ=(q_1,\dots,q_n)$ for such collections and denote its length by $|\cQ|=n$.

    \noindent Given two sets of compatible indices $\cQ=(q_1,\dots, q_n)$ and $\cP=(p_1,\dots, p_m)$, we say that $\cQ$ and $\cP$ are \emph{compatible} if the collection of all indices $\{q_i,p_j\}_{i\in[n], j\in[m]}$ is compatible.
\end{definition}

\begin{remark}
    Thus, compatibility comprises, on the one hand, that all indices are of the same sign and that the distances between the inverses are bounded by 1. The former is, in principle, not strictly necessary; we assume it for simplicity's sake and leave it to future work to explore mixed-sign Schatten spaces. In contrast, the latter is a direct extension of the notion of compatibility in \cite{Kochanowski.2026}, which was necessary there, at least for \cite[Theorem 1.4]{Kochanowski.2026}.
\end{remark}

With this we now define multi-indexed Schatten spaces in the following manner.
\begin{definition}[Multi-indexed Schatten spaces]\label{def:GeneralMultiIndexed}
Let $\cQ=(q_1,...,q_n)$ be an ordered collection of $|\cQ|=n$ mutually compatible indices $q_i>0$ and Hilbert spaces $\cH_i$.
Then we associate to these the (quasi-)norm $\|X\|_{(\cQ)}$ on $X\in \cB(\otimes_{i=1}^n\cH_i)$ defined iteratively via
\begin{align}\label{equ:Multi.Def.1}
    \|X\|_{(q_1,q_2,\cdots, q_n)} := \begin{cases}
        &\!\!\inf_{\underset{\Pi_a=L_X,\, \Pi_b=R_X}{a,b\geq0}} \!\!\!\|a\|_{2r}\|b\|_{2r}\|a^{-1}_1X_{1\dots n}b^{-1}_1\|_{(q_2,q_2,q_3,\cdots,q_n)} \!\!\!\qquad q_1\leq q_2,\\ &\!\!\sup_{\underset{\Pi_a=L_X,\, \Pi_b=R_X}{a,b\geq0}} \!\!\!\|a\|^{-1}_{2r}\|b\|^{-1}_{2r}\|a_1X_{1\dots n}b_1\|_{(q_2,q_2,q_3,\cdots,q_n)} \qquad q_1\geq q_2,
    \end{cases}
\end{align} where $\frac{1}{r}:=|\frac{1}{q_1}-\frac{1}{q_2}|$ and
\begin{align}\label{equ:Multi.Def.2}
    \|X\|_{(q)}:=\|X\|_q.
\end{align}
Importantly here $\cQ^\prime\equiv(q_2,q_2,q_3,\cdots,q_n)$ is taken as an $(n-1)$-indexed Schatten (quasi-)norm acting on the Hilbert spaces $(\cH_1\otimes\cH_2, \cH_3,\cdots,\cH_n)$ with indices $q_2,q_3,\dots, q_n$. 

\noindent We will be interchangeably writing
\begin{align}
    \|X\|_{(\cQ)} \equiv \|X\|_{(q_1,q_2,\dots, q_n)}.
\end{align}

\noindent The space $\cB(\otimes_{i=1}^n\cH_i)$ endowed with the quasi-norm $\|\cdot\|_{(\cQ)}$ is then the associated \emph{multi-indexed Schatten space} $\cS_{\cQ}$.

\noindent For compatible $q_i<0$ we take the same definition as above but restricted to $X>0$ and $a=b^*>0$. The associated `space' is then technically not a linear space but a cone of positive operators $\cB^+(\otimes_{i=1}^n\cH_i)$ endowed with the anti-norm $\|\cdot\|_{(\cQ)}$.
\end{definition}

\begin{remark}
Clearly, the two-indexed Schatten quasi-norms studied in \cite{Kochanowski.2026} are the $|\cQ|=2$ special cases of this definition. It also manifestly recovers Pisier's norms of \cref{def:Multi.Pisier.Schatten} for an arbitrary number of indices, since iterated application of this definition there yields exactly \eqref{equ:iterated.pisier}. Thus whenever $\cQ=(q_1,\dots,q_n)$ with $q_i\in[1,\infty]$ the functionals $X\mapsto \|X\|_{(\cQ)}$ are norms by Pisier's theory \cite{book:Pisier.1998} and satisfy in particular Hölder duality.

\noindent The restriction that all indices $q_i$ need to be mutually compatible is a direct and straightforward extension of the two-index compatibility which was sufficient and necessary there. One non-trivial aspect of this section will be to show that it indeed is also sufficient in the general multi-indexed setting.  

\noindent Note that in the iteration whenever $q_i=q_{i+1}$ then both the supremum and infimum apply and yield the same, both via $a=b=\1$. This follows via an inductive argument and the generalized Hölder inequality \eqref{equ:Schatten.gen.Hölder}.

\noindent While at the current stage it is not obvious that these definitions indeed give rise to quasi- or anti-norms, we will be proving this in \cref{sub:quasi.anti}, yet already calling them that.
\end{remark}

Concerning negative indices, which will be helpful both as a proof tool and in certain applications, such as completely bounded versions of reverse hypercontractivity \cite{Cubitt.2015, Beigi.2020, Bardet.2022}, we have the following straightforward consequence.

\begin{proposition}[Negatively indexed Schatten spaces]\label{prop:negative.multi.indexed}
Let $q_i<0$ be a compatible collection of negative indices and $X>0$, then for the multi-indexed Schatten space $\cQ=(q_1,\dots, q_n)$ it holds that
    \begin{align}
        \|X\|_{(\cQ)}= \|X^{-1}\|^{-1}_{(-\cQ)},
    \end{align}  where $(-\cQ)= (-q_1,\dots, -q_n)$ denotes the multi-indexed Schatten quasi-norm with indices $-q_i\!>\!0$. 
\end{proposition}

\begin{proof}
This follows quite simply via induction in $|\cQ|=n$. For Schatten spaces, $n=1$, this is obvious. For $|\cQ|=n+1$ consider $q_1\leq q_2$, then 
\begin{align}
     \|X\|_{(q_1,q_2,\cdots, q_n)} &= 
        \inf_{\underset{}{a,b>0}} \|a\|_{2r}\|b\|_{2r}\|a^{-1}_1X_{1\dots n}b^{-1}_1\|_{(q_2,q_2,q_3,\cdots,q_n)} \\ &= \inf_{\underset{}{a,b>0}}\|a\|_{2r}\|b\|_{2r}\|a_1X^{-1}_{1\dots n}b_1\|^{-1}_{(-q_2,-q_2,-q_3,\cdots,-q_n)} \\
        &= \Big(\sup_{\underset{}{a,b>0}}\|a\|^{-1}_{2r}\|b\|^{-1}_{2r}\|a_1X^{-1}_{1\dots n}b_1\|_{(-q_2,-q_2,-q_3,\cdots,-q_n)}\Big)^{-1} \\
        &= \|X^{-1}\|_{(-q_1,-q_2,\dots, -q_n)}^{-1}
\end{align} where the last equality is again the definition, since $-q_1\geq -q_2$ and $\frac{1}{r}=|\frac{1}{q_1}-\frac{1}{q_2}|=|\frac{1}{-q_1}-\frac{1}{-q_2}|$. Since the case $q_1\geq q_2$ is exactly analogous we omit writing it out. Here, by \Cref{def:GeneralMultiIndexed} for negative indices, all optimizations run over $a=b>0$. That the supremum (respectively infimum) in the definition of $\|X^{-1}\|_{(-\cQ)}$ for the positive operator $X^{-1}$ may likewise be restricted to $a=b$ follows from the self-adjoint case of \Cref{thm:multi.quasi.relational.consistency}, applied with $\cQ=(-q_2,\dots,-q_n)$, whose proof for positive indices does not use the present proposition.
\end{proof}

\begin{remark}
    As we will see, for negative indices these functionals $0<X\mapsto \|X\|_{(\cQ)}$ do not define quasi-norms but anti-norms, nor are they defined on a linear space as they are really only non-trivially defined for invertible $X>0$. 
    However, in order to keep notation and nomenclature simple, we will in abuse of notation call both the positive and negative indexed versions multi-indexed Schatten spaces.
    Where the range of indices is important we will always mention it, so it will be clear from context if we are considering only positive, or positive and negatively indexed ones. If it is not mentioned then one should assume the result to hold in both cases, on their respective spaces or cones.  
\end{remark}

While these are arguably somewhat straightforward definitions, as in \cref{sec:Multi.Pisier} they alone do not yet make these quasi-norms tractable, in the sense that they do not directly allow one to prove the properties we would expect to hold. To this end we prove two important factorization formulas for these multi-indexed Schatten spaces. These will constitute the technical main result of this work.

\subsubsection{Two Key Factorization Formulas}

\begin{theorem}\label{thm:Multi.Thm.Pisier.like}
Let $\cQ$ be a multi-indexed Schatten space with positive indices $\{q_i>0\}_{i=1}^{|\cQ|}$ and let $0<s$ and $0<t\leq 1$ be such that $s\geq q_i\geq t>0$ for all $i$, i.e., all indices in $\cQ$ are bounded from below by $t$ and from above by $s$. Then it holds that
\begin{align}
    \|X\|_{(\cQ)}&=\inf_{\underset{\Pi_{AA^*}=\Pi_{XX^*}, \Pi_{B^*B}=\Pi_{X^*X}}{A,B}}\|AA^*\|^\frac{1}{2}_{(\cS)}\|B^*B\|^\frac{1}{2}_{(\cS)}\|A^{-1}XB^{-1}\|_{s}, \label{equ:Multi.Thm.Pisier.like.inf} \\
    \|X\|_{(\cQ)}&=\sup_{\underset{\Pi_{A^*A}=\Pi_{XX^*}, \Pi_{BB^*}=\Pi_{X^*X}}{A,B}}\|A^*A\|^{-\frac{1}{2}}_{(\cT)}\|BB^*\|^{-\frac{1}{2}}_{(\cT)}\|AXB\|_{t}, \label{equ:Multi.Thm.Pisier.like.sup}
\end{align} where $\cS=(s_1,\cdots,s_n)$ is the multi-indexed Schatten quasi-norm with $\frac{1}{s_i}:=\frac{1}{q_i}-\frac{1}{s}\geq0$ and $\cT=(t_1,\cdots,t_n)$ is the multi-indexed Schatten quasi-norm with $\frac{1}{t_i}:=\frac{1}{t}-\frac{1}{q_i}\geq0$. They are normed with $\frac{1}{s_i},\frac{1}{t_j}\in[0,1]$ whenever $s$, respectively $t$, is compatible with $\cQ$.

\noindent Moreover if $X=X^*$, then one can choose $A=B^*\geq0$.
\end{theorem}

\begin{remark}\label{rem:sup.arbitrary}
The support restriction in the supremum formula \eqref{equ:Multi.Thm.Pisier.like.sup} can be lifted: the supremum may equivalently be taken over arbitrary operators $A,B$ on $\otimes_{i=1}^n\cH_i$, i.e.
\begin{align}
    \|X\|_{(\cQ)}=\sup_{A,B}\|A^*A\|^{-\frac12}_{(\cT)}\|BB^*\|^{-\frac12}_{(\cT)}\|AXB\|_t
\end{align}
holds for all $A,B$. We justify this at the end of the proof below.
\end{remark}

Before giving a remark and the proof, we observe that as a first direct consequence of \Cref{prop:negative.multi.indexed} we get the following version for negative indices. 

\begin{corollary}\label{cor:multi.Pisier.like.negative}
   Let $\cQ$ be a multi-indexed Schatten space with negative indices $\{q_i<0\}_{i=1}^{|\cQ|}$ and let $s<0$ and $-1\leq t<0$ be such that $s\leq q_i\leq t<0$ for all $i$. Then for all $X>0$ it holds that
\begin{align}
    \|X\|_{(\cQ)}&=\inf_{\underset{}{A>0}}\|AA^*\|_{(\cT)}\|A^{-1}XA^{-1}\|_{t}  \\
    \|X\|_{(\cQ)}&=\sup_{\underset{}{A>0}}\|A^*A\|^{-1}_{(\cS)}\|AXA\|_{s}
\end{align} where $\cS=(s_1,\cdots,s_n)$ is the multi-indexed Schatten norm with $\frac{1}{s_i}:=\frac{1}{s}-\frac{1}{q_i}\geq0$ and $\cT=(t_1,\cdots,t_n)$ is the multi-indexed Schatten norm with $\frac{1}{t_i}:=\frac{1}{q_i}-\frac{1}{t}\geq0$. Again, these lie in $[0,1]$ whenever $s$, respectively $t$, is compatible with $\cQ$.
\end{corollary}

\begin{remark}\label{rem:mulit.pisier.like}
    Note that when the upper indices $s$ and lower indices $t\leq1$ are chosen compatible, then $\cS, \cT$ of the statements are normed multi-indexed Schatten spaces as in \cref{def:Multi.Pisier.Schatten}. Such compatible $s,t$ always exist for any positive multi-indexed Schatten space $\cQ$. These are in particular $s=\max_iq_i$, which by definition is compatible with all other indices and likewise $t=\min\{1,\min_iq_i\}$ which is also, as $1$ is compatible with all $q_i\geq1$. The same applies to the negatively indexed case. We prove the result, however, for arbitrary $s,t$, which need not be compatible with $\cQ$. 

\noindent  In the simple case where $\cQ=(q)$ contains just one index, these formulas reduce to well-known variational formulas for Schatten-(quasi-)norms in \Cref{prop:SpVariationalFormulas}. In the two-indexed case they recover \cite[Lemma 1.1]{Kochanowski.2026} when $s=1$ and \cite[Lemma 3.4]{Kochanowski.2026} for $t=\min\{1,q,p\}$. 

\noindent  Note that a contribution of this theorem is also the fact that self-adjointness $X=X^*$, in particular in the supremum case, is sufficient for the symmetric formulas here, instead of requiring $X\geq0$ as in \cite{Devetak.2006, Beigi.2023, Kochanowski.2026}. For the two-indexed infimum formula it was already observed in \cite[Remark 4.6]{Beigi.2023} and \cite[Proof of Proposition 5.2 (iii)]{Bardet.2024} that self-adjointness $X=X^*$ is sufficient. To achieve this we apply the same novel proof idea we gave in \cite[Proof of Lemma 3.2]{Kochanowski.2026} to both the infimum and supremum 
cases here.

\noindent Note that, non-obviously, this theorem with its negative-indexed corollary will suffice in deriving all the structural properties of the general multi-indexed Schatten spaces. Hence it should be considered the technical main result of this work.
\end{remark}

\begin{proof}[Proof of \Cref{thm:Multi.Thm.Pisier.like}]
We will prove these two statements separately, each by induction over $|\cQ|$. Note that both formulas for $|\cQ|=1$ are true by \Cref{prop:SpVariationalFormulas}. We first consider the case where $s,t$ are compatible with $\cQ$.
Let's now consider the induction step for the infimum formula. Let $\cQ=(q_1,q_2,\cdots, q_n)$ with $q_i>0$, denote by $\cQ^\prime=(q_2,q_2,q_3,\cdots, q_n)$ and assume the statement is true for multi-indexed Schatten spaces of size $n-1$. Then we consider two cases separately. First, in the case where $q_1\leq q_2$ with $\frac{1}{r}=\frac{1}{q_1}-\frac{1}{q_2}$ we have
\begin{align}
    \|X\|_{(\cQ)} &=\!\inf_{\underset{\Pi_a=L_X, \Pi_b=R_X}{a,b\geq0}} \|a\|_{2r}\|b\|_{2r}\|a^{-1}_1X_{1\cdots n}b^{-1}_1\|_{(\cQ^\prime)} \\ 
    &=\!\inf_{\underset{\Pi_a=L_X, \Pi_b=R_X}{a,b\geq0}} \|a\|_{2r}\|b\|_{2r}\!\!\!\!\!\!\!\!\!\!\!\!\!\!\!\inf_{\underset{\Pi_{AA^*}=\Pi_{XX^*}, \Pi_{B^*B}=\Pi_{X^*X}}{A,B}}\!\!\!\!\!\!\!\!\!\!\!\!\!\!\!\|AA^*\|^\frac{1}{2}_{(\cS^\prime)}\|B^*B\|^\frac{1}{2}_{(\cS^\prime)}\|A^{-1}a^{-1}_1X_{1\cdots n}b^{-1}_1B^{-1}\|_{s} \\
    &=\!\inf_{\underset{\Pi_a=L_X, \Pi_b=R_X}{a,b\geq0, \|\cdot\|_1\leq 1}}\inf_{\underset{\Pi_{CC^*}=\Pi_{XX^*}, \Pi_{D^*D}=\Pi_{X^*X}}{C,D\geq0}}\!\!\!\!\!\!\!\!\!\!\!\!\!\!\|a_1^{-\frac{1}{2r}}CC^*a_1^{-\frac{1}{2r}}\|^\frac{1}{2}_{(\cS^\prime)}\|b_1^{-\frac{1}{2r}}D^*Db_1^{-\frac{1}{2r}}\|^\frac{1}{2}_{(\cS^\prime)}\|C^{-1}XD^{-1}\|_{s} \\
    &=\!\inf_{\underset{X=\Pi_CX\Pi_D}{C,D\geq0}}\inf_{\underset{\Pi_a=L_X, \Pi_b=R_X}{a,b\geq0, \|\cdot\|_1\leq 1}}\|a_1^{-\frac{1}{2r}}CC^*a_1^{-\frac{1}{2r}}\|^\frac{1}{2}_{(\cS^\prime)}\|b_1^{-\frac{1}{2r}}D^*Db_1^{-\frac{1}{2r}}\|^\frac{1}{2}_{(\cS^\prime)}\|C^{-1}XD^{-1}\|_{s} \\
    &=\!\inf_{\underset{\Pi_{CC^*}=\Pi_{XX^*}, \Pi_{D^*D}=\Pi_{X^*X}}{C,D\geq0}}\|CC^*\|^\frac{1}{2}_{(\cS)}\|D^*D\|^\frac{1}{2}_{(\cS)}\|C^{-1}XD^{-1}\|_{s}, 
\end{align} where the steps are similar to the proof of \cite[Theorem 1.3]{Kochanowski.2026}. In the second line we implicitly used that $\Pi_{AA^*}=\Pi_{YY^*}=\Pi_{XX^*}$ for $Y=a_1^{-1}Xb_1^{-1}$ and the same holds for the right support $\Pi_{B^*B}=\Pi_{X^*X}$ by assumptions on $a_1,b_1$. In the third line we defined $C:=a_1A$ and $D:=Bb_1$, which implies $\Pi_{CC^*}=\Pi_{a_1AA^*a_1}=\Pi_{XX^*}$ and $\Pi_{D^*D}=\Pi_{X^*X}$. Since $\Pi_{AA^*}=\Pi_{XX^*}\leq L_X\otimes\1=\Pi_{a_1}$ we can write $AA^*=a_1^{-1}CC^*a_1^{-1}$ and likewise for $B$. In the fourth line we used that two infima always commute and in the last line we used that $\cS^\prime=(s_2,s_2,s_3,\cdots, s_n)$ and $\cS=(s_1,s_2,\cdots,s_n)$, with $\frac{1}{s_1}=\frac{1}{q_1}-\frac{1}{s}$ and $\frac{1}{s_2}=\frac{1}{q_2}-\frac{1}{s}$, so  $\frac{1}{r}=\frac{1}{q_1}-\frac{1}{q_2} = \frac{1}{s_1}-\frac{1}{s_2}$.

\medskip

Next, the case where $q_1\geq q_2$ is somewhat more involved since it requires an application of Sion's minimax theorem. We have 
\begin{align}
    \|X\|_{(\cQ)} &= \!\sup_{\underset{\Pi_a=L_X, \Pi_b=R_X}{a,b\geq0}} \|a\|^{-1}_{2r}\|b\|^{-1}_{2r}\|a_1X_{1\cdots n}b_1\|_{(\cQ^\prime)} \\ 
    &= \!\sup_{\underset{\Pi_a=L_X, \Pi_b=R_X}{a,b\geq0}}\hspace{-0.2cm} \|a\|^{-1}_{2r}\|b\|^{-1}_{2r} \!\!\!\!\!\!\!\!\!\!\!\!\!\!\!\inf_{\underset{\Pi_{AA^*}=\Pi_{XX^*}, \Pi_{B^*B}=\Pi_{X^*X}}{A,B}}\!\!\!\!\!\!\!\!\!\!\!\!\!\!\!\!\!\!\|AA^*\|^\frac{1}{2}_{(\cS^\prime)}\|B^*B\|^\frac{1}{2}_{(\cS^\prime)}\|A_{1\dots n}^{-1}a_1X_{1\dots n}b_1B_{1\dots n}^{-1}\|_{s} \\ 
    &= \!\sup_{\underset{\Pi_a=L_X, \Pi_b=R_X}{a,b\geq0, \|\cdot\|_1\leq 1}} \inf_{\underset{\Pi_{CC^*}=\Pi_{XX^*}, \Pi_{D^*D}=\Pi_{X^*X}}{C,D}}\!\!\!\!\!\!\!\!\!\!\!\!\!\!\!\!\!\! \|a^\frac{1}{2r}CC^*a^\frac{1}{2r}\|^\frac{1}{2}_{(\cS^\prime)}\|b^\frac{1}{2r}D^*Db^\frac{1}{2r}\|^\frac{1}{2}_{(\cS^\prime)}\|C^{-1}XD^{-1}\|_{s} \\
    &= \Bigg\{\sup_{\underset{\Pi_a=L_X, \Pi_b=R_X}{a,b\geq0, \|\cdot\|_1\leq 1}} \inf_{\underset{\Pi_{\tilde{C}}=\Pi_{XX^*},\Pi_{\tilde{D}}=\Pi_{X^*X}}{\tilde{C},\tilde{D}\geq0}}\!\!\!\!\!\!\!\!\!\!{\eta}\|a_1^\frac{1}{2r}\tilde{C}a_1^\frac{1}{2r}\|_{(\cS^\prime)}+\eta\|b_1^\frac{1}{2r}\tilde{D}b_1^\frac{1}{2r}\|_{(\cS^\prime)}\\&\qquad+\frac{\eta}{\eta_2}\|\tilde{C}^{-\frac{1}{2}}X\tilde{D}^{-\frac{1}{2}}\|^{s}_{s} \Bigg\}^\frac{1}{2\eta} 
\end{align} where we set $A=aC$ and replaced $a$ by $a^\frac{1}{2r}$ and likewise for $B=Db$. By assumptions on the supports of $a_1,b_1$ it follows that $C=a_1^{-1}A$ with $\Pi_{CC^*}=\Pi_{XX^*}$ and likewise for $D$. In the last line we used \Cref{lem:Inf.Prod.to.Sum} with $\eta_2:=\frac{s}{2}$ and $\eta^{-1}:=2+\eta_2^{-1}$. Since $s\geq q_2, \, s\geq q_1$ we have  $\frac{1}{r}=\frac{1}{q_2}-\frac{1}{q_1} \leq \frac{1}{q_2}-\frac{1}{s}= \frac{1}{s_2}\leq 1$, so $1\leq s_2\leq r$. Thus by \Cref{lem.app.op.concavity} we get concavity of the first two summands in $a,b$ and convexity in $\tilde{C}, \tilde{D}$ as well as suitable continuity and by \cite[Theorem 1.1 (2)]{Zhang.2020} we get that $\|\tilde{C}^{-\frac{1}{2}}X\tilde{D}^{-\frac{1}{2}}\|^{s}_{s}$ is jointly convex and continuous in $\tilde{C}, \tilde{D}$ on $\Pi_{\tilde{C}}=\Pi_{XX^*}, \Pi_{\tilde{D}}=\Pi_{X^*X}$. Thus the whole expression is jointly concave and upper semi-continuous in $a,b$, which live in a compact set, and jointly convex in $\tilde{C}, \tilde{D}$ so Sion's minimax theorem \Cref{thm:Sion} applies and we can swap the infimum and supremum. This yields  
\begin{align}
    \|X\|_{(\cQ)} &= \Bigg\{\!\inf_{\underset{\Pi_{\tilde{C}}=\Pi_{XX^*}, \Pi_{\tilde{D}}=\Pi_{X^*X}}{\tilde{C},\tilde{D}\geq0}}\!\sup_{\underset{\Pi_a=L_X, \Pi_b=R_X}{a,b\geq0, \|\cdot\|_1\leq 1}}\!\!\eta\|a_1^\frac{1}{2r}\tilde{C}a_1^\frac{1}{2r}\|_{(\cS^\prime)}+\eta\|b_1^\frac{1}{2r}\tilde{D}b_1^\frac{1}{2r}\|_{(\cS^\prime)}\\&\qquad+\frac{\eta}{\eta_2}\|\tilde{C}^{-\frac{1}{2}}X\tilde{D}^{-\frac{1}{2}}\|^{s}_{s} \Bigg\}^\frac{1}{2\eta} \\
    &= \left\{\inf_{\underset{}{\tilde{C},\tilde{D}\geq0}}\eta\|\tilde{C}\|_{(\cS)}+\eta\|\tilde{D}\|_{(\cS)}+\frac{\eta}{\eta_2}\|\tilde{C}^{-\frac{1}{2}}X\tilde{D}^{-\frac{1}{2}}\|^{s}_{s} \right\}^\frac{1}{2\eta} \\
    &= \inf_{\tilde{C},\tilde{D}\geq0}\|\tilde{C}\|_{(\cS)}^\frac{1}{2}\|\tilde{D}\|_{(\cS)}^\frac{1}{2}\|\tilde{C}^{-\frac{1}{2}}X\tilde{D}^{-\frac{1}{2}}\|_{s},
\end{align} where $\cS=(s_1,s_2,\cdots,s_n)$ with $\frac{1}{r}=\frac{1}{q_2}-\frac{1}{q_1}=\frac{1}{s_2}-\frac{1}{s_1}$, which follows from \Cref{thm:VariationalSpOperatorspace}. 

\medskip

The supremum formulas follow similarly. In the case $q_1\geq q_2$, proceeding analogously to above we get a variational formula involving two suprema, which can readily be swapped yielding the claimed statement. The more interesting part is now the $q_1\leq q_2$ part which we spell out below and which requires the condition $t\leq 1$ to work. Write $(\cT^\prime)=(t_2,t_2,t_3,\dots, t_n)$ with $\frac{1}{t_i}=\frac{1}{t}-\frac{1}{q_i}$, $\cQ^\prime=(q_2,q_2,q_3,\cdots, q_n)$, and set $\frac{1}{r}=\frac{1}{q_1}-\frac{1}{q_2}$. For notational simplicity we drop writing the support conditions in the following. We have
\begin{align}
    \|X\|_{(\cQ)} &= \inf_{\underset{\Pi_a=L_X,\, \Pi_b=R_X}{a,b\geq0}}\|a\|_{2r}\|b\|_{2r}\|a^{-1}_1Xb_1^{-1}\|_{(\cQ^\prime)} \\
    &= \inf_{\underset{\Pi_a=L_X,\, \Pi_b=R_X}{a,b\geq0}}\sup_{A,B}\|a\|_{2r}\|b\|_{2r}\|A^*A\|^{-\frac{1}{2}}_{(\cT^\prime)}\|BB^*\|^{-\frac{1}{2}}_{(\cT^\prime)}\|A_{1\dots n}a_1^{-1}X_{1\dots n}b_1^{-1}B_{1\dots n}\|_t \\ 
    &=\inf_{\underset{\Pi_a=L_X,\, \Pi_b=R_X}{a,b\geq0, \,\|\cdot\|_1\leq1}}\sup_{\tilde{C},\tilde{D}\geq0}\|a^\frac{1}{2r}\tilde{C}a^\frac{1}{2r}\|^{-\frac{1}{2}}_{(\cT^\prime)}\|b^\frac{1}{2r}\tilde{D}b^\frac{1}{2r}\|^{-\frac{1}{2}}_{(\cT^\prime)}\|\tilde{C}^\frac{1}{2}X\tilde{D}^\frac{1}{2}\|_t \\
    &= \left\{\inf_{\underset{\Pi_a=L_X,\, \Pi_b=R_X}{a,b\geq0, \,\|\cdot\|_1\leq1}}\!\sup_{\tilde{C},\tilde{D}\geq0} \eta\|a^\frac{1}{2r}\tilde{C}a^\frac{1}{2r}\|_{(\cT^\prime)}\!+\!\eta\|b^\frac{1}{2r}\tilde{D}b^\frac{1}{2r}\|_{(\cT^\prime)} \!+\!\frac{\eta}{\eta_2} \|\tilde{C}^\frac{1}{2}X\tilde{D}^\frac{1}{2}\|^\gamma_t \right\}^{-\frac{1}{2\eta}}
\end{align} where we fixed $0<\gamma=\frac{t}{2}<t$, set $\eta_2=-\frac{\gamma}{2},$ and $\eta^{-1}=2+\eta_2^{-1}= 2-2\gamma^{-1}$. Then for any $q_2>0$  we have $\eta<0$ since $0<\gamma<t\leq 1$.
By the assumptions, similar to before we have
\begin{align}
    \frac{1}{r}=\frac{1}{q_1}-\frac{1}{q_2} \leq \frac{1}{t}-\frac{1}{q_2}=\frac{1}{t_2}\leq 1,
\end{align} so $1\leq t_2\leq r$ and  \Cref{lem.app.op.concavity} applies yielding that the first two summands are convex in $a,b$ and concave in $\tilde{C},\tilde{D}$ as well as suitably semicontinuous. The last summand is continuous and jointly concave in $\tilde{C},\tilde{D}$ as a consequence of the continuity, joint concavity, and positivity of $\|\tilde{C}^\frac{1}{2}X\tilde{D}^\frac{1}{2}\|^t_t$, which follows from \cite[Theorem 1.1 (1)]{Zhang.2020} when $t\leq 1$ and the fact that the map $j_{\frac{\gamma}{t}}:x\mapsto x^\frac{\gamma}{t}$ is monotonically increasing, continuous and concave on $[0,\infty]$. Thus their composition is also jointly concave and continuous. This allows the application of Sion's minimax theorem \Cref{thm:Sion} which yields
\begin{align}
   \|X\|_{(\cQ)} &= \left\{\sup_{\tilde{C},\tilde{D}\geq0}\inf_{\underset{\Pi_a=L_X,\, \Pi_b=R_X}{a,b\geq0, \,\|\cdot\|_1\leq1}} \!\eta\|a^\frac{1}{2r}\tilde{C}a^\frac{1}{2r}\|_{(\cT^\prime)}\!+\!\eta\|b^\frac{1}{2r}\tilde{D}b^\frac{1}{2r}\|_{(\cT^\prime)} \!+\!\frac{\eta}{\eta_2} \|\tilde{C}^\frac{1}{2}X\tilde{D}^\frac{1}{2}\|^\gamma_t \right\}^{-\frac{1}{2\eta}} \\
    &= \left\{\sup_{\tilde{C},\tilde{D}\geq0}\eta\|\tilde{C}\|_{(\cT)}+\eta\|\tilde{D}\|_{(\cT)} +\frac{\eta}{\eta_2} \|\tilde{C}^\frac{1}{2}X\tilde{D}^\frac{1}{2}\|^\gamma_t \right\}^{-\frac{1}{2\eta}} \\
   &=\sup_{\tilde{C},\tilde{D}\geq0} \|\tilde{C}\|^{-\frac{1}{2}}_{(\cT)}\|\tilde{D}\|^{-\frac{1}{2}}_{(\cT)}\|\tilde{C}^\frac{1}{2}X\tilde{D}^\frac{1}{2}\|_t,
\end{align} where we first used that
\begin{align}
    \inf_{\underset{\Pi_a=L_X,}{a\geq0, \,\|a\|_1\leq1}} \eta\|a_1^\frac{1}{2r}\tilde{C}a_1^\frac{1}{2r}\|_{(\cT^\prime)} &= \eta\left(\sup_{\underset{\Pi_a=L_X,}{a\geq0, \,\|a\|_1\leq1}} \|a_1^\frac{1}{2r}\tilde{C}a_1^\frac{1}{2r}\|_{(\cT^\prime)}\right) = \eta\|\tilde{C}\|_{(\cT)}\,,
\end{align} since $\frac{1}{r}=\frac{1}{q_1}-\frac{1}{q_2}=\frac{1}{t_2}-\frac{1}{t_1}$ ($\frac{1}{t_i}=\frac{1}{t}-\frac{1}{q_i}$ and thus $t_1\geq t_2$). Here $\cT=(t_1,t_2,\dots,t_n)$ and $\cT^\prime=(t_2,t_2,t_3,\dots,t_n)$. Lastly we used \Cref{lem:Sup.Prod.to.Sum} in reverse again. 

\medskip

\noindent Let us now prove the simplification in the case of self-adjoint $X=X^*$. 
For the infimum case, first recall that $\|\cdot\|_{(\cS)}$ defines a norm, since it coincides with Pisier's operator-valued Schatten norm as all the indices $s_i\geq1$. 
Now we proceed as in \cite[Proof of Lemma 3.2]{Kochanowski.2026} by first observing that due to a simple scaling argument we may restrict the optimization to $A,B$ s.t. $\|AA^*\|_{(\cS)}=\|B^*B\|_{(\cS)}\leq 1$. This uses the positive homogeneity of $\cdot \mapsto \|\cdot\|_{(\cS)}$. Now, given admissible $A,B$, set $\tilde{A}:=AA^*, \tilde{B}:=B^*B$. Due to unitary invariance of the Schatten-$s$-quasi-norm we thus have
\begin{align}
     \|X\|^s_{(\cQ)}&=\inf_{\underset{\Pi_{AA^*}=\Pi_{XX^*}= \Pi_{B^*B}}{A,B}}\|AA^*\|^\frac{s}{2}_{(\cS)}\|B^*B\|^\frac{s}{2}_{(\cS)}\|A^{-1}XB^{-1}\|^s_{s} \\&=\inf_{\underset{\Pi_{\tilde{A}}=\Pi_{XX^*}=\Pi_{\tilde{B}}}{\tilde{A},\tilde{B}, \, \|\tilde{A}\|_{(\cS)}=\|\tilde{B}\|_{(\cS)}\leq 1}} \underbrace{\|\tilde{A}^{-\frac{1}{2}}X\tilde{B}^{-\frac{1}{2}}\|^s_s}_{=:G(\tilde{A},\tilde{B})},
\end{align} where $G(\tilde{A},\tilde{B})$ is jointly convex in $\tilde{A},\tilde{B}\geq0$ by \cite[Theorem 1.1 (2)]{Zhang.2020} and $G(\tilde{A},\tilde{B})= G(\tilde{B},\tilde{A})$ by self-adjointness of $X=X^*$. Let $\tilde{C}:=\frac{\tilde{A}+\tilde{B}}{2}$, then $\Pi_{\tilde{C}}\geq \Pi_{\tilde{A}}, \Pi_{\tilde{B}}$. Then on the one hand using the triangle inequality we get $\|\tilde{C}\|_{(\cS)}\leq \frac{1}{2}( \|\tilde{A}\|_{(\cS)}+\|\tilde{B}\|_{(\cS)})\leq 1$.
On the other, by symmetry and joint convexity of $G$ we get
\begin{align}
    G(\tilde{C},\tilde{C}) = G\left(\frac{\tilde{A}+\tilde{B}}{2},\frac{\tilde{B}+\tilde{A}}{2}\right) \leq \frac{1}{2}\left(G(\tilde{A},\tilde{B})+G(\tilde{B},\tilde{A})\right)= G(\tilde{A},\tilde{B}).
\end{align} Combining everything yields 
\begin{align}
\inf_{\underset{\Pi_{\tilde{C}}=\Pi_X}{\tilde{C}, \, \|\tilde{C}\|_{(\cS)}\leq 1}} G(\tilde{C},\tilde{C})\leq 
   \inf_{\underset{\Pi_{\tilde{A}}=\Pi_X=\Pi_{\tilde{B}}}{\tilde A,\tilde B, \, \|\tilde{A}\|_{(\cS)}=\|\tilde{B}\|_{(\cS)}\leq 1}} G(\tilde{A},\tilde{B})
\end{align} which proves the claimed statement after taking the $s$-root.

\medskip

\noindent For the supremum case we need a novel argument compared to \cite{Devetak.2006}, which explicitly requires $X\geq0$. We will treat it similarly to above and proceed  by first observing that due to a simple scaling argument we may restrict the optimization to $A,B$ s.t. $\|A^*A\|_{(\cT)}=\|BB^*\|_{(\cT)}\leq 1$. This uses the positive homogeneity of $\cdot \mapsto \|\cdot\|_{(\cT)}$. Now, given admissible $A,B$, set $\tilde{A}:=A^*A, \tilde{B}:=BB^*$. Due to unitary invariance of the Schatten-$t$-quasi-norm we thus have
\begin{align}
     \|X\|^t_{(\cQ)}&=\sup_{\underset{\Pi_{A^*A}=\Pi_{XX^*}= \Pi_{BB^*}}{A,B}}\|A^*A\|^{-\frac{t}{2}}_{(\cT)}\|BB^*\|^{-\frac{t}{2}}_{(\cT)}\|AXB\|^t_{t} \\&=\sup_{\underset{\Pi_{\tilde{A}}=\Pi_{XX^*}=\Pi_{\tilde{B}}}{\tilde{A},\tilde{B}, \, \|\tilde{A}\|_{(\cT)}=\|\tilde{B}\|_{(\cT)}\leq 1}} \underbrace{\|\tilde{A}^{\frac{1}{2}}X\tilde{B}^{\frac{1}{2}}\|^t_t}_{=:H(\tilde{A},\tilde{B})},
\end{align} where $H(\tilde{A},\tilde{B})$ is jointly concave in $\tilde{A},\tilde{B}\geq0$ by \cite[Theorem 1.1 (1)]{Zhang.2020} since $t\leq 1$ and $H(\tilde{A},\tilde{B})= H(\tilde{B},\tilde{A})$ by self-adjointness of $X=X^*$ as above. Let $\tilde{C}:=\frac{\tilde{A}+\tilde{B}}{2}$, then $\Pi_{\tilde{C}}\geq \Pi_{\tilde{A}}, \Pi_{\tilde{B}}$ and
just as before, using the triangle inequality we get $\|\tilde{C}\|_{(\cT)}\leq \frac{1}{2}( \|\tilde{A}\|_{(\cT)}+\|\tilde{B}\|_{(\cT)})\leq 1$.
On the other, by symmetry and joint concavity of $H$ we get
\begin{align}
    H(\tilde{C},\tilde{C}) = H\left(\frac{\tilde{A}+\tilde{B}}{2},\frac{\tilde{B}+\tilde{A}}{2}\right) \geq \frac{1}{2}\left(H(\tilde{A},\tilde{B})+H(\tilde{B},\tilde{A})\right)= H(\tilde{A},\tilde{B}).
\end{align}
Combining everything yields 
\begin{align}
\sup_{\underset{\Pi_{\tilde{C}}=\Pi_X}{\tilde{C}, \, \|\tilde{C}\|_{(\cT)}\leq 1}} H(\tilde{C},\tilde{C})\geq 
   \sup_{\underset{\Pi_{\tilde{A}}=\Pi_X=\Pi_{\tilde{B}}}{\tilde A,\tilde B, \, \|\tilde{A}\|_{(\cT)}=\|\tilde{B}\|_{(\cT)}\leq 1}} H(\tilde{A},\tilde{B})
\end{align} which proves the claimed statement after taking the $t$-root.

\medskip

Having proven \eqref{equ:Multi.Thm.Pisier.like.inf} for $s$ compatible with $\cQ=(q_1,\dots, q_n)$ we now extend it to arbitrary $q_i\leq s$. Let $r:=\max_i\{q_i\}\leq s$, then $r$ is compatible with $\cQ$ and by the above proof we have for $\cR=(r_1,\dots, r_n)$ with $\frac{1}{r_i}=\frac{1}{q_i}-\frac{1}{r}\leq 1$ that
\begin{align}
    \|X\|_{(\cQ)}&=\inf_{\underset{\Pi_{\tilde A}=\Pi_{XX^*}, \Pi_{\tilde B}=\Pi_{X^*X}}{\tilde A,\tilde B\geq0}}\|\tilde A\|^\frac{1}{2}_{(\cR)}\|\tilde B\|^\frac{1}{2}_{(\cR)}\|\tilde A^{-\frac12}X\tilde B^{-\frac12}\|_{r} 
    \\&=\inf_{\underset{\Pi_{\tilde A}=\Pi_{XX^*}, \Pi_{\tilde B}=\Pi_{X^*X}}{\tilde A,\tilde B\geq0}}\|\tilde A\|^\frac{1}{2}_{(\cR)}\|\tilde B\|^\frac{1}{2}_{(\cR)}\inf_{\underset{\Pi_a={\Pi_{XX^*}}, \Pi_{b}=\Pi_{X^*X}}{a,b\geq0}}\|a\|_{2u}\|b\|_{2u}\|a^{-1}\tilde A^{-\frac12}X\tilde B^{-\frac12}b^{-1}\|_{s} \\
    &= \inf_{\underset{\Pi_{\tilde{A}}=\Pi_{CC^*}=\Pi_{XX^*},\,\Pi_{\tilde{B}}=\Pi_{DD^*}=\Pi_{X^*X}}{\tilde A,\tilde B\geq0,\, C,D}}\!\!\!\!\!\!\!\!\!\!\!\!\!\!\!\|\tilde A\|^\frac{1}{2}_{(\cR)}\|\tilde B\|^\frac{1}{2}_{(\cR)}\|\tilde{A}^{-\frac12}CC^*\tilde{A}^{-\frac12}\|^\frac{1}{2}_u\|\tilde{B}^{-\frac12}D^*D\tilde{B}^{-\frac12}\|^\frac{1}{2}_u\|C^{-1}XD^{-1}\|_s \\
    &= \inf_{\underset{\Pi_{CC^*}=\Pi_{XX^*}, \, \Pi_{D^*D}=\Pi_{X^*X}}{C,D}}\|CC^*\|^\frac{1}{2}_{(\cS)}\|D^*D\|^\frac{1}{2}_{(\cS)}\|C^{-1}XD^{-1}\|_s 
\end{align} where in the second line we used \cref{prop:SpVariationalFormulas} with $\frac{1}{u}=\frac{1}{r}-\frac{1}{s}$. In the third we set $C^{-1}=a^{-1}\tilde{A}^{-\frac12}$ and $D^{-1}=\tilde{B}^{-\frac12}b^{-1}$ which implies, respectively, $aa^*=\tilde{A}^{-\frac12}CC^*\tilde{A}^{-\frac12}$ and $b^*b=\tilde{B}^{-\frac12}D^*D\tilde{B}^{-\frac12}$ and $\Pi_{CC^*}=\Pi_{XX^*}$, $\Pi_{D^*D}=\Pi_{X^*X}$. The last line follows by another use of the previously proven version for self-adjoint $CC^*, D^*D$ and compatible indices in the form of
\begin{align}
   \inf_{\underset{\Pi_{\tilde{A}}=\Pi_{CC^*}}{\tilde{A}\geq0}} \|\tilde A\|^\frac{1}{2}_{(\cR)}\|\tilde{A}^{-\frac12}CC^*\tilde{A}^{-\frac12}\|^\frac{1}{2}_u = \|CC^*\|^\frac{1}{2}_{(\cS)},
\end{align} where we used that $\cS=(s_1,\dots, s_n)$ with $\frac{1}{s_i}=\frac{1}{q_i}-\frac{1}{s}$ is compatible with $u$ since $[0,1]\ni \frac{1}{r_i}=\frac{1}{q_i}-\frac{1}{r} = \frac{1}{s_i}-\frac{1}{u}$.

The version of \eqref{equ:Multi.Thm.Pisier.like.sup} for incompatible $t\leq q_i,\, t\leq 1$ follows similarly. Let $r:=\min\{1,\min_i\{q_i\}\}\geq t$, then $r$ is compatible with $\cQ$ and by the above proof we have for $\cR=(r_1,\dots, r_n)$ with $\frac{1}{r_i}=\frac{1}{r}-\frac{1}{q_i}\leq 1$ that
\begin{align}
    \|X\|_{(\cQ)}&=\sup_{\underset{\Pi_{\tilde A}=\Pi_{XX^*}, \Pi_{\tilde B}=\Pi_{X^*X}}{\tilde A,\tilde B\geq0}}\|\tilde A\|^{-\frac{1}{2}}_{(\cR)}\|\tilde B\|^{-\frac{1}{2}}_{(\cR)}\|\tilde A^{\frac12}X\tilde B^{\frac12}\|_{r} \\
    &= \sup_{\underset{\Pi_{\tilde{A}}=\Pi_{C^*C}=\Pi_{XX^*},\,\Pi_{\tilde{B}}=\Pi_{DD^*}=\Pi_{X^*X}}{\tilde A,\tilde B\geq0,\, C,D}}\!\!\!\!\!\!\!\!\!\!\!\!\!\!\!\|\tilde A\|^{-\frac{1}{2}}_{(\cR)}\|\tilde B\|^{-\frac{1}{2}}_{(\cR)}\|\tilde{A}^{-\frac12}C^*C\tilde{A}^{-\frac12}\|^{-\frac{1}{2}}_u\|\tilde{B}^{-\frac12}DD^*\tilde{B}^{-\frac12}\|^{-\frac{1}{2}}_u\|CXD\|_t \\
    &= \sup_{\underset{\Pi_{C^*C}=\Pi_{XX^*}, \, \Pi_{DD^*}=\Pi_{X^*X}}{C,D}}\|C^*C\|^{-\frac{1}{2}}_{(\cT)}\|DD^*\|^{-\frac{1}{2}}_{(\cT)}\|CXD\|_t 
\end{align} where in the second line we used \cref{prop:SpVariationalFormulas} with $\frac{1}{u}=\frac{1}{t}-\frac{1}{r}$. In the third we set $C=a\tilde{A}^{\frac12}$ and $D=\tilde{B}^{\frac12}b$ which implies, respectively, $a^*a=\tilde{A}^{-\frac12}C^*C\tilde{A}^{-\frac12}$ and $bb^*=\tilde{B}^{-\frac12}DD^*\tilde{B}^{-\frac12}$ and $\Pi_{C^*C}=\Pi_{XX^*}$, $\Pi_{DD^*}=\Pi_{X^*X}$. The last line follows by another use of the previously proven version for self-adjoint $C^*C, DD^*$ and compatible indices in the form of
\begin{align}
   \sup_{\underset{\Pi_{\tilde{A}}=\Pi_{C^*C}}{\tilde{A}\geq0}} \|\tilde A\|^{-\frac{1}{2}}_{(\cR)}\|\tilde{A}^{-\frac12}C^*C\tilde{A}^{-\frac12}\|^{-\frac{1}{2}}_u = \bigg(\inf_{\underset{\Pi_{\tilde{A}}=\Pi_{C^*C}}{\tilde{A}\geq0}} \|\tilde A\|_{(\cR)}\|\tilde{A}^{-\frac12}C^*C\tilde{A}^{-\frac12}\|_u\bigg)^{-\frac{1}{2}} = \|C^*C\|^{-\frac{1}{2}}_{(\cT)},
\end{align} where we used that $\cT=(t_1,\dots, t_n)$ with $\frac{1}{t_i}=\frac{1}{t}-\frac{1}{q_i}$ is compatible with $u$ since $[0,1]\ni \frac{1}{r_i}=\frac{1}{r}-\frac{1}{q_i} = \frac{1}{t_i}-\frac{1}{u}$.

The symmetrization argument clearly also applies in these cases, since for $X=X^*$ we get $\tilde{A}=\tilde{B}$ in both cases and by \cref{prop:SpVariationalFormulas} $a=b^*$.

\medskip

It remains to justify \Cref{rem:sup.arbitrary}, i.e.\ that
\begin{align}\label{equ:sup.arbitrary}
    \|AXB\|_t\leq\|A^*A\|^{\frac12}_{(\cT)}\|BB^*\|^{\frac12}_{(\cT)}\|X\|_{(\cQ)}
\end{align}
holds for all $X$ and all operators $A,B$. We first note that for every multi-indexed Schatten space $\cP$ with positive indices and every $P\geq0$ we have $\limsup_{\delta\downarrow0}\|P+\delta\1\|_{(\cP)}\leq\|P\|_{(\cP)}$. Indeed, let $s=\max_ip_i$, so that the associated collection $\cS$ is normed. For $\varepsilon>0$, the infimum formula \eqref{equ:Multi.Thm.Pisier.like.inf} with symmetric optimizers gives some $\tilde A\geq0$ with $\Pi_{\tilde A}=\Pi_P$ and $\|\tilde A\|_{(\cS)}\|\tilde A^{-\frac12}P\tilde A^{-\frac12}\|_s\leq\|P\|_{(\cP)}+\varepsilon$. The full-rank operators $\tilde A_\delta:=\tilde A+\sqrt\delta(\1-\Pi_P)$ are admissible for $P+\delta\1$, and
\begin{align}
    \|P+\delta\1\|_{(\cP)}&\leq\|\tilde A_\delta\|_{(\cS)}\big\|\tilde A^{-\frac12}P\tilde A^{-\frac12}+\delta\tilde A^{-1}+\sqrt\delta(\1-\Pi_P)\big\|_s\\&\xrightarrow{\delta\to0}\|\tilde A\|_{(\cS)}\|\tilde A^{-\frac12}P\tilde A^{-\frac12}\|_s\leq\|P\|_{(\cP)}+\varepsilon.
\end{align}

\noindent\emph{Positive operators.} Let $P\geq0$ and let $A$ be arbitrary, and set $Z:=A^*A$. Since $APA^*$ and $P^{\frac12}ZP^{\frac12}$ have the same non-zero spectrum, and $\|\cdot\|_t$ is monotone on positive operators, we get for $\varepsilon,\delta>0$
\begin{align}
    \|APA^*\|_t&=\|P^{\frac12}ZP^{\frac12}\|_t\leq\|(Z+\varepsilon\1)^{\frac12}P(Z+\varepsilon\1)^{\frac12}\|_t\\&\leq\|(Z+\varepsilon\1)^{\frac12}(P+\delta\1)(Z+\varepsilon\1)^{\frac12}\|_t\leq\|Z+\varepsilon\1\|_{(\cT)}\|P+\delta\1\|_{(\cQ)}.
\end{align}
The last step is \eqref{equ:Multi.Thm.Pisier.like.sup} for the full-rank operator $P+\delta\1$, with the admissible full-rank choice $A=B=(Z+\varepsilon\1)^{\frac12}$. Letting $\delta,\varepsilon\to0$ and using the upper semicontinuity above for $\cP=\cQ$ and $\cP=\cT$ yields
\begin{align}\label{equ:sup.arbitrary.pos}
    \|APA^*\|_t\leq\|A^*A\|_{(\cT)}\|P\|_{(\cQ)}\qquad\text{for all }P\geq0\text{ and all }A.
\end{align}

\noindent\emph{General operators.} Let $X$ be arbitrary. By the infimum formula \eqref{equ:Multi.Thm.Pisier.like.inf} with $s=\infty$ gives $\cS=\cQ$, for every $\varepsilon>0$ there are $C,D$ with $X=CYD$, $Y=C^{-1}XD^{-1}$, and $\|CC^*\|^{\frac12}_{(\cQ)}\|D^*D\|^{\frac12}_{(\cQ)}\|Y\|_\infty\leq\|X\|_{(\cQ)}+\varepsilon$. The generalized Hölder inequality with $\frac1t=\frac1{2t}+\frac1\infty+\frac1{2t}$ and \eqref{equ:sup.arbitrary.pos}, applied to $P=CC^*$ with $A$ and to $P=D^*D$ with $B^*$, give
\begin{align}
    \|AXB\|_t\leq\|AC\|_{2t}\|Y\|_\infty\|DB\|_{2t}&=\|ACC^*A^*\|^{\frac12}_t\|B^*D^*DB\|^{\frac12}_t\|Y\|_\infty\\&\leq\|A^*A\|^{\frac12}_{(\cT)}\|BB^*\|^{\frac12}_{(\cT)}\big(\|X\|_{(\cQ)}+\varepsilon\big),
\end{align}
and $\varepsilon\to0$ proves \eqref{equ:sup.arbitrary}.
\end{proof}

Another structural factorization result we can derive as a consequence of the previous \Cref{thm:Multi.Thm.Pisier.like} is the following, which should be regarded as a multi-indexed extension of \cite[Theorem 1.3]{Kochanowski.2026} and as a quasi-norm extension of \Cref{thm:VariationalSpOperatorspace}.

\begin{theorem}[Multi-indexed relational consistency]\label{thm:multi.quasi.relational.consistency}
    Let $\cQ$ be a multi-indexed Schatten space and $q\leq p$ compatible indices which are compatible with $\cQ$, then with $\frac{1}{r}=\frac{1}{q}-\frac{1}{p}$ it holds that
    \begin{align}
        \|X\|_{(q,\cQ)} &= \inf_{\underset{\Pi_a=L_X,\, \Pi_b=R_X}{a,b\geq0}}\|a\|_{2r}\|b\|_{2r}\|a_1^{-1}Xb_1^{-1}\|_{(p,\cQ)}, \\
         \|X\|_{(p,\cQ)} &= \sup_{\underset{\Pi_a=L_X,\, \Pi_b=R_X}{a,b\geq0}}\|a\|^{-1}_{2r}\|b\|^{-1}_{2r}\|a_1Xb_1\|_{(q,\cQ)},
    \end{align} 
    where we suppressed explicitly writing $\1$-operators and $a^{\mp1}_1, b^{\mp1}_1$ act only on the first subsystem. When $q,p,\cQ$ are negative indices we assume $X>0$.  \newline 
    Moreover, when $X=X^*$ one may take $a=b^*$.
\end{theorem}

\begin{remark}
In the single-indexed special case $(\cQ)=(q)$ this recovers \cite[Theorem 1.3]{Kochanowski.2026} and when $(\cQ)$ is some normed multi-indexed Schatten space and $q,p\in[1,\infty]$ it reproduces \Cref{thm:VariationalSpOperatorspace}. Essentially it tells us that as far as factorization formulas go, quasi-norms satisfy the same structural properties as Pisier's operator-valued Schatten norms.
\end{remark}

\begin{proof}
   We will prove this as a consequence of both \Cref{thm:Multi.Thm.Pisier.like} and the factorization formulas for single-indexed operator-valued Schatten spaces. Take $\cQ,q,p$ from the theorem. Without loss of generality we may assume all indices $q,p,\cQ$ to be positive. In the negative case the result follows directly via \Cref{prop:negative.multi.indexed}. Let $\frac{1}{r}=\frac{1}{q}-\frac{1}{p}$, $\frac{1}{r_q}=\frac{1}{q}-\frac{1}{s}$, and $\frac{1}{r_p}=\frac{1}{p}-\frac{1}{s}$ for some suitable compatible $s\geq q,p,\cQ$. Since $p,q$ are compatible with $\cQ$ it follows that $s$ can be chosen compatible with $\cQ$ as well and hence $\cR=(r_i)_i$ with $\frac{1}{r_i}=\frac{1}{q_i}-\frac{1}{s}\in[0,1]$ is normed. Then by assumption $1\leq r_q\leq r_p\leq \infty$ and by \Cref{thm:VariationalSpOperatorspace}
   \begin{align}
      \|AA^*\|_{(r_q,\cR)} = \inf_{\underset{\Pi_a=L_{AA^*}=R_{AA^*}}{a\geq0} } \|a\|^2_{2r} \|a_1^{-1}AA^*a^{-1}_1\|_{(r_p,\cR)}.
   \end{align}
   This, together with \Cref{thm:Multi.Thm.Pisier.like} implies 
   \begin{align}
       \|X\|_{(q,\cQ)} &= \inf_{A,B}\|AA^*\|^\frac{1}{2}_{(r_q,\cR)}\|B^*B\|^\frac{1}{2}_{(r_q,\cR)} \|A^{-1}XB^{-1}\|_s \\
       &=\inf_{A,B}\inf_{a\geq0} \|a\|_{2r} \|a_1^{-1}AA^*a^{-1}_1\|^\frac{1}{2}_{(r_p,\cR)}\inf_{b\geq0} \|b\|_{2r} \|b_1^{-1}B^*Bb^{-1}_1\|^\frac{1}{2}_{(r_p,\cR)} \|A^{-1}XB^{-1}\|_s \\
       &= \inf_{a,b\geq0}\|a\|_{2r}\|b\|_{2r}\inf_{C,D}\|CC^*\|^\frac{1}{2}_{(r_p,\cR)}\|D^*D\|^\frac{1}{2}_{(r_p,\cR)}\|C^{-1}a_1^{-1}Xb_1^{-1}D^{-1}\|_{s} \\
       &= \inf_{a,b\geq0}\|a\|_{2r}\|b\|_{2r}\|a_1^{-1}Xb_1^{-1}\|_{(p,\cQ)},
   \end{align} which proves the first equality. Here for notational simplicity we dropped writing the support conditions explicitly and in the second to last line defined $a_1^{-1}A=C$.
   For the other equality take a suitable $0<t\leq 1$ with $t\leq q,p,\cQ,1$ and define $\frac{1}{r^\prime_q}=\frac{1}{t}-\frac{1}{q}$ and $\frac{1}{r^\prime_p}=\frac{1}{t}-\frac{1}{p}$, then $r^\prime_p\leq r^\prime_q$ and as above, by \Cref{thm:VariationalSpOperatorspace} with $\cR^\prime=(r^\prime_i)_i$ with $\frac{1}{r^\prime_i}=\frac{1}{t}-\frac{1}{q_i}$ we have
    \begin{align}
      \|A^*A\|^{-1}_{(r^\prime_p,\cR^\prime)} = \left(\inf_{\underset{\Pi_a={L_{A^*A}}}{a\geq0}}\!\!\! \|a\|^{2}_{2r} \|a^{-1}_1A^*Aa^{-1}_1\|_{(r^\prime_q,\cR^\prime)}\right)^{\!-1}\!\!\! = \!\!\!\sup_{\underset{\Pi_a={L_{A^*A}}}{a\geq0}}\!\!\!\|a\|^{-2}_{2r}\|a^{-1}_1A^*Aa^{-1}_1\|^{-1}_{(r^\prime_q,\cR^\prime)},
   \end{align} which together with the supremum formula of \Cref{thm:Multi.Thm.Pisier.like} implies
   \begin{align}
       \|X\|_{(p,\cQ)} &= \sup_{A,B}\|A^*A\|^{-\frac{1}{2}}_{(r^\prime_p,\cR^\prime)}\|BB^*\|^{-\frac{1}{2}}_{(r^\prime_p,\cR^\prime)} \|AXB\|_t  \\ &= \sup_{A,B}\sup_{a\geq0}\|a\|^{-1}_{2r}\|a^{-1}_1A^*Aa^{-1}_1\|^{-\frac{1}{2}}_{(r^\prime_q,\cR^\prime)}\sup_{b\geq0}\|b\|^{-1}_{2r}\|b^{-1}_1BB^*b^{-1}_1\|^{-\frac{1}{2}}_{(r^\prime_q,\cR^\prime)} \|AXB\|_t \\
       &= \sup_{a,b\geq0} \|a\|^{-1}_{2r}\|b\|^{-1}_{2r} \sup_{C,D}\|C^*C\|^{-\frac{1}{2}}_{(r^\prime_q,\cR^\prime)}\|DD^*\|^{-\frac{1}{2}}_{(r^\prime_q,\cR^\prime)}\|Ca_1Xb_1D\|_t \\
       &=\sup_{a,b\geq0} \|a\|^{-1}_{2r}\|b\|^{-1}_{2r} \|a_1Xb_1\|_{(q,\cQ)},
   \end{align} which proves the other equality. Here we defined $C=Aa_1^{-1}, D=b_1^{-1}B$ in the second to last line.

   \noindent When $X=X^*$, then in the above derivation we may take $A=B^*$ by \Cref{thm:Multi.Thm.Pisier.like}, yielding the claim in that case. 
\end{proof}

Having established two foundational variational formulas we consider now some of their consequences in terms of properties of multi-indexed Schatten spaces.

\subsubsection{Some Structural Properties of Multi-Indexed Quasi-Norms}\label{sub:structural}

The first one is the generalization of \cite[Corollary 3.2]{Kochanowski.2026} to the multi-indexed setting.

\begin{corollary}[Partial traces of multi-indexed norms]\label{cor:partial.trace.1}
Let $\cQ=(q_1,\dots, q_n)$ be some multi-indexed Schatten space compatible with index $1$ and let $X\geq0$, then
\begin{align}
    \|X\|_{(\cQ,1)}= \|\tr_{n+1}[X]\|_{(\cQ)}.
\end{align}
\end{corollary}

\begin{proof}
The proof builds on \Cref{thm:multi.quasi.relational.consistency} and induction over $|\cQ|$. For $|\cQ|=1$ the statement is true by \cite[Corollary 3.2]{Kochanowski.2026}. For $\cQ=(q_1,\dots,q_n)$ set $\cQ^\prime=(q_2,\dots, q_n)$ and observe that by twofold application of \Cref{thm:multi.quasi.relational.consistency} we have
\begin{align}
    \|X\|_{(\cQ,1)} &= \|X\|_{(q_1,\cQ^\prime,1)} = \opt_{a\geq0}\|a\|^{\mp2}_{2r}\|a_1^{\pm1}Xa_1^{\pm1}\|_{(q_2,\cQ^\prime,1)} \\ &\overset{(1)}{=} \opt_{a\geq0}\|a\|^{\mp2}_{2r}\|a_1^{\pm1}\tr_{n+1}[X]a_1^{\pm1}\|_{(q_2,\cQ^\prime)} \\
    &= \|\tr_{n+1}[X]\|_{(q_1,\cQ^\prime)} = \|\tr_{n+1}[X]\|_{(\cQ)}, 
\end{align} where for notational simplicity we omitted writing the support conditions. $(1)$ follows from the induction hypothesis applied to the multi-indexed Schatten space $(q_2,q_3,\dots, q_n)$ of size $n-1$ acting on $\cH_1\otimes\cH_2, \cH_3,\dots, \cH_n$.
\end{proof}

Another important property is Hölder duality for the multi-indexed Schatten norms, i.e.\ for indices in $[1,\infty]$. We get this effectively for free since in the normed regime \Cref{def:GeneralMultiIndexed} coincides with certain operator-valued Schatten norms of Pisier, via \eqref{equ:iterated.pisier}.

\begin{proposition}[Hölder duality \cite{book:Pisier.1998, Beigi.2023}]\label{prop:MultiHölderDual}
      Let $\cQ=(q_1,\dots,q_n)$ be a multi-indexed Schatten space s.t. $1\leq q_i$. Let $q_i^\prime$ be their Hölder dual indices and $\cQ^\prime=(q_1^\prime,\dots,q_n^\prime)$. Then for any $X$ it holds that
    \begin{align}
        \|X\|_{(\cQ)}=\sup_{\underset{}{Y}}|\Tr[Y^*X]|\,\|Y\|^{-1}_{(\cQ^\prime)}
    \end{align}
Note that in the above statement the Hölder dual indices also parametrize a norm as $q_i^\prime\geq 1$.
\end{proposition}

\begin{proof}
    This follows directly from the fact that for the indices $q_i\in[1,\infty]$ the norm $\|X\|_{(\cQ)}$ is an operator-valued Schatten norm of Pisier's by \eqref{equ:iterated.pisier}. For these norms this duality holds by \cite[Lemma 1.8]{book:Pisier.1998}, see also \cite[Proposition 4.3 (ii)]{Beigi.2023}. 
\end{proof}

The next and non-trivial property is the reverse Hölder duality which complements the above regular Hölder duality statement.
Analogously to the two-indexed setting, the reverse Hölder duality will be a direct consequence of \eqref{equ:Multi.Thm.Pisier.like.inf} of \Cref{thm:Multi.Thm.Pisier.like}.

\begin{corollary}[Reverse Hölder duality]\label{cor:reverse.multi.hölder}
    Let $\cQ=(q_1,\dots,q_n)$ be a multi-indexed Schatten space s.t. $q_i\leq 1$. Let $q_i^\prime$ be their Hölder dual indices, and $-\cQ^\prime=(-q_1^\prime,\dots,-q_n^\prime)$. Then for any positive semidefinite $X\geq0$ ($X>0$ when $q_i<0$) it holds that
    \begin{align}
        \|X\|_{(\cQ)}=\inf_{\underset{}{Y>0}}\Tr[Y^*X]\,\|Y^{-1}\|_{(-\cQ^\prime)}.
    \end{align}
In particular, the reverse Hölder inequality
    \begin{align}\label{equ:reverse.multi.hölder.negative.ineq}
        \Tr[YX]\geq \|X\|_{(\cQ)}\,\|Y\|_{(\cQ^\prime)}
    \end{align}
    holds for all $X\geq0, Y>0$ ($X,Y>0$ for $q_i<0$).
    
Note that in the above statement the Hölder dual indices are $q_i^\prime<0$ for $0< q_i\leq 1$ and vice versa.
\end{corollary}

\begin{remark}
This is the multi-indexed and negative-index generalization of \cite[Proposition 1.1]{Kochanowski.2026}.
\end{remark}

\begin{proof}
Consider first the case in which $0<q_i\leq 1$.
Then the statement is a direct consequence of the infimum formula \eqref{equ:Multi.Thm.Pisier.like.inf} of \Cref{thm:Multi.Thm.Pisier.like} for positive operators $X\geq0$ with $s=1$, which need not be compatible with $\cQ$. Explicitly with $\cQ^\prime=(q_1^\prime,\dots, q_n^\prime)$ s.t. $\frac{1}{q_i}+\frac{1}{q_i^\prime}=1$ 
\begin{align}
    \|X\|_{(\cQ)}
    &\overset{}{=}\inf_{\underset{\Pi_{\tilde{A}}=\Pi_X}{\tilde{A}\geq0}}\|\tilde{A}\|_{(-\cQ^\prime)}\|\tilde{A}^{-\frac{1}{2}}X\tilde{A}^{-\frac{1}{2}}\|_1 \\&= \inf_{\underset{\Pi_{\tilde{A}}=\Pi_X}{\tilde{A}\geq0}}\|\tilde{A}\|_{(-\cQ^\prime)}\Tr[\tilde{A}^{-1}X] \\&\overset{}{=} \inf_{\underset{\Pi_{Y^{-1}}=\Pi_X}{Y^{-1}\geq0}}\|Y^{-1}\|_{(-\cQ^\prime)}\Tr[YX].
\end{align} 
interpreting $Y=(Y^{-1})^{-1}=\tilde{A}^{-1}$ as the generalized Moore-Penrose inverse of $\tilde{A}$.

It remains to pass from $Y^{-1}\geq0$ with $\Pi_{Y^{-1}}=\Pi_X$ to arbitrary $Y^{-1}>0$ and so $Y>0$. On the one hand, for $Y_\lambda^{-1}:=Y^{-1}+\lambda^{-1}(\1-\Pi_X)$ $\Pi_{Y^{-1}}=\Pi_X$ satisfies $Y_\lambda=Y+\lambda(\1-\Pi_X)>0$, so $\Tr[Y_\lambda X]=\Tr[YX]$ and $Y_\lambda^{-1}=Y^{-1}+\lambda^{-1}(\1-\Pi_X)\to Y^{-1}$ as $\lambda\to\infty$. Then,  by continuity of the quasi-norm $\|\cdot\|_{(-\cQ^\prime)}$ which will be implies by the independently proven \Cref{thm:Multi.Quasi.Norms}, the infimum over $Y>0$ is not larger. On the other hand, for any $Y>0$ and $\epsilon>0$ the operator $X+\epsilon\1$ has full support, so that the above formula yields $\|X+\epsilon\1\|_{(\cQ)}\leq\Tr[Y(X+\epsilon\1)]\|Y^{-1}\|_{(-\cQ^\prime)}$. Letting $\epsilon\to0$ and using the continuity of $\|\cdot\|_{(\cQ)}$, again by \Cref{thm:Multi.Quasi.Norms}, gives $\|X\|_{(\cQ)}\leq\Tr[YX]\|Y^{-1}\|_{(-\cQ^\prime)}$. This proves the claimed formula, and \eqref{equ:reverse.multi.hölder.negative.ineq} follows with $\|Y^{-1}\|_{(-\cQ^\prime)}=\|Y\|^{-1}_{(\cQ^\prime)}$ by \Cref{prop:negative.multi.indexed} for $Y>0$.

Now the case $q_i<0$ goes similarly, building on \Cref{cor:multi.Pisier.like.negative} and the reverse Hölder duality for Schatten spaces \eqref{equ:reverse.Hölder.Schatten}. Let $X>0$ and take $\cQ=(q_1,\dots, q_n)$ with $q_i<0$ mutually compatible. Set $t:=\max\big\{-1,\max_i q_i\big\}\in[-1,0)$, then $q_i\leq t$ for all $i$, and $t$ is compatible with every $q_i$. Define the multi-indexed Schatten space $\cT=(t_1,\dots,t_n)$ by
    \begin{align}
        \frac{1}{t_i}:=\frac{1}{q_i}-\frac{1}{t}\in[0,1],
    \end{align}
    and let $t^\prime=\frac{t}{t-1}\in(0,\frac{1}{2}]$ be the Hölder dual of $t$.
By \Cref{cor:multi.Pisier.like.negative} we have
\begin{align}
\|X\|_{(\cQ)}&=\inf_{A>0}\|A^2\|_{(\cT)}\,\|A^{-1}XA^{-1}\|_{t}  \\ &=\inf_{A>0}\|A^2\|_{(\cT)}\,\inf_{Z>0}\Tr[ZA^{-1}XA^{-1}]\|Z\|^{-1}_{t^\prime} \\
&= \inf_{Z>0}\inf_{A>0}\|A^2\|_{(\cT)}\|Z\|_{t^\prime}^{-1}\Tr[A^{-1}ZA^{-1}X] \\
&= \inf_{Y>0}\big(\sup_{A>0}\|A^2\|^{-1}_{(\cT)}\|AYA\|_{t^\prime}\Big)^{-1}\Tr[YX] \\&=\inf_{Y>0}\|Y\|_{(\cQ^\prime)}^{-1}\Tr[YX] =\inf_{Y>0}\Tr[Y^*X]\,\|Y^{-1}\|_{(-\cQ^\prime)}
    \end{align}
where the second equality is \eqref{equ:reverse.Hölder.Schatten} and in the fourth we introduced $Y=A^{-1}ZA^{-1} \Leftrightarrow Z=AYA$. The last line follows via \Cref{thm:Multi.Thm.Pisier.like} since $\frac{1}{t_i}=\frac{1}{q_i}-\frac{1}{t}=\frac{1}{t^\prime}-\frac{1}{q_i^\prime}$.
\end{proof}

Another important property of multi-indexed Schatten quasi-norms is the multiplicativity under tensor products along arbitrary cuts. 

\begin{corollary}[Multiplicativity under tensor products]\label{cor.multi.tensor.multiplicativity}
Let $\cQ, \cP$ be compatible multi-indexed Schatten spaces in the sense that $(\cQ,\cP)$ indexes a multi-indexed Schatten space. Let $X\in\cS_{\cQ}, Y\in\cS_{\cP}$, then it holds that 
\begin{align}
    \|X\otimes Y\|_{(\cQ,\cP)}= \|X\|_{(\cQ)} \|Y\|_{(\cP)},
\end{align} i.e. the multi-indexed Schatten quasi-norms split along arbitrary tensor products.
\end{corollary}
\begin{remark}
This is a classic and very much expected property of these quasi-norms. The important part of this result is that the multiplicativity holds for any tensor product cuts compatible with the index splitting. A particularly simple consequence of this is that for some multi-indexed space $(\cQ)=(q_1,\dots, q_n)$ it holds that
\begin{align}
    \|\bigotimes_{i=1}^nX_i\|_{(\cQ)}= \prod_{i=1}^n\|X_i\|_{q_i}.
\end{align}
\end{remark}

\begin{proof}
Using \Cref{thm:Multi.Thm.Pisier.like} we will lift the operator space theoretic result to this general multi-indexed one, proving the two inequalities hidden in \Cref{cor.multi.tensor.multiplicativity} using the infimum, respectively supremum formulas of \Cref{thm:Multi.Thm.Pisier.like}.

\noindent Without loss of generality assume all involved indices to be positive. Else the result follows from that case and \Cref{prop:negative.multi.indexed}.
Let $X\in\cS_{\cQ}, Y\in\cS_{\cP}$, with $\cQ=(q_1,\dots,q_n)$ and $\cP=(p_1,\dots,p_m)$, then by \Cref{thm:Multi.Thm.Pisier.like} we have, on the one hand, for the compatible $s:=\max_{ij} \{q_i,p_j\}$ and with $\cR_Q,\cR_P$ the collections given by $\frac{1}{q_i}-\frac{1}{s}$, respectively $\frac{1}{p_j}-\frac{1}{s}$, as in \Cref{thm:Multi.Thm.Pisier.like}, that
\begin{align}
    \|X\otimes Y\|_{(\cQ,\cP)} &= \inf_{A,B}\|AA^*\|^\frac{1}{2}_{(\cR_Q,\cR_P)}\|B^*B\|^\frac{1}{2}_{(\cR_Q,\cR_P)}\|A^{-1}(X\otimes Y)B^{-1}\|_{s} \\
    &\leq \inf_{A_Q,A_P,B_Q,B_P}\|A_QA_Q^*\otimes A_PA_P^*\|^\frac{1}{2}_{(\cR_Q,\cR_P)}\|B_Q^*B_Q\otimes B_P^*B_P\|^\frac{1}{2}_{(\cR_Q,\cR_P)}\\&\qquad\times\|A_Q^{-1}XB_Q^{-1}\otimes A_P^{-1}YB_P^{-1}\|_s \\
    &= \inf_{A_Q,A_P,B_Q,B_P}\|A_QA_Q^*\|^\frac{1}{2}_{(\cR_Q)}\|A_PA_P^*\|^\frac{1}{2}_{(\cR_P)}\|B_Q^*B_Q\|^\frac{1}{2}_{(\cR_Q)}\|B_P^*B_P\|^\frac{1}{2}_{(\cR_P)}\\&\qquad\times\|A_Q^{-1}XB_Q^{-1}\|_s\|A_P^{-1}YB_P^{-1}\|_s \\
    &= \|X\|_{(\cQ)} \|Y\|_{(\cP)}.
\end{align}
where we used the multiplicativity of the normed multi-indexed Schatten spaces under tensor products of \Cref{cor:Pisier.multi.tensor.splitting} with $\cX=\cR_P$, which applies since $\cR_Q,\cR_P$ are normed as $s$ was constructed to be compatible with $\cQ,\cP$. 

The exact same argument using the supremum formulas of \Cref{thm:Multi.Thm.Pisier.like} with some compatible $0<t\leq 1$ yields the other inequality, establishing the statement of the corollary.
\end{proof}

The next property, which holds for multi-indexed Schatten spaces $\cS_q[\cS_{\cQ}]$, is that they collapse under block-diagonal operators, e.g.\ classical-quantum operators, to $\ell_q[\cS_{\cQ}]$. 

\begin{corollary}[Collapse on block-diagonal operators]\label{cor:multi.block.collapse}
Let $q\notin\{-\infty,0,\infty\}$ and $\cQ$ a multi-indexed Schatten space compatible with $q$. Let $X_i\in\cS_{\cQ}$, which we assume to be strictly positive when the indices are negative. Then it holds that
\begin{align}
    \|\oplus_iX_i\|_{(q,\cQ)}= \left(\sum_i\|X_i\|^q_{(\cQ)}\right)^\frac{1}{q},
\end{align} which extends naturally for $q=\infty$ to
\begin{align}
     \|\oplus_iX_i\|_{(\infty,\cQ)}= \max_i\|X_i\|_{(\cQ)}
\end{align} and for $q=-\infty$ to
\begin{align}
    \|\oplus_iX_i\|_{(-\infty,\cQ)}= \min_i\|X_i\|_{(\cQ)}
\end{align}
\end{corollary}
\begin{remark}
This is the natural and expected multi-indexed extension of \cite[Theorem 1.4]{Kochanowski.2026} and the quasi- and anti-norm extension of \cite[Corollary 1.3]{book:Pisier.1998}.
In the study of quantum information this result comes up in the study of classical-quantum states and conditioning on classical side information. For the particular case where $(\cQ)=(1,\alpha)$ with $\alpha\in[\frac{1}{2},\infty]$ it is proven in \cite[Proposition 5.4 (5.34)]{book:Tomamichel.2016}.
\end{remark}

\begin{proof}
The idea of this proof is to lift the operator-space-theoretic result \cite[Corollary 1.3]{book:Pisier.1998} to the multi-indexed Schatten setting via use of \Cref{thm:Multi.Thm.Pisier.like}. Explicitly we prove the two inequalities hidden in the statement of \Cref{cor:multi.block.collapse} using the two formulas of \Cref{thm:Multi.Thm.Pisier.like}, in each case lower bounding the supremum, respectively upper bounding the infimum by block-diagonal optimizers. Together these will show that block-diagonal optimizers are optimal for block-diagonal $X=\oplus_iX_i$ leading to the above statement. This proof is thus on the surface different from its two-indexed special case (\cite[Theorem 1.4]{Kochanowski.2026}) as it does not directly make use of some joint concavity, respectively convexity properties.

\noindent Without loss of generality we can again assume all involved indices to be positive. Else the result follows with \Cref{prop:negative.multi.indexed}.
Let $\cQ=(q_1,\dots, q_n)$ and take first some compatible $0<s$ s.t. $s\geq q_i$ $\forall i$. Then by \Cref{thm:Multi.Thm.Pisier.like} with $\cR=(r,r_1,\dots,r_n)$, where $\frac{1}{r_i}=\frac{1}{q_i}-\frac{1}{s}\leq 1$ for $i\in\{\emptyset,1,\dots,n\}$, we have, on the one hand, that
\begin{align}
    \|\oplus_iX_i\|_{(q,\cQ)} 
    &= \inf_{A,B}\|AA^*\|^{\frac{1}{2}}_{(r,\cR)}\|B^*B\|^{\frac{1}{2}}_{(r,\cR)}\|A^{-1}(\oplus_iX_i)B^{-1}\|_{s} \\
    &\leq \inf_{\oplus_iA_i, \oplus_iB_i}\|\oplus_iA_iA_i^*\|^{\frac{1}{2}}_{(r,\cR)}\|\oplus_iB_i^*B_i\|^{\frac{1}{2}}_{(r,\cR)}\|\oplus_iA_i^{-1}X_iB_i^{-1}\|_{s} \\ &= \inf_{A_i,B_i} \big\|\, \|A_iA_i^*\|_{(\cR)} \,\big\|^\frac{1}{2}_{\ell_{r}}\big\|\, \|B^*_iB_i\|_{(\cR)} \,\big\|^\frac{1}{2}_{\ell_{r}}\big\|\, \|A_i^{-1}X_iB_i^{-1}\|_s \,\big\|_{\ell_s} \\ 
    &\leq\inf_{\|A_iA_i^*\|_{(\cR)},\|B^*_iB_i\|_{(\cR)}} \big\|\, (\|A_iA_i^*\|^\frac{1}{2}_{(\cR)})_i \,\big\|_{\ell_{2r}}\big\|\, (\|B^*_iB_i\|^\frac{1}{2}_{(\cR)})_i \,\big\|_{\ell_{2r}}\\ &\qquad\times\big\|\,( \|A_iA_i^*\|_{(\cR)}^{-\frac{1}{2}}\|B^*_iB_i\|_{(\cR)}^{-\frac{1}{2}}\|X_i\|_{(\cQ)}(1+\epsilon))_i \,\big\|_{\ell_s} \\
    &= \big\|\, (1+\epsilon)\|X_i\|_{(\cQ)}\, \big\|_{\ell_q} = (1+\epsilon)\big\|\, \|X_i\|_{(\cQ)}\, \big\|_{\ell_q}.
\end{align} where in the second line we lower bounded the infimum over all $A,B$ by block-diagonal ones and in the third used the known block collapse due to Pisier's \cite[Corollary 1.3]{book:Pisier.1998}. In the fourth line we used the achievability part of \Cref{thm:Multi.Thm.Pisier.like} guaranteeing us, for any fixed $\epsilon>0$, the existence of operators $A_i,B_i$ s.t. 
\begin{align}
    (1+\epsilon)\|X_i\|_{(\cQ)}\geq \|A_iA_i^*\|^\frac{1}{2}_{(\cR)}\|B^*_iB_i\|^\frac{1}{2}_{(\cR)}\|A_i^{-1}X_iB_i^{-1}\|_s.
\end{align} The remaining degree of freedom we optimize over is their `size', since the above achievability is scale invariant, i.e.\ if it holds for some pair $(A_i,B_i)$, then it also does for $(\alpha A_i,\beta B_i)$ for any $\alpha,\beta\neq0$.
The last equality now follows from \Cref{prop:SpVariationalFormulas} applied to diagonal operators, i.e.\ the classical factorization of $\ell_q$-(quasi-)norms. Since this holds for any $\epsilon>0$ we have proven the ``$\leq$'' inequality of \Cref{cor:multi.block.collapse}. To prove the other inequality we repeat the same argument, but now using the supremum formula from \Cref{thm:Multi.Thm.Pisier.like} instead of the infimum. 
To this end take some compatible $0<t\leq 1$ s.t. $t\leq q_i$ $\forall i$ and observe that by \Cref{thm:Multi.Thm.Pisier.like} we have, analogously to above,
\begin{align}
    \|\oplus_iX_i\|_{(q,\cQ)} 
    &= \sup_{A,B}\|A^*A\|^{-\frac{1}{2}}_{(r,\cR)}\|BB^*\|^{-\frac{1}{2}}_{(r,\cR)}\|A(\oplus_iX_i)B\|_{t} \\
    &\geq \sup_{\oplus_iA_i, \oplus_iB_i}\|\oplus_iA^*_iA_i\|^{-\frac{1}{2}}_{(r,\cR)}\|\oplus_iB_iB^*_i\|^{-\frac{1}{2}}_{(r,\cR)}\|\oplus_iA_iX_iB_i\|_{t} 
    \\ &= \sup_{A_i,B_i} \big\|\, \|A^*_iA_i\|_{(\cR)} \,\big\|^{-\frac{1}{2}}_{\ell_{r}}\big\|\, \|B_iB^*_i\|_{(\cR)} \,\big\|^{-\frac{1}{2}}_{\ell_{r}}\big\|\, \|A_iX_iB_i\|_t \,\big\|_{\ell_t} \\ 
    &\geq\sup_{\|A^*_iA_i\|_{(\cR)},\|B_iB^*_i\|_{(\cR)}} \big\|\, (\|A^*_iA_i\|^{\frac{1}{2}}_{(\cR)})_i \,\big\|^{-1}_{\ell_{2r}}\big\|\, (\|B_iB^*_i\|^\frac{1}{2}_{(\cR)})_i \,\big\|^{-1}_{\ell_{2r}} \\ &\qquad\times\big\|\,( \|A^*_iA_i\|_{(\cR)}^{\frac{1}{2}}\|B_iB^*_i\|_{(\cR)}^{\frac{1}{2}}\|X_i\|_{(\cQ)}(1-\epsilon))_i \,\big\|_{\ell_t} \\
    &= \big\|\, (1-\epsilon)\|X_i\|_{(\cQ)}\, \big\|_{\ell_q} = (1-\epsilon)\big\|\, \|X_i\|_{(\cQ)}\, \big\|_{\ell_q}.
\end{align} Since this works for any $\epsilon>0$ this concludes the proof. 
\end{proof}

\subsubsection{Quasi- and Anti-Normability: The Triangle and Reverse Triangle Inequality}\label{sub:quasi.anti}

In this subsection we finally prove that the multi-indexed quasi-norms, as we called them, are actually quasi-norms, i.e.\ satisfy the modified triangle inequality. We also prove the regular triangle inequality when all indices are at least $1$, and the reverse triangle (Minkowski) inequality when all indices are at most $1$. The latter turns the negatively indexed ones into anti-norms. We begin with the former two.

\begin{theorem}[Quasi-norm and triangle inequality]\label{thm:Multi.Quasi.Norms}
Let $\cQ=(q_1,\dots,q_n)$ be a multi-indexed Schatten space with $q_i>0$. Then $\cQ$ is a quasi-normed space which is $\kappa$-normable, where $\kappa:=\min\{1,\min_i\{q_i\}\}$.
\end{theorem}

\begin{remark}
This is the natural generalization of \cite[Theorem 1.2]{Kochanowski.2026} to the multi-indexed setting.
On the one hand, when all indices are at least $1$, i.e.\ $\kappa=1$, it says that the functionals $\|\cdot\|_{(\cQ)}$ are norms and hence satisfy the triangle inequality. This is more or less directly a consequence of the fact that these functionals are by construction equal to multi-indexed Pisier norms, which are norms to begin with.\footnote{While it might appear that the proof we give of this is independent of Pisier's operator space construction, this is technically not quite true, since it relies on \Cref{thm:Multi.Thm.Pisier.like}. So technically implicitly it still requires through \Cref{lem.app.op.concavity} that operator spaces embed completely isometrically into some $\cB(\cK)$, which on its own is sufficient to prove the triangle inequality.} 

The more interesting part is probably when some indices are below one but still positive, i.e. $\kappa<1$, then the claim of the theorem is that
\begin{align}
    \|X+Y\|^\kappa_{(\cQ)}\leq  \|X\|^\kappa_{(\cQ)}+ \|Y\|^\kappa_{(\cQ)}.
\end{align}
It also implies (sequential) continuity of the map $X\mapsto \|X\|^\kappa_{(\cQ)}$ and thus also of $X\mapsto \|X\|^\kappa_{(\cQ)}$. This is because
\begin{align}
    \left|\|X+\epsilon Y\|^\kappa_{(\cQ)}-\|X\|^\kappa_{(\cQ)} \right| \leq \epsilon^\kappa\|Y\|^\kappa_{(\cQ)}\rightarrow  0 \quad \text{ for } \epsilon\to 0
.\end{align}
\end{remark}

\begin{proof}
This follows analogously to its two-indexed counterpart \cite[Theorem 1.2]{Kochanowski.2026}. Positive homogeneity and positive definiteness are easily checked from the definition, so it remains to show the $\kappa$-normability. Let $\kappa:=\min\{1,\min_i\{q_i\}\}$, then by \Cref{thm:Multi.Thm.Pisier.like} with $t=\kappa$, where by \Cref{rem:sup.arbitrary} the supremum may be taken over arbitrary $A,B$ and is hence the same for $X$, $Y$ and $X+Y$, and the $\kappa$-normability of the $\kappa$-Schatten quasi-norm it follows that
\begin{align}
    \|X+Y\|^\kappa_{(\cQ)}&=\sup_{A,B}\|A^*A\|^{-\frac{\kappa}{2}}_{(\cR)}\|BB^*\|^{-\frac{\kappa}{2}}_{(\cR)}\|A(X+Y)B\|^\kappa_\kappa \\ &\leq\sup_{A,B}\|A^*A\|^{-\frac{\kappa}{2}}_{(\cR)}\|BB^*\|^{-\frac{\kappa}{2}}_{(\cR)}(\|AXB\|^\kappa_\kappa+\|AYB\|^\kappa_\kappa) \\
    &\leq \sup_{A,B}\|A^*A\|^{-\frac{\kappa}{2}}_{(\cR)}\|BB^*\|^{-\frac{\kappa}{2}}_{(\cR)}\|AXB\|^\kappa_\kappa + \sup_{A,B}\|A^*A\|^{-\frac{\kappa}{2}}_{(\cR)}\|BB^*\|^{-\frac{\kappa}{2}}_{(\cR)}\|AYB\|^\kappa_\kappa \\
    &= \|X\|_{(\cQ)}^\kappa+\|Y\|_{(\cQ)}^\kappa,
\end{align} where the second to last line follows from the subadditivity of the supremum. 
\end{proof}

Finally, let us consider the reverse Minkowski (triangle) inequality in the multi-indexed setting.

\begin{theorem}[Reverse Minkowski (Triangle) inequality]\label{thm:Multi.Rev.Triangle}
Let $\cQ=(q_1,\dots,q_n)$ be a multi-indexed Schatten space s.t. $q_i\leq 1\, \forall i$. Then $\cQ$ satisfies the reverse Minkowski (reverse triangle) inequality. That is for $X,Y\geq0$ ($X,Y>0$ if the indices involved are negative):
\begin{align}
    \|X+Y\|_{(\cQ)} \geq \|X\|_{(\cQ)}+\|Y\|_{(\cQ)}.
\end{align}
This means that $\|\cdot\|_{(\cQ)}$ are so-called anti-norms.
\end{theorem}

\begin{remark}
This is the multi-indexed generalization of the reverse Minkowski (triangle) inequality for Schatten quasi-norms to indices $q\leq 1$ (and negative explicitly allowed), see e.g. \cite{Cubitt.2015}.
\end{remark}

\begin{proof}
We first consider $X,Y>0$. By the assumption on $\cQ$ there exists some $q\leq 1$ compatible with $\cQ$ and s.t.\ $q_i\leq q$, where for negative indices we choose $q\in[-1,0)$ as in the proof of \Cref{cor:reverse.multi.hölder}. Since $X$, $Y$ and $X+Y$ are positive definite, they all have full support, and the formulas for self-adjoint operators of \Cref{thm:Multi.Thm.Pisier.like}, respectively \Cref{cor:multi.Pisier.like.negative} for negative indices, express their quasi-norms as infima over the same set of positive definite $\tilde A$. Together with the reverse triangle inequality for the $q$-Schatten quasi-norm \cite[Lemma 6]{Cubitt.2015} we obtain with $(\cR)=(r_1,\dots, r_n)$, where $\frac{1}{r_i}=\frac{1}{q_i}-\frac{1}{q}$,
\begin{align}
    \|X+Y\|_{(\cQ)} &= \inf_{\tilde{A}>0}\|\tilde{A}\|_{(\cR)}\|\tilde{A}^{-\frac{1}{2}}(X+Y)\tilde{A}^{-\frac{1}{2}}\|_q \\
    &\geq \inf_{\tilde{A}>0}\|\tilde{A}\|_{(\cR)}\big(\|\tilde{A}^{-\frac{1}{2}}X\tilde{A}^{-\frac{1}{2}}\|_q+\|\tilde{A}^{-\frac{1}{2}}Y\tilde{A}^{-\frac{1}{2}}\|_q\big) \\
    &\geq \inf_{\tilde{A}>0}\|\tilde{A}\|_{(\cR)}\|\tilde{A}^{-\frac{1}{2}}X\tilde{A}^{-\frac{1}{2}}\|_q+ \inf_{\tilde{A}>0}\|\tilde{A}\|_{(\cR)}\|\tilde{A}^{-\frac{1}{2}}Y\tilde{A}^{-\frac{1}{2}}\|_q =\|X\|_{(\cQ)}+\|Y\|_{(\cQ)}.
\end{align}
For negative indices $X,Y>0$ is assumed, which completes the proof in that case. For positive indices and $X,Y\geq0$ we apply this to $X+\varepsilon\1$ and $Y+\varepsilon\1$ and let $\varepsilon\to0$, using that $\|\cdot\|_{(\cQ)}$ is a $\kappa$-normable quasi-norm by \Cref{thm:Multi.Quasi.Norms} and hence continuous, $\big|\|Z+E\|^\kappa_{(\cQ)}-\|Z\|^\kappa_{(\cQ)}\big|\leq\|E\|^\kappa_{(\cQ)}$. (A compression argument $\tilde A\mapsto\Pi_X\tilde A\Pi_X$ cannot be used here: the bimodule property only holds for operators acting on the outermost system, and a global compression can increase $\|\cdot\|_{(\cR)}$, e.g.\ for $\cR=(\infty,1)$.)
\end{proof}

\section{Linear Maps Between Multi-Indexed Schatten Spaces}
\label{sec:LinerMaps}

In this section we consider the induced (quasi-)norms and co-(quasi-)norms on the space of linear maps between the previously developed multi-indexed Schatten spaces. We will also refer to them as \emph{induced} or \emph{mixed quasi-norms} even if they may not actually be quasi- but anti-norms. 
We will derive properties of such induced (co)-quasi-norms and important general tools, which we will use in the following section to prove various  multiplicativity results and chain rules in the context of quantum information theory. Prominent examples of such mixed quasi-norms include completely bounded (quasi-)norms. These are of interest from a structural mathematical as well as a quantum information point of view \cite{book:Paulsen.2003, book:Watrous.2018, Li.2026, Kochanowski.2026, Beigi.2016, Gupta.2015, Johnston.09, Devetak.2006}. 
From a mathematical point of view these induced norms are natural quantities to consider, and a rich structure emerges, in particular w.r.t.\ multiplicativity 
under suitable compositions and tensor products of different maps, given certain ranges of indices. 

From a physics or quantum information point of view a different interest in these quantities emerges, in particular because, as we will see in \cref{sec:ApplictoQI}, certain (quasi-)norms are directly related to entropic quantities, see also \cite{Li.2026, Kochanowski.2026, Devetak.2006}. The resulting multiplicativity statements or (in)equalities of this section can thus be interpreted as \emph{entropy additivity} statements or \emph{chain rules}, governing how quantum channels can generate or manipulate properties such as entanglement with surrounding systems. Such additivity results of entropic quantities under composition of linear maps are powerful tools in the study of quantum key distribution protocols \cite{Fawzi.2026}, entanglement certification, but also dissipative dynamics in an environment-agnostic manner \cite{Beigi.2016, Bardet.2022}. 

In \Cref{Lin.1} we consider the induced \emph{completely bounded} (co-)(quasi-)norms of linear maps between multi-indexed Schatten spaces of \Cref{sec:MultiIndex}. These completely bounded (co-)norms should be interpreted as quantities which represent ancilla-assisted versions of their non-completely-bounded counterparts, because they factor in possible entanglement with arbitrary `environments'. Then, to derive various multiplicativity statements, we extend two seminal tools to the multi-indexed and quasi-normed setting. First, in \Cref{Lin.2}, we consider the $\SWAP$ channel, which exchanges two adjacent systems, and prove its complete contractiveness, respectively expansiveness, depending on the order of the exchanged indices, see \Cref{thm:QuasiNormMinkowskiInequality}. This inequality is also known as a non-commutative Minkowski inequality and we prove a (quasi-normed) multi-indexed version generalizing the previously established normed two-indexed version \cite[Corollary 1.10]{book:Pisier.1998}, see also \cite{Devetak.2006}.
Secondly, in \Cref{sec:3.removing.indices}, we prove two general chain-rule type equalities which, informally, allow us to remove identity channels from the innermost systems of mixed norms under suitable conditions on the involved indices. The first one is \Cref{lem:reductionLemma}, which generalizes the normed statement \cite[Lemma 5]{Devetak.2006} and \cite[Theorem 4.1]{Fawzi.2026} to completely positive maps between arbitrary normed multi-indexed Schatten spaces. Such a statement was instrumental in proving the multiplicativity results of \cite{Devetak.2006, Fawzi.2026, Kochanowski.2026} and many more. The second one is \Cref{thm:id.right.quasi}, which is effectively a quasi-normed analogue of the former for co-quasi-norms with indices in $[\frac{1}{2},1]$. We present the most general multi-indexed equality of this type and later apply it to multiplicativity statements involving indices in that range. 
As a consequence of the former, the mixed multi-indexed Schatten (quasi-)norms of completely positive maps are attained on positive inputs.

\medskip

\textbf{Notation:} \newline
In this and the following sections we will be using a slightly adapted notation with the aim to make the associations between indices and subsystems more transparent. Whereas in the previous section the association between indices and systems was positional, here we explicitly label each (sub)-system and index with an uppercase Latin letter $A,B,R,S,Q,P,\dots$ and the order is determined by the order of indices in the (quasi)-norm. That is for an operator $X_{QP}$ where $Q,P$ are system labels to which we associate the compatible indices $q,p$, respectively, we will write
\begin{align}
    \|X\|_{(q,p)}\equiv \|X\|_{(Q:q,P:p)} = \|X_{QP}\|_{(Q:q,P:p)} = \|X_{PQ}\|_{(Q:q,P:p)}.
\end{align} We will, wherever possible, avoid the last notation, where the order of the indices on the operator and on the (quasi-)norm is not matching. Similarly, when $Q=Q_1Q_2\dots Q_n$ is a collection of $n$ systems, on which the multi-indexed Schatten space $\cQ= (q_1,q_2,\dots, q_n)$ is defined, then we will write
\begin{align}
    \|X\|_{(q_1,q_2,\dots, q_n)} \equiv \|X\|_{(Q:\cQ)} \equiv \|X_Q\|_{(Q:\cQ)}.
\end{align}

\subsection{Completely Bounded Quasi-Norms and Co-Quasi-Norms}\label{Lin.1}

Completely bounded norms \cite{book:Paulsen.2003} find many applications in quantum information theory, giving rise to tensor-stable and thus environment-agnostic quantities. Perhaps most famously, they appear in the form of the diamond norm \cite[Chapter 3]{book:Watrous.2018} which is the completely bounded $1\to 1$ norm and serves as an ancilla-assisted optimal one-shot distinguishability between quantum channels when applied to a difference of channels. When the input and output spaces are (multi-indexed) Schatten spaces, applications include channel entropy uncertainty relations \cite{Gao.2018.1}, completely bounded hypercontractivity \cite{Bardet.2022, Bardet.2024,Beigi.2016} and complete modified logarithmic Sobolev constants \cite{Gao.2022, Bardet.2024}, completely bounded minimal output entropies \cite{Devetak.2006}, and minimal output (conditional) entropies under arbitrary environments \cite{Fawzi.2026, Li.2026}.
In this section we extend the notion of completely bounded norms to quasi-normed spaces and consider completely bounded quasi-norms and co-quasi-norms used, among others, in the multiplicativity results of \Cref{sec:Mixed.Multiplicativity.Statements} in analogy to \cite[Section 4]{Devetak.2006}.

We will make use of the following standard definitions: given any two multi-indexed Schatten spaces $\mathcal{Q},\mathcal{P}$ with associated quasi- or anti-norms labeled as $\|\cdot \|_{\cQ}$ and $\|\cdot \|_{\cP}$, and a bounded map $\Phi:Q\to P$:
\begin{align}
\|\Phi\|_{(Q:\cQ)\to (P:\cP)}:=\sup_{X\in{\cS_{\cQ}}}\frac{\|\Phi(X)\|_{(P:\cP)}}{\|X\|_{(Q:\cQ)}}\quad \text{ and }\quad
\|\Phi\|^+_{(Q:\cQ)\to (P:\cP)}:=\sup_{\underset{X>0}{X\in{\cS_{\cQ}}}}\frac{\|\Phi(X)\|_{(P:\cP)}}{\|X\|_{(Q:\cQ)}}\,~\,.
\end{align} 
Note that it is not difficult to see that these functionals are quasi- or anti-norms on the cone of all bounded (CP) linear maps. In fact they are $\kappa=\min\{1,\min_ip_i\}$-normable, when $(\cP)=(p_1,\dots,p_n)$ with $p_i>0$.

Similarly, we define the co-(quasi-)norms of CP maps $\Phi$ as 
\begin{align}
\|\Phi\|^+_{\co,(Q:\cQ)\to (P:\cP)}:=\inf_{\underset{X>0}{X\in\mathcal{S}_{\cQ}}}\frac{\|\Phi(X)\|_{(P:\cP)}}{\|X\|_{(Q:\cQ)}}\,~\,
\end{align} and refer to them analogously as `mixed' co-quasi-norms. These are generally not quasi-norms and will generically be trivial when the optimization is not restricted to positive operators. 

\noindent It is a standard fact that these mixed quasi-norms are sub-multiplicative under composition of linear maps and the mixed co-quasi-norms are super-multiplicative under composition of completely positive maps.
\begin{proposition}\label{prop:standard.sub.and.sup.mult}
Let $\Phi:Q\to P, \Psi: P\to R$ be linear, then
\begin{align}
    \|\Psi\circ\Phi\|_{(Q:\cQ)\to (R:\cR)} \leq \|\Phi\|_{(Q:\cQ)\to (P:\cP)}\|\Psi\|_{(P:\cP)\to (R:\cR)}.
\end{align} Moreover when $\Phi: Q\to P, \Psi: P\to R$ are completely positive linear maps, then
\begin{align}
    \|\Psi\circ\Phi\|^{(+)}_{(Q:\cQ)\to (R:\cR)} &\leq \|\Phi\|^{(+)}_{(Q:\cQ)\to (P:\cP)}\|\Psi\|^{(+)}_{(P:\cP)\to (R:\cR)} \\  \|\Psi\circ\Phi\|^+_{\co,(Q:\cQ)\to (R:\cR)} &\geq \|\Phi\|^+_{\co,(Q:\cQ)\to (P:\cP)}\|\Psi\|^+_{\co,(P:\cP)\to (R:\cR)}
\end{align}
\end{proposition}

\begin{proof}
We prove the first case here, as the others follow analogously. 
\begin{align}
    \|\Psi\circ\Phi\|_{(Q:\cQ)\to (R:\cR)} &= \sup_X \|\Psi(\underbrace{\Phi(X)}_{=:Y})\|_{(R:\cR)}\|X\|_{(Q:\cQ)}^{-1} \\&=\sup_X\|\Psi(Y)\|_{(R:\cR)}\|Y\|^{-1}_{(P:\cP)}\|\Phi(X)\|_{(P:\cP)}\|X\|_{(Q:\cQ)}^{-1} \\ &\leq \|\Psi\|_{(P:\cP)\to(R:\cR)}\sup_X\|\Phi(X)\|_{(P:\cP)}\|X\|_{(Q:\cQ)}^{-1} \\ &= \|\Phi\|_{(Q:\cQ)\to (P:\cP)}\|\Psi\|_{(P:\cP)\to (R:\cR)}.
\end{align}
\end{proof}

\subsubsection{Completely Bounded Quasi-Norms of Linear Maps}

Inspired by the completely bounded norm between operator spaces and its expression in terms of operator-valued Schatten norms \cite{book:Paulsen.2003, book:Pisier.1998}, we define the concept of complete boundedness between (quasi-normed) multi-indexed Schatten spaces in the following.

\begin{definition}[Completely bounded quasi-norm]\label{def:cb.quasi.norm}
Let $\cQ,\cP$ be multi-indexed Schatten spaces on systems $Q,P$ such that $t\neq0$ is compatible with both of them and $\Phi:Q\to P$ be a linear map, then we define
\begin{align}
    &\|\Phi\|_{\cb,(Q:\cQ)\to (P:\cP)}:=\sup_E\|\id_E\otimes\Phi\|_{(E:t,Q:\cQ)\to (E:t,P:\cP)}\\
    &\|\Phi\|^+_{\cb,(Q:\cQ)\to (P:\cP)}:=\sup_E\|\id_E\otimes\Phi\|^+_{(E:t,Q:\cQ)\to (E:t,P:\cP)}\,,
\end{align}
where the supremum is over any finite-dimensional reference system $E$. \newline
When $t<0$ and $\cQ,\cP$ are negatively indexed we will only work with the $^+$ version of these induced quantities. 
\end{definition}

\begin{remark}
 When $\cQ,\cP$ are normed operator-valued Schatten spaces with indices $\in[1,\infty]$ this reduces to the usual completely bounded norm \cite[Lemma 1.7]{book:Pisier.1998}. The above, though, is well defined for any multi-indexed Schatten spaces, not necessarily normed ones. Note also that we do not require $\cQ$ and $\cP$ themselves to be compatible with each other, just that both are compatible with the same $t\neq0$. This works as long as $\big|\frac{1}{q_i}-\frac{1}{p_j}\big|\leq 2$ for all $q_i\in\cQ, p_j\in\cP$.   
\end{remark}

A priori it is not yet clear that this above defined completely bounded quasi-norm is indeed well defined, as it may depend on the choice of $t$. The fact that it is well defined is due to the following \Cref{lem:cb.norm.via.t}. In the mentioned normed setting, this is a well-known statement going back to Pisier \cite[Lemma 1.7]{book:Pisier.1998}.

\begin{lemma}\label{lem:cb.norm.via.t}
Let $\cQ,\cP$ be multi-indexed Schatten spaces on systems $Q,P$ and $\Phi:Q\to P$ be a linear map, then its completely bounded quasi-norm defined in \Cref{def:cb.quasi.norm} is independent of the index $t$ put on the reference system. In particular for any $q,p\neq0$ which are compatible and compatible with both $\cQ,\cP$ it holds that
\begin{align}
    \|\Phi\|^{(+)}_{\cb,(Q:\cQ)\to (P:\cP)}=\sup_E\|\id_E\otimes\Phi\|^{(+)}_{(E:q,Q:\cQ)\to (E:q,P:\cP)}=\sup_E\|\id_E\otimes\Phi\|^{(+)}_{(E:p,Q:\cQ)\to (E:p,P:\cP)}.
\end{align}
Moreover let $\cE$ be a multi-indexed Schatten space on system $E$ which is compatible with both $\cQ$ and $\cP$, then it holds that
\begin{align}
    \|\Phi\|^{(+)}_{\cb,(Q:\cQ)\to (P:\cP)}=\sup_E\|\id_{E}\otimes\Phi\|^{(+)}_{(E:\cE,Q:\cQ)\to (E:\cE,P:\cP)}.
\end{align}
Here we enforce the $^+$ when $q,p<0$ and else prove both versions of the statement.
\end{lemma}

\begin{proof}
Take for now $0<q\leq p$ to be compatible indices compatible with both $\cQ,\cP$. 
This equivalence is a simple consequence of the factorization formulas from \Cref{thm:multi.quasi.relational.consistency}. We show the equality by separately proving the following slightly stronger equality
\begin{align}
    \|\id_E\otimes\Phi\|_{(E:q,Q:\cQ)\to (E:q,P:\cP)} = \|\id_E\otimes\Phi\|_{(E:p,Q:\cQ)\to (E:p,P:\cP)},
\end{align} which upon taking a supremum over the system $E$ yields the claim of the lemma. 

On the one hand we have, for any input $X_{EQ}$ and any $\epsilon>0$, that there exists approximately optimal $\hat{a}_E,\hat{b}_E,\hat{Z}_{EQ}$ s.t. $X_{EQ}=\hat{a}_E\hat{Z}_{EQ}\hat{b}_E$ and
\begin{align}
    \|X\|_{(E:q,Q:\cQ)}\geq \|\hat{a}\|_{2r}\|\hat{b}\|_{2r}\|\hat{Z}\|_{(E:p,Q:\cQ)}-\epsilon=\|\hat{Z}\|_{(E:p,Q:\cQ)}-\epsilon,
\end{align}
as we may assume $\|\hat{a}\|_{2r}=\|\hat{b}\|_{2r}=1$. At the same time
\begin{align}
    \|(\id_E\otimes\Phi)(X_{EQ})\|_{(E:q,P:\cP)}\leq \|a\|_{2r}\|b\|_{2r}\|(\id_E\otimes\Phi)(Z_{EQ})\|_{(E:p,P:\cP)},
\end{align} for any $a,b,Z$ with $X_{EQ}=a_EZ_{EQ}b_E$, since this yields $(\id_E\otimes\Phi)(X_{EQ})=a_E(\id_E\otimes\Phi)(Z_{EQ})b_E$. Together these yield
\begin{align}
    \frac{\|(\id_E\otimes\Phi)(X)\|_{(E:q,P:\cP)}}{\|X\|_{(E:q,Q:\cQ)}}\leq \frac{\|\hat{a}\|_{2r}\|\hat{b}\|_{2r}\|(\id_E\otimes\Phi)(\hat{Z})\|_{(E:p,P:\cP)}}{\|\hat{a}\|_{2r}\|\hat{b}\|_{2r}\|\hat{Z}\|_{(E:p,Q:\cQ)}-\epsilon}= \frac{\|(\id_E\otimes\Phi)(\hat{Z})\|_{(E:p,P:\cP)}}{\|\hat{Z}\|_{(E:p,Q:\cQ)}-\epsilon}.
\end{align} Since this works for any $\epsilon>0$ it follows that
\begin{align}
    \frac{\|(\id_E\otimes\Phi)(X)\|_{(E:q,P:\cP)}}{\|X\|_{(E:q,Q:\cQ)}}\leq \sup_{Z_{EQ}}\frac{\|(\id_E\otimes\Phi)(Z)\|_{(E:p,P:\cP)}}{\|Z\|_{(E:p,Q:\cQ)}}
\end{align}
Hence we have shown that
\begin{align}
    \|\id_E\otimes\Phi\|_{(E:q,Q:\cQ)\to (E:q,P:\cP)} \leq \|\id_E\otimes\Phi\|_{(E:p,Q:\cQ)\to (E:p,P:\cP)}.
\end{align}

On the other hand we similarly get that for any $\epsilon>0$ there exists $\hat{a},\hat{b}$ such that 
\begin{align}
   \|(\id_E\otimes\Phi)(X_{EQ})\|_{(E:p,P:\cP)}\leq \|\hat{a}\|^{-1}_{2r}\|\hat{b}\|^{-1}_{2r}\|(\id_E\otimes\Phi)(\hat{a}_EX_{EQ}\hat{b}_E)\|_{(E:q,P:\cP)}+\epsilon, 
\end{align} where again we may assume $\|\hat{a}\|^{-1}_{2r}=\|\hat{b}\|^{-1}_{2r}=1$. Analogously to before, it also holds that
\begin{align}
    \|X_{EQ}\|_{(E:p,Q:\cQ)}\geq \|\hat{a}\|^{-1}_{2r}\|\hat{b}\|^{-1}_{2r}\|\hat{a}_EX_{EQ}\hat{b}_E\|_{(E:q,Q:\cQ)}
\end{align} and hence
\begin{align}
  \frac{\|(\id_E\otimes\Phi)(X)\|_{(E:p,P:\cP)}}{\|X\|_{(E:p,Q:\cQ)}}&\leq \frac{\|\hat{a}\|^{-1}_{2r}\|\hat{b}\|^{-1}_{2r}\|(\id_E\otimes\Phi)(\hat{a}_EX_{EQ}\hat{b}_E)\|_{(E:q,P:\cP)}+\epsilon}{\|\hat{a}\|^{-1}_{2r}\|\hat{b}\|^{-1}_{2r}\|\hat{a}_EX_{EQ}\hat{b}_E\|_{(E:q,Q:\cQ)}} \\&= \frac{\|(\id\otimes\Phi)(\hat{a}_EX_{EQ}\hat{b}_E)\|_{(E:q,P:\cP)}+\epsilon}{\|\hat{a}_EX_{EQ}\hat{b}_E\|_{(E:q,Q:\cQ)}}.
\end{align} Since this works for any $\epsilon>0$ this implies that
\begin{align}
  \frac{\|(\id_E\otimes\Phi)(X)\|_{(E:p,P:\cP)}}{\|X\|_{(E:p,Q:\cQ)}}  \leq\sup_{Z_{EQ}}\frac{\|(\id\otimes\Phi)(Z)\|_{(E:q,P:\cP)}}{\|Z\|_{(E:q,Q:\cQ)}},
\end{align} which implies
\begin{align}
    \|\id_E\otimes\Phi\|_{(E:p,Q:\cQ)\to (E:p,P:\cP)} \leq \|\id_E\otimes\Phi\|_{(E:q,Q:\cQ)\to (E:q,P:\cP)}.
\end{align}
Taking the supremum over all ancilla systems $E$ completes the proof. For the positive case the exact same proof applies, since by \Cref{thm:multi.quasi.relational.consistency} if $X\geq0$ we may take $a=b^*\geq0$ and hence $Z\geq0$ as well. This positive case then equally works for positive and negative indices $q,p<0$ as \cref{thm:multi.quasi.relational.consistency} does.

\medskip

The second formula, where the completely bounding system carries a multi-indexed Schatten space $\cE=(e_1,\dots,e_{n})$, follows by iterative application of the just proved result. On the one hand, we have
\begin{align}
   \|\id_{E}\otimes\Phi&\|^{(+)}_{(E:\cE,Q:\cQ)\to (E:\cE,P:\cP)} \\&= \|\id_{E_1}\otimes\id_{E_2\dots E_n}\otimes\Phi\|^{(+)}_{(E_1:e_1,E_2:e_2,\dots,E_n:e_n,Q:\cQ)\to (E_1:e_1,E_2:e_2,\dots,E_n:e_n,P:\cP)} \\ &=\|\id_{E_1E_2}\otimes\id_{E_3\dots E_n}\otimes\Phi\|^{(+)}_{(E_1E_2:e_2,E_3:e_3,\dots,E_n:e_n,Q:\cQ)\to (E_1E_2:e_2,E_3:e_3,\dots,E_n:e_n,P:\cP)} \\ &=\cdots =\|\id_{E_1E_2\dots E_n}\otimes\Phi\|^{(+)}_{(E:e_n,Q:\cQ)\to (E:e_n,P:\cP)}.
\end{align}
Taking the supremum over $E=E_1\dots E_n$ then yields the claim.
\end{proof}

\subsubsection{Completely Bounded Co-Quasi-Norms of CP Maps}

Instead of considering the usual quasi-norm of a linear map $\Phi:Q\to P$ between two quasi-normed spaces, which are given as maximizations, one may be interested in analogous quantities involving a minimization. This is sometimes called the \textit{minimum modulus} or \textit{co}-(quasi-)norm \cite{Mbekhta.1996, Amelin.1973}. Next, we define the \textit{co-quasi-norm} and \textit{completely bounded co-quasi-norm} of a CP map. This quantity is much less well behaved than the former since it is often trivially $0$, which is why we will be focusing on the version restricted to positive inputs and CP maps. 
Despite this, the co-(quasi-)norm will also play an important role in quantum information theory applications, see \cite{Li.2026, Kochanowski.2026} and \Cref{sec:ApplictoQI}. 
 
\begin{definition}\label{def:cb.co.quasi.norm}
Let $\cQ,\cP$ be multi-indexed Schatten spaces on systems $Q,P$ such that $t\neq0$ is compatible with both of them and $\Phi:Q\to P$ be a linear map, then we define
\begin{align}
    \|\Phi\|^+_{\cb,\co,(Q:\cQ)\to (P:\cP)}:=\inf_E\|\id_E\otimes\Phi\|^+_{\co,(E:t,Q:\cQ)\to (E:t,P:\cP)}\,,
\end{align}
where the infimum is over any finite-dimensional reference system $E$.
\end{definition}

\begin{remark}
    A special instance of this completely bounded co-quasi-norm had been, concurrently to \cite{Kochanowski.2026}, investigated in \cite{Li.2026} where $\cQ=(q), \cP=(p)$ were single-index Schatten spaces with $q=1, p=\alpha\in[\frac{1}{2},1]$. In that peculiar case they define their $\|\cdot\|^+_{\cb,\co,1\to\alpha}$, which they call ``completely bounded quasi-norm'' by a minimization over input states normalized in $\|\cdot\|_{(\alpha,1)}$-quasi-norm, which thanks to \Cref{cor:partial.trace.1} can be expressed in closed form with a partial trace. They then arrive at the same $\alpha\geq\frac12$ condition via an operator convexity requirement.
    
\noindent This \Cref{def:cb.co.quasi.norm} is the multi-indexed extension of the completely bounded co-quasi-norm of \cite{Kochanowski.2026} and hence also \cite{Li.2026}.
\end{remark}

\noindent In exactly the same way as for the (completely bounded) quasi-norm, we can show that for the co-quasi-norm the reference system index $t$ can be chosen arbitrarily, as long as it is compatible with both $\cQ,\cP$. In particular, this means that the above quantity is well defined.

\begin{lemma}\label{lem:cb.co.norm.via.t}
Let $\cQ,\cP$ be multi-indexed Schatten spaces on systems $Q,P$ and $\Phi:Q\to P$ be a linear map, then its completely bounded co-quasi-norm defined in \Cref{def:cb.co.quasi.norm} is independent of the index $t$ put on the reference system. In particular for any $q,p\neq0$ which are compatible (and thus of the same sign) and compatible with both $\cQ,\cP$ it holds that
\begin{align}
    \|\Phi\|^+_{\cb,\co,(Q:\cQ)\to (P:\cP)}=\inf_E\|\id_E\otimes\Phi\|^+_{\co,(E:q,Q:\cQ)\to (E:q,P:\cP)}=\inf_E\|\id_E\otimes\Phi\|^+_{\co,(E:p,Q:\cQ)\to (E:p,P:\cP)}.
\end{align}
Moreover let $\cE$ be a multi-indexed Schatten space on systems $E$ which is compatible with both $\cQ$ and $\cP$, then it holds that
\begin{align}
    \|\Phi\|^+_{\cb,\co,(Q:\cQ)\to (P:\cP)}=\inf_E\|\id_{E}\otimes\Phi\|^+_{\co,(E:\cE,Q:\cQ)\to (E:\cE,P:\cP)}.
\end{align}
\end{lemma}

\begin{proof}
Exactly as in the proof of \Cref{lem:cb.norm.via.t}, when taking infima instead of suprema, one shows that
\begin{align}
    \|\id_E\otimes\Phi\|^+_{\co,(E:q,\cQ)\to (E:q,\cP)}=\|\id_E\otimes\Phi\|^+_{\co,(E:p,\cQ)\to (E:p,\cP)},
\end{align} which upon taking the infimum over $E$ yields the first claim of the lemma. The second one follows analogously to \Cref{lem:cb.norm.via.t}.
\end{proof}

\subsection{Swapping Systems: The Non-Commutative Minkowski Inequality}\label{Lin.2}

Having a notion of completely bounded quasi-norm and co-quasi-norm, we may extend a seminal result concerning the swapping of adjacent systems to the (quasi-normed) multi-indexed setting. In particular, we show that the swap map, which swaps two adjacent systems, is a complete contraction in the sense of \Cref{thm:QuasiNormMinkowskiInequality} below.

\begin{theorem}[The non-commutative Minkowski inequality]\label{thm:QuasiNormMinkowskiInequality}
Let $\cR$ be a multi-indexed Schatten space on systems $R$ and $q\leq p$ compatible indices, respectively, on systems $Q,P$, which are compatible with $\cR$. Then it holds that
\begin{align}
    \|\SWAP_{Q\leftrightarrow P}\otimes\id_{R}\|_{\cb,(Q:q,P:p,R:\cR)\to (P:p,Q:q,R:\cR)}\leq 1
\end{align} and 
\begin{align}
    \|\SWAP_{Q\leftrightarrow P}\otimes\id_{R}\|^+_{\cb,\co,(P:p,Q:q,R:\cR)\to (Q:q,P:p,R:\cR)}\geq 1.
\end{align}
\end{theorem}

\begin{remark}
    This is one of the most important tools in this line of work. It is central to the seminal multiplicativity statements of \cite{Devetak.2006}, among others, and will feature in essentially all proofs in the rest of this section. It is the most general statement of this form one could expect to hold and the versions of \cite{book:Pisier.1998, Devetak.2006}, used in the subsequent \cite{Fawzi.2026, Beigi.2023}, are special cases of this formula restricted to indices $q,p\geq 1$ and without the inner systems $R$. In \cite{Kochanowski.2026} these had been extended to compatible quasi-normed indices $0<q \leq p$, but still without the inner indices $\cR$.
    
\noindent The statement is equivalent to and should be read as
    \begin{align}
         \|X_{EPQR}\|_{(E:\cE,P:p,Q:q,R:\cR)} \leq  \|X_{EQPR}\|_{(E:\cE,Q:q,P:p,R:\cR)}
    \end{align} where $\cE,\cR$ are arbitrary multi-indexed Schatten spaces which are compatible with $q,p$ and each other.
\end{remark}

Before we give a proof of \Cref{thm:QuasiNormMinkowskiInequality} we state a direct consequence.
\begin{corollary}\label{cor:equal.adjacient.indices}
Swapping adjacent equal indices does not change the quasi-norm. Let $\cE,\cR$ be multi-indexed Schatten spaces compatible with $p$ and each other, then
\begin{align}
    \|X_{EPQR}\|_{(E:\cE,P:p,Q:p,R:\cR)} = \|X_{EQPR}\|_{(E:\cE,Q:p,P:p,R:\cR)}.
\end{align}
\end{corollary}

\begin{proof}[Proof of \Cref{cor:equal.adjacient.indices}]
While this is also a consequence of the local unitary invariance and the fact that by definition equal adjacent indices are viewed as just one index on the larger system, the proof via two applications of the non-commutative Minkowski inequality is very elegant:
\begin{align}
  \|X_{EPQR}\|_{(E:\cE,P:p,Q:p,R:\cR)} \leq  \|X_{EQPR}\|_{(E:\cE,Q:p,P:p,R:\cR)} \leq \|X_{EPQR}\|_{(E:\cE,P:p,Q:p,R:\cR)}.
\end{align}
\end{proof}

\begin{proof}[Proof of \Cref{thm:QuasiNormMinkowskiInequality}]
While there are several ways to prove this, to save us some time, we will first prove a multi-indexed operator space version of this statement, which we will then lift directly to the result claimed in the theorem by means of \Cref{thm:Multi.Thm.Pisier.like}.
To this end we first see that by \Cref{lem:reductionLemma} (also in \cite[Theorem 4.1]{Fawzi.2026}, which will be independently proven relying only on \Cref{thm:main.multi.pisier} and \Cref{cor.multi.tensor.multiplicativity}) and the complete contractiveness of the SWAP by \cite[Corollary 1.10]{book:Pisier.1998} we have for any $1\leq q\leq p\leq\infty$ and $\cX$ any operator space, that
\begin{align}\label{equ:op.swap.reduction}
\|\id_{E}\otimes\SWAP_{Q\leftrightarrow P}\otimes\id_{\cX}\|&_{(E:\cE,Q:q,P:p;\cX)\to (E:\cE,P:p,Q:q;\cX)} \\&\leq
    \|\SWAP_{Q\leftrightarrow P}\otimes\id_{\cX}\|_{\cb,(Q:q,P:p;\cX)\to (P:p,Q:q;\cX)} \\&= \|\SWAP\|_{\cb,(Q:q,P:p)\to (P:p,Q:q)} \\&\leq 1.
\end{align}
Now let $m$ be a compatible upper bound to all indices in $(\cE,q,p,\cR)$, then by \Cref{thm:Multi.Thm.Pisier.like} we have
\begin{align}
    \|X\|_{(E:\cE,Q:q,P:p,R:\cR)} &=\inf_{A,B}\|AA^*\|^\frac{1}{2}_{(E:\cS_{\cE},Q:r_q,P:r_p,R:\cS_{\cR})}\|B^*B\|^\frac{1}{2}_{(E:\cS_{\cE},Q:r_q,P:r_p,R:\cS_{\cR})}\|A^{-1}XB^{-1}\|_m \\
    &\geq \inf_{A,B}\|AA^*\|^\frac{1}{2}_{(E:\cS_{\cE},P:r_p,Q:r_q,R:\cS_{\cR})}\|B^*B\|^\frac{1}{2}_{(E:\cS_{\cE},P:r_p,Q:r_q,R:\cS_{\cR})}\|A^{-1}XB^{-1}\|_m \\
    &= \|X\|_{(E:\cE,P:p,Q:q,R:\cR)}, 
\end{align} where the inequality follows from above since $q\leq p$ implies $r_q\leq r_p$, as $\frac{1}{r_q}=\frac{1}{q}-\frac{1}{m}$ and $\frac{1}{r_p}=\frac{1}{p}-\frac{1}{m}$, as well as from the unitary invariance of the Schatten-$m$-quasi-norm. Here we denote by $\cS_{\cR}$ the multi-indexed Schatten space with indices $\frac{1}{s_i}=\frac{1}{r_i}-\frac{1}{m}$, when $\cR=(r_1,\dots,r_n)$ and likewise for $\cS_{\cE}$. Due to the compatibility assumption we have $s_i\in[1,\infty]$, hence $\cS_{\cR}, \cS_{\cE}$ are normed (operator-valued) multi-indexed Schatten spaces, and \eqref{equ:op.swap.reduction} applies with $\cX=\cS_{\cR}$. This works both for positive $0<q\leq p$ and, using \Cref{cor:multi.Pisier.like.negative} with $m\in[-1,0)$, for negative $q\leq p<0$ and $X>0$.
\end{proof}

\subsection{Removing Inner Indices: Two Important Chain Rules}\label{sec:3.removing.indices}

In this subsection we consider mixed norms and co-quasi-norms for maps of the form $\Phi\otimes\id$, where the $\id$ acts on the innermost (or `rightmost') system with unchanging index. We show when this extra identity `to the right' does not change the value of the global mixed quasi-norm, respectively co-quasi-norm, and discuss some implications of this fact. Equalities of this type are central to proving multiplicativity statements involving mixed norms, and essentially all multiplicativity and additivity results of \Cref{sec:ApplictoQI} build on one of the following two results. 

The first instance of an equality of this type was \cite[Lemma 5]{Devetak.2006}, which the authors dubbed `an important Lemma', and which led to multiplicativity of $\cb,q\to p$ norms, finding applications in \cite{Devetak.2006, Beigi.2016, Bardet.2022, Beigi.2023} among others.
That was the statement that for a CP map $\Phi:Q\to P$ and indices $q,p,r$ we have
\begin{align}
    \|\Phi\otimes\id_R\|_{(Q:q,R:r)\to (P:p,R:r)} = \|\Phi\|_{(Q:q)\to (P:p)}.
\end{align}

In \cite[Theorem 4.3]{Kochanowski.2026} we extended the above equality to any compatible $q,p\in[\frac{1}{2},\infty]$ and the maximal range $r\in[1,\infty]$ and in \cite[Theorem 4.1]{Fawzi.2026} we extended the above equality to two-indexed input and output norms.

\subsubsection{The Normed Version}
In the following we present a multi-indexed generalization of these statements.

\begin{theorem}\label{lem:reductionLemma}
Let $\cQ, \cP$ be normed multi-indexed Schatten spaces on systems $Q,P$ compatible with $\infty$ and $\cX$ an operator space. Then for $\Phi:Q\to P$ a CP map it holds that
    \begin{align}\label{equ:id.right.reduction}
        \|\Phi\otimes\id_{\mathcal{X}}\|_{(Q:\cQ;\mathcal{X})\to(P:\cP;\mathcal{X})} = \|\Phi\|^+_{(Q:\cQ)\to (P:\cP)}\,,
    \end{align} 
    where the superscript $^+$ denotes optimization over positive semidefinite operators. Here $\id_{\cX}:\cX\to \cX$ is the identity map.
\end{theorem}

\begin{remark}
This is a generalization of \cite[Lemma 5]{Devetak.2006}, as well as \cite[Theorem 4.1]{Fawzi.2026}, to arbitrary multi-indexed Schatten spaces $\cQ,\cP$ from the one-, respectively two-indexed cases discussed there. The proof of this statement is a natural extension of the proof idea of the former, building manifestly on our multi-indexed extension of Pisier's formula \Cref{thm:main.multi.pisier}. 

\noindent For readers uncomfortable with operator spaces, $\cX$ can be replaced here by any normed multi-indexed Schatten space with indices in $[1,\infty]$.

\noindent A very natural question one might have is whether the parameter range presented in this theorem is optimal, and indeed this is true. The requirements we have currently are s.t. $(Q:\cQ;\mathcal{X}), (P:\cP;\mathcal{X})$ are well defined and so cannot be weakened. The assumption on $\cX$ to be normed is indeed optimal, as in \cite[Counterexample 4.1]{Kochanowski.2026} it is shown via a simple counterexample that for $\cX=\cS_t$ with $t<1$, \eqref{equ:id.right.reduction} no longer holds. When replacing $\cX$ by some normed multi-indexed Schatten space $\cT$, however, one would expect the most general, in terms of indices, result to be for $\cQ,\cP$ which are each compatible with $\cT$. 
For the two-indexed version of this more general result see \cite[Theorem 4.3]{Kochanowski.2026}.

\noindent In the context of quantum information theory this result can be interpreted as an entropic chain rule, see e.g.\ \cite[Corollaries 4.2, 4.3, and Remark 4.4]{Kochanowski.2026}. We will discuss this and its consequences in terms of (conditional) entropies later in \Cref{sec:ApplictoQI} and shall use it to prove several chain rules and multiplicativity statements, in particular about $\alpha>1$ minimal and maximal output entropies and conditional entropies.
\end{remark}

An important consequence of the proof is that completely positive maps achieve their maximal mixed multi-indexed Schatten norms on positive operators.
\begin{corollary}\label{cor:mixed.positivity.normed}
    Let $\Phi$ be a CP map between two multi-indexed Schatten spaces $\cQ,\cP$ with $q_i,p_j>0$, then
     \begin{align}
        \|\Phi\|_{(Q:\cQ)\to (P:\cP)}=\|\Phi\|^+_{(Q:\cQ)\to (P:\cP)}.
    \end{align}
\end{corollary}

\begin{remark}
    This generalizes the corresponding results for mixed Schatten norms ($q\to p$, $q,p\geq1$) \cite{Watrous.2004} and completely bounded mixed Schatten norm $(\cb,q\to p)$, $q,p\geq1$ \cite{Devetak.2006}, as well as $\cb,q\to (r,s)$ norms, $q,r,s\geq1$ \cite[Lemma A.2]{Fawzi.2026}. It is thus the most general statement of this kind, since for negatively indexed input spaces the quantity is trivially $0$ without the positivity restriction.
\end{remark}

\begin{proof}[Proof of \Cref{lem:reductionLemma}]
First of all notice that $\|\Phi\otimes\id_{\mathcal{X}}\|_{(Q:\cQ;\mathcal{X})\to(P:\cP;\mathcal{X})} \geq \|\Phi\|^+_{(Q:\cQ)\to (P:\cP)}$, since we can just restrict the supremum on the left-hand side to product operators. In fact, consider $X_{Q}\otimes Y_{\mathcal{X}}$, then by multiplicativity of the operator-valued Schatten norms \Cref{cor:Pisier.multi.tensor.splitting}, we obtain both
\begin{align} 
\|(\Phi\otimes\id_{\mathcal{X}})(X_{Q}\otimes Y_{\mathcal{X}})\|_{(P:\cP;\mathcal{X})} &= \|\Phi(X_{Q})\otimes Y_{\mathcal{X}}\|_{(P:\cP;\mathcal{X})}= \|\Phi(X_{Q})\|_{(P:\cP)}\cdot \|Y_{\mathcal{X}}\|_{\mathcal{X}}, \\ 
\|X_{Q}\otimes Y_\mathcal{X}\|_{(Q:\cQ;\mathcal{X})} &= \|X_Q\|_{(Q:\cQ)}\cdot\|Y_\mathcal{X}\|_{\mathcal{X}}.
\end{align}

Since $\cX$ is an operator space there exists $\cK$ s.t. $\cX\subset\cB(\cK)$. Denote by $\1_\mathcal{K}$ the identity element of $\mathcal{B}(\cK)$. Then $\1_{\mathcal{K}}$ restricted to the linear subspace $\cX$ is $\1_{\cX}$.
To prove the non-trivial side of the inequality, inspired by \cite[Proof of Lemma 5]{Devetak.2006}, we consider a Kraus representation of $\Phi$: $\Phi(\rho)=\sum^\nu_{i=1}K_i \rho K_i^*$. Then for a given $\rho\in\cS_{\cQ}[\otimes_i^{|\cQ|}\cH_i;\cX]$, there exist for any $\epsilon>0$, by our extension of Pisier's formula in \Cref{thm:main.multi.pisier} for this norm, operators $A,B,Y$, s.t. $\rho=(A_{Q}\otimes\1_\mathcal{X})Y_{Q\cX}(B_{Q}\otimes \1_\mathcal{X})$ and $ \|\rho\|_{(Q:\cQ;\mathcal{X})}\geq \|AA^*\|^{\frac{1}{2}}_{(Q:\cQ)}\|B^*B\|^{\frac{1}{2}}_{(Q:\cQ)}\|Y\|_{(Q:\infty;\mathcal{X})}-\epsilon$.
We have \begin{align}
    (\Phi\otimes\id_\mathcal{X})(\rho)& = \sum_{i=1}^\nu(K_i\otimes\1_\mathcal{X})\rho(K_i^*\otimes\1_\mathcal{X})=\sum_{i=1}^\nu(K_iA\otimes\1_\mathcal{X})Y(BK_i^*\otimes\1_\mathcal{X}) \\&= V_A(\1_{\mathbb{C}^\nu}\otimes Y)V^*_B,
\end{align} where $V_A= (K_1A\otimes\1_\mathcal{K},K_2A\otimes\1_\mathcal{K},...,K_\nu A\otimes\1_\mathcal{K})$ is a block-row vector with blocks $K_iA\otimes \1_\mathcal{K}$, and $V_B^*$ a block-column vector with blocks $BK_i^*\otimes\1_\mathcal{K}$, where we denote with $\nu$ the number of blocks and with $N$ the system in which these blocks live and on which $\1_{\mathbb{C}^\nu}$ acts.  
These operators $V_A, V_B$ can be embedded into the space $\mathcal{B}(\mathbb{C}^\nu\otimes\bigotimes_{i=1}^n\mathcal{H}_{Q_i}\otimes\mathcal{K},\mathbb{C}^\nu\otimes\bigotimes_{i=1}^m\mathcal{H}_{P_i}\otimes\mathcal{K})$ by padding suitably with rows of $0$ operators. For $V^*_A, V^*_B$ similarly into $\mathcal{B}(\mathbb{C}^\nu\otimes\bigotimes_{i=1}^m\mathcal{H}_{P_i}\otimes\mathcal{K},\mathbb{C}^\nu\otimes\bigotimes_{i=1}^n\mathcal{H}_{Q_i}\otimes\mathcal{K})$, by padding with columns of $0$ operators. Call these extended operators, respectively, $V_A^\prime, V_B^\prime, V_A^{\prime*} , V_B^{\prime*}$. We get $V^\prime_A=\sum^{\nu}_{i,j=1}\delta_{i1}|i\rangle\langle j|_N\otimes K_jA\otimes\1_\mathcal{K} \in \mathcal{B}(\mathbb{C}^\nu\otimes\bigotimes_{i=1}^n\mathcal{H}_{Q_i}\otimes\mathcal{K},\mathbb{C}^\nu\otimes\bigotimes_{i=1}^m\mathcal{H}_{P_i}\otimes\mathcal{K})$. 
Hence it holds that $V^\prime_AV_A^{\prime *}=|1\rangle\langle1|_N\otimes V_AV_A^*$, similarly for $B$ and $V^\prime_A(\1_N\otimes Y)V_B^{\prime *}=|1\rangle\langle 1|_N\otimes (\Phi\otimes\id_{\mathcal{X}})(\rho)$. 
Now using \Cref{thm:main.multi.pisier} 
we get, for any $r\geq1$,
\begin{align}
    \|(\Phi\otimes\id_{\mathcal{X}})(\rho_{Q;\cX})\|_{(P:\cP;\mathcal{X})} &= \||1\rangle\langle 1|_N\otimes(\Phi\otimes\id_\mathcal{X})(\rho_{Q;\cX})\|_{(N:r,P:\cP;\mathcal{X})} \\ 
    &= \|V^\prime_A(\1_N\otimes Y)V_B^{\prime *}\|_{(N:r,P:\cP;\mathcal{X})}  
    \\ &\leq \|V^\prime_AV_A^{\prime *}\|^{\frac{1}{2}}_{(N:r,P:\cP)}\|V^\prime_BV_B^{\prime *}\|^{\frac{1}{2}}_{(N:r,P:\cP)}\|\1_N\otimes Y\|_{(N:\infty,Q:\infty;\mathcal{X})} \\ &= \|V_AV_A^{*}\|^{\frac{1}{2}}_{(P:\cP)}\|V_BV_B^{*}\|^{\frac{1}{2}}_{(P:\cP)}\|\1_N\otimes Y\|_{(N:\infty,Q:\infty;\mathcal{X})} \\ &= \|\sum_iK_iAA^*K_i^*\|^{\frac{1}{2}}_{(P:\cP)}\|\sum_iK_iB^*BK_i^*\|^{\frac{1}{2}}_{(P:\cP)}\|Y\|_{(Q:\infty;\mathcal{X})} \\&= \|\Phi(AA^*)\|^{\frac{1}{2}}_{(P:\cP)}\|\Phi(B^*B)\|^{\frac{1}{2}}_{(P:\cP)} \|Y\|_{(Q:\infty;\mathcal{X})} \\&\leq  \|\Phi\|^+_{(Q:\cQ)\to(P:\cP)}\|AA^*\|^{\frac{1}{2}}_{(Q:\cQ)}\|B^*B\|^{\frac{1}{2}}_{(Q:\cQ)}\|Y\|_{(Q:\infty;\mathcal{X})}
    \\&\leq \|\Phi\|^+_{(Q:\cQ)\to(P:\cP)}(\|\rho\|_{(Q:\cQ;\mathcal{X})}+\epsilon).
\end{align} In the first, third, and fourth equality, we used the multiplicativity of the operator-valued Schatten norms~\Cref{cor:Pisier.multi.tensor.splitting}.
In the second line, we combined systems of equal indices. Since $\epsilon$ was arbitrary the claim follows. 
\end{proof} 

\begin{proof}[Proof of \cref{cor:mixed.positivity.normed}]
Clearly $\|\Phi\|_{(\cQ)\to (\cP)}\geq \|\Phi\|^+_{(\cQ)\to (\cP)}$. For the alternate inequality we repeat more or less the same argument as above, with the difference that we are not confined to the normed regime.
Let $\rho\in\cS_{\cQ}$, then by \cref{thm:Multi.Thm.Pisier.like} with $s=\infty$ we have for any $\epsilon>0$ that $\|\rho\|_{(Q:\cQ)}\geq \|AA^*\|^{\frac{1}{2}}_{(Q:\cQ)}\|B^*B\|^{\frac{1}{2}}_{(Q:\cQ)}\|Y\|_{(Q:\infty)}-\epsilon$ for some suitable $A,B,Y$.
We take the same construction for
\begin{align}
    \Phi(\rho)= \sum_{i=1}^\nu K_i\rho K_i^* = V_A(\1_{\mathbb{C}^\nu}\otimes Y)V_B^*
\end{align} with $V_A= (K_1A,K_2A,...,K_\nu A)$ is a block-row vector with blocks $K_iA$, and $V_B^*$ a block-column vector with blocks $BK_i^*$, and $V_A^\prime, V_B^{\prime\,*}$ their matrix embeddings. Repeating more or less the above argument, but now for some $r>0$ compatible with $\cP$, we have
\begin{align}
    \|\Phi(\rho)\|_{(\cP)} &= \|\Phi(\rho_{Q})\|_{(N:r,P:\cP)} = \|V^\prime_A(\1_N\otimes Y)V_B^{\prime *}\|_{(N:r,P:\cP)}\\ &\leq \|V^\prime_AV_A^{\prime *}\|^{\frac{1}{2}}_{(N:r,P:\cP)}\|V^\prime_BV_B^{\prime *}\|^{\frac{1}{2}}_{(N:r,P:\cP)}\|\1_N\otimes Y\|_{(N:\infty,Q:\infty)} \\ &= \|V_AV_A^{*}\|^{\frac{1}{2}}_{(P:\cP)}\|V_BV_B^{*}\|^{\frac{1}{2}}_{(P:\cP)}\|Y\|_{(Q:\infty)} \\ &= \|\sum_iK_iAA^*K_i^*\|^{\frac{1}{2}}_{(P:\cP)}\|\sum_iK_iB^*BK_i^*\|^{\frac{1}{2}}_{(P:\cP)}\|Y\|_{(Q:\infty)} \\&= \|\Phi(AA^*)\|^{\frac{1}{2}}_{(P:\cP)}\|\Phi(B^*B)\|^{\frac{1}{2}}_{(P:\cP)} \|Y\|_{(Q:\infty)} \\&\leq  \|\Phi\|^+_{(Q:\cQ)\to(P:\cP)}\|AA^*\|^{\frac{1}{2}}_{(Q:\cQ)}\|B^*B\|^{\frac{1}{2}}_{(Q:\cQ)}\|Y\|_{(Q:\infty)} \\&\leq \|\Phi\|^+_{(Q:\cQ)\to(P:\cP)}(\|\rho\|_{(Q:\cQ)}+\epsilon),
\end{align}
where the first inequality, for the arbitrary factors $V_A^\prime$, $V_B^{\prime*}$, follows from \eqref{equ:Multi.Thm.Pisier.like.inf} of \Cref{thm:Multi.Thm.Pisier.like}. Since this works for any $\epsilon>0$ the claim follows.   
\end{proof}

\subsubsection{The Co-Quasi-Normed Version}

Having considered the normed setting, let us now consider the analogous question for the quasi-normed setting. As demonstrated in \cite[Counterexample 4.1]{Kochanowski.2026} appending inner identities with indices $<1$ does generically increase the mixed norm. The natural question is hence whether it can also decrease the norm, that is, we consider the same statement for the co-quasi-norm.

In this sense, the following is the natural quasi-normed analogue of the normed statement.

\begin{theorem}\label{thm:id.right.quasi}
Let $\cQ,\cP,\cT$ be multi-indexed Schatten spaces on systems $Q,P,T$ with indices $q_i,p_i,t_i\in[\frac{1}{2},1]$,  respectively, and $\Phi:Q\to P$ a CP map. Then it holds that
\begin{align}
    \|\Phi\otimes\id_T\|^+_{\co,(Q:\cQ,T:\cT)\to (P:\cP,T:\cT)}=\|\Phi\|^+_{\co, (Q:\cQ)\to (P:\cP)}.
\end{align}
\end{theorem}

\begin{remark}
    This is the natural `co' analogue of \Cref{lem:reductionLemma} where, as compared to before, the multi-indexed Schatten spaces involved are all s.t. their indices are in $[\frac{1}{2},1]$. This is the region where the reverse Hölder dual indices are in $[-\infty,-1]$ which via \Cref{prop:negative.multi.indexed} relates back to the normed regime. The proof makes manifest use of this fact. 
    In principle, one could, however, imagine that this statement extends to arbitrary multi-indexed Schatten spaces $\cQ,\cP,\cT$ with indices in $(0,1]$.
    Analogously to \Cref{lem:reductionLemma}, which requires that the inner indices are all $\geq 1$, this statement explicitly requires that the inner indices $\cT$ are $\leq 1$, since otherwise there is a rather simple (geometric) counterexample. 

\noindent    Unlike the previous normed statement \Cref{lem:reductionLemma} this theorem does not build on or generalize existing results in the literature.
\end{remark}

\begin{corollary}\label{cor:cb.id.right.quasi}
    As a direct consequence it follows for $\alpha\in[\frac{1}{2},1]$ and $\Phi:Q\to AB$ that, both
\begin{align}
    &\|\Phi\otimes\id_T\|^+_{\co,(Q:1,T:\alpha)\to (A:1,B:\alpha,T:\alpha)}=\|\Phi\|^+_{\co,(Q:1)\to (A:1,B:\alpha)}, \\
    &\|\Phi\otimes\id_T\|^+_{\cb,\co,(Q:1,T:\alpha)\to (A:1,B:\alpha,T:\alpha)}=\|\Phi\|^+_{\cb,\co,(Q:1)\to (A:1,B:\alpha)}.
\end{align} which are, respectively, equivalent to
\begin{align}
     \inf_{\rho_{TQ}}&\left\{\tilde{H}_\alpha^\uparrow(BT|A)_{(\Phi\otimes\id_T)(\rho_{QT})}-\tilde{H}^\uparrow_\alpha(T|Q)_{\rho_{QT}}\right\}=\inf_{\rho_Q}\tilde{H}^\uparrow_\alpha(B|A)_{\Phi(\rho_Q)}, \\
      \inf_{E\,,\rho_{EQT}}&\left\{\tilde{H}_\alpha^\uparrow(BT|AE)_{(\id_E\otimes\Phi\otimes\id_T)(\rho_{EQT})}-\tilde{H}^\uparrow_\alpha(T|QE)_{\rho_{EQT}}\right\}=\inf_{E\,,\rho_{EQ}}\tilde{H}^\uparrow_\alpha(B|AE)_{(\id_E\otimes\Phi)(\rho_{EQ})}.
\end{align}
\end{corollary}

Let us now prove \Cref{thm:id.right.quasi}; \Cref{cor:cb.id.right.quasi} then follows by applying it to $\Phi$ and to $\id_E\otimes\Phi$ with reference index $1$ (\Cref{lem:cb.co.norm.via.t}), together with \Cref{prop:Extremal:Entropies.Mixed.Norms}.
The idea of the proof of this statement is to use reverse Hölder duality \Cref{cor:reverse.multi.hölder} to relate to a regime which is amenable to the same proof idea as \Cref{lem:reductionLemma}. Technically it is somewhat more involved, however.
\begin{proof}[Proof of \Cref{thm:id.right.quasi}]
    The achievability part of this equality, i.e.\ the ``$\leq$'' inequality, is the easy direction, which follows from plugging in product states, since by \Cref{cor.multi.tensor.multiplicativity} we have
    \begin{align}
        \|(\Phi\otimes\id_T)(X\otimes Y)\|_{(P:\cP,T:\cT)}= \|\Phi(X)\|_{(P:\cP)}\|Y\|_{(T:\cT)}.
    \end{align}
    
\noindent    The central tool in this proof is the reverse Hölder duality \Cref{cor:reverse.multi.hölder} which maps the problem via \Cref{prop:negative.multi.indexed} to the one of showing
    \begin{align}
        \|\Phi^*\otimes\id_T\|^+_{\co,(P:\cP^\prime,T:\cT^\prime)\to (Q:\cQ^\prime,T:\cT^\prime)}=\|\Phi^*\|^+_{\co,(P:\cP^\prime)\to (Q:\cQ^\prime)},
    \end{align} where the Hölder dual indices $q_i^\prime,p_i^\prime,t_i^\prime$ of $q_i,p_i,t_i$ lie in $[-\infty,-1]$. Explicitly this is
\begin{align}
     \|\Phi\otimes\id_T\|^+&_{\co,(Q:\cQ,T:\cT)\to (P:\cP,T:\cT)} \\&= \inf_{X>0} \|(\Phi\otimes\id_T)(X)\|_{(P:\cP,T:\cT)} \|X\|^{-1}_{(Q:\cQ,T:\cT)} \\&=\inf_{X>0}\inf_{Y>0}\Tr[Y(\Phi\otimes\id_T)(X)]\,\|X\|^{-1}_{(Q:\cQ,T:\cT)}\|Y\|^{-1}_{(P:\cP^\prime,T:\cT^\prime)} \\&=\inf_{Y>0}\inf_{X>0}\Tr[(\Phi^*\otimes\id_T)(Y)X]\,\|X^{-1}\|_{(Q:-\cQ,T:-\cT)}\|Y\|^{-1}_{(P:\cP^\prime,T:\cT^\prime)} \\&=\inf_{Y>0}\|(\Phi^*\otimes\id_T)(Y)\|_{(Q:\cQ^\prime,T:\cT^\prime)}\|Y\|^{-1}_{(P:\cP^\prime,T:\cT^\prime)} \\&=\|\Phi^*\otimes\id_T\|^+_{\co,(P:\cP^\prime,T:\cT^\prime)\to (Q:\cQ^\prime,T:\cT^\prime)},
\end{align} which of course also holds without the $\id_T$.
    
\noindent  We will prove this statement now in a manner similar to the previous \Cref{lem:reductionLemma}. 
    To prove the remaining $\geq$ bound first take some $X_{PT}>0$ and observe that by \Cref{thm:main.multi.pisier} and \Cref{prop:negative.multi.indexed} we have
    \begin{align}
        \|X\|_{(P:\cP^\prime,T:\cT^\prime)} &= \|X^{-1}\|^{-1}_{(P:-\cP^\prime,T:-\cT^\prime)} \\
        &=\left(\inf_{A>0}\|A\|_{(-\cP^\prime)}\|A_P^{-\frac{1}{2}}X_{PT}^{-1}A_P^{-\frac{1}{2}}\|_{(P:\infty,T:-\cT^\prime)}\right)^{-1} \\
        &= \sup_{A>0}\|A\|^{-1}_{(-\cP^\prime)}\|A_P^\frac{1}{2}X_{PT}A_P^\frac{1}{2}\|_{(P:-\infty,T:\cT^\prime)}.
    \end{align}
    Thus for any fixed $\epsilon>0$ there exists $A_{P}>0$ s.t.
\begin{align}
    \|X\|_{(P:\cP^\prime,T:\cT^\prime)}&\leq \|A\|^{-1}_{(P:-\cP^\prime)}\|A^\frac{1}{2}_{P}X_{PT}A^\frac{1}{2}_{P}\|_{(P:-\infty,\,T:\cT^\prime)}+\epsilon \\&\equiv \|A\|^{-1}_{(P:-\cP^\prime)}\|Y_{PT}\|_{(P:-\infty,\,T:\cT^\prime)}+\epsilon,
\end{align} where hence $X_{PT}=A^{-\frac{1}{2}}_{P}Y_{PT}A^{-\frac{1}{2}}_{P}$. 

Now similarly to in the previous proof, let $\{K_i\}^\nu_{i=1}$ be the Kraus operators of the map $\Phi^*$, then
\begin{align}
    (\Phi^*\otimes\id_T)(X_{PT})&=\sum^\nu_{i=1}(K_iA^{-\frac{1}{2}}\otimes\1_T)Y_{PT}(A^{-\frac{1}{2}}K_i^*\otimes\1_T) \\&= (W_{A}\otimes\1_T)(\1_N\otimes Y_{PT})(W_A^*\otimes\1_T)
\end{align} where $W_A:=(K_1A^{-\frac{1}{2}}\ K_2A^{-\frac{1}{2}}\ \cdots\ K_\nu A^{-\frac{1}{2}}): NP\to Q$ is a block-row vector with operator-valued entries $K_iA^{-\frac{1}{2}}:P\to Q$ and $W_A^*$ is its adjoint block-column vector. Then it is easy to check that $W_AW_A^*=\Phi^*(A^{-1})$ which we may assume to be strictly positive. Otherwise, if $W_AW_A^*=\Phi^*(A^{-1})\geq0$ has a non-trivial kernel, then $\|\Phi^*(A^{-1})\|_{\cQ^\prime}=0$ and thus $\|\Phi^*\|^+_{\co,(P:\cP^\prime)\to (Q:\cQ^\prime)}=0$ and the inequality is trivially true.
Further, for $Y_{PT}$ and this $W_A:NP\to Q$ we will now show that
\begin{align}
    \|Y_{PT}\|_{(P:-\infty,T:\cT^\prime)} &\leq \|W_AW_A^*\|^{-1}_{(Q:\cQ^\prime)} \|(W_A\otimes \1_T)(\1_N\otimes Y_{PT})(W_A^*\otimes\1_T)\|_{(Q:\cQ^\prime,T:\cT^\prime)} \\ &\equiv \|W_AW_A^*\|^{-1}_{(Q:\cQ^\prime)} \|Z_{QT}\|_{(Q:\cQ^\prime,T:\cT^\prime)}.
\end{align} holds, where we set $Z_{QT}:=(W_A\otimes \1_T)(\1_N\otimes Y_{PT})(W_A^*\otimes\1_T)$. To see this set $S:=(W_AW_A^*)^{-\frac{1}{2}}>0$ and define $V:=SW_A$, then $VV^*=\1_Q$ and hence we can apply the operator Jensen's inequality \cite{Choi.1974, HansenPedersen.2003} to get
\begin{align}
    (SZS)^{-1} = [(V\otimes\1_T)(\1_N\otimes Y_{PT})(V^*\otimes\1_T)]^{-1} \leq (V\otimes\1_T)(\1_N\otimes Y_{PT}^{-1})(V^*\otimes\1_T).
\end{align}
Thus it follows that
\begin{align}
    \|SZS\|^{-1}_{(Q:-\infty,T:\cT^\prime)} &= \|[(V\otimes\1)(\1_N\otimes Y_{PT})(V^*\otimes\1_T)]^{-1}\|_{(Q:\infty,T:-\cT^\prime)} \\&\leq \|(V\otimes\1_T)(\1_N\otimes Y_{PT}^{-1})(V^*\otimes\1_T)\|_{(Q:\infty,T:-\cT^\prime)} \\&\leq \|VV^*\|_{(Q:\infty)}\|\1_N\otimes Y_{PT}^{-1}\|_{(N:\infty,P:\infty,T:-\cT^\prime)} \\ &= \| Y_{PT}^{-1}\|_{(P:\infty,T:-\cT^\prime)} =  \|Y_{PT}\|^{-1}_{(P:-\infty,T:\cT^\prime)}.
\end{align}
Finally, to relate the left-hand side to the desired quantity, we can use the multi-indexed version of Pisier's theorem \Cref{thm:main.multi.pisier} in the following manner:
\begin{align}
    \|(SZS)\|^{-1}_{(Q:-\infty,T:\cT^\prime)} &= \|(S\otimes\1_T)^{-1}Z_{QT}^{-1}(S\otimes\1_T)^{-1}\|_{(Q:\infty,T:-\cT^\prime)} \\ &\geq \|SS^*\|_{(Q:-\cQ^\prime)}^{-1}\|Z_{QT}^{-1}\|_{(Q:-\cQ^\prime,T:-\cT^\prime)} \\ &= \|W_AW_A^*\|_{(Q:\cQ^\prime)}\|(W_A\otimes\1_T)(\1_N\otimes Y_{PT})(W_A^*\otimes\1_T)\|^{-1}_{(Q:\cQ^\prime,T:\cT^\prime)}.
\end{align}
Combining these last two inequalities yields the claimed inequality
\begin{align}
  \|Y_{PT}\|_{(P:-\infty,T:\cT^\prime)}  
  \leq \|W_AW_A^*\|^{-1}_{(Q:\cQ^\prime)}\|(W_A\otimes\1_T)(\1_N\otimes Y_{PT})(W_A^*\otimes\1_T)\|_{(Q:\cQ^\prime,T:\cT^\prime)}
\end{align}
Now we can use these to finish the proof of the theorem, analogously to the proof of \Cref{lem:reductionLemma}. We have
\begin{align}
    \|(\Phi^*\otimes\id_T)(X_{PT})\|_{(Q:\cQ^\prime,T:\cT^\prime)} &= \|(W_A\otimes\1_T)(\1_N\otimes Y_{PT})(W_A^*\otimes\1_T)\|_{(Q:\cQ^\prime,T:\cT^\prime)} \\&\geq \|W_AW_A^*\|_{(Q:\cQ^\prime)} \|Y_{PT}\|_{(P:-\infty,T:\cT^\prime)} \\&= \|\Phi^*(A^{-1})\|_{(Q:\cQ^\prime)}\|Y_{PT}\|_{(P:-\infty,T:\cT^\prime)} \\ &\geq \|\Phi^*\|^+_{\co,(P:\cP^\prime)\to (Q:\cQ^\prime)}\|A^{-1}\|_{(P:\cP^\prime)}\|Y_{PT}\|_{(P:-\infty,T:\cT^\prime)} \\ &= \|\Phi^*\|^+_{\co,(P:\cP^\prime)\to (Q:\cQ^\prime)}\|A\|^{-1}_{(P:-\cP^\prime)}\|Y_{PT}\|_{(P:-\infty,T:\cT^\prime)} \\
    &\geq \|\Phi^*\|^+_{\co,(P:\cP^\prime)\to (Q:\cQ^\prime)} (\|X_{PT}\|_{(P:\cP^\prime,\, T:\cT^\prime)}-\epsilon).
\end{align} Since this works for any $\epsilon>0$, by taking the infimum over $X_{PT}$ this proves the claim of the theorem.
\end{proof}
Just like the previous statement, \Cref{thm:id.right.quasi} can be interpreted as an entropic chain rule, see \Cref{cor:cb.id.right.quasi}.

\subsection{Some Multiplicativity Statements of Mixed Norms}\label{sec:Mixed.Multiplicativity.Statements}

Having established the non-commutative Minkowski inequality and the two chain rules \Cref{lem:reductionLemma} and \Cref{thm:id.right.quasi} we now apply them to derive multiplicativity statements and chain rules of mixed (co)-(quasi-)norms under tensor products or suitable composition of linear maps. To this end we first require a helpful operator-level multiplicativity.

\subsubsection{A Helpful Operator Multiplicativity}

As a consequence of \Cref{thm:QuasiNormMinkowskiInequality} and \Cref{thm:multi.quasi.relational.consistency} we first derive the following multiplicativity result, which will be helpful in achievability proofs of \Cref{sec:Mult.cb.reverse.hyper} and \Cref{sec:ApplictoQI}. 

\begin{lemma}[An achievability lemma]\label{lem:basic.4.index.multiplicativity}
    Let $X_{Q_1P_1}, Y_{Q_2P_2}$ be operators and $q,p\neq0$ compatible indices, then it holds that
    \begin{align}
        \|X_{Q_1P_1}\otimes Y_{Q_2P_2}\|_{(Q_1:q,Q_2:q,P_1:p,P_2:p)} = \|X\|_{(Q_1:q,P_1:p)}\|Y\|_{(Q_2:q,P_2:p)}.
    \end{align}
    When $q,p<0$ we assume $X,Y>0$.
\end{lemma}

\begin{remark}
    For $q=p$ this is obvious by the unitary invariance of the Schatten quasi-norms. 
    A well-known special case of this statement is the $q=1,p=\alpha\in[\frac{1}{2},\infty]$ case, which corresponds to the additivity of the optimized sandwiched Rényi conditional entropy $\tilde{H}_\alpha^\uparrow$ under tensor-product states \cite[Corollary 5.9]{book:Tomamichel.2016}. We will be using this lemma in similar ways, in particular to prove the additivity of the sandwiched Rényi umlaut information under tensor-product states, \Cref{prop.add.sand.umlaut}, and several achievability directions of mixed-norm multiplicativity statements in the remainder of this work.
\end{remark}

\begin{proof}
    Without loss of generality we may assume $q,p>0$ else the result follows with \cref{prop:negative.multi.indexed}.
    Assume $q\leq p$, then on the one hand via \Cref{thm:QuasiNormMinkowskiInequality} and \Cref{cor.multi.tensor.multiplicativity} we have
    \begin{align}
         \|X_{Q_1P_1}\otimes Y_{Q_2P_2}\|_{(Q_1:q,Q_2:q,P_1:p,P_2:p)} &\geq \|X_{Q_1P_1}\otimes Y_{Q_2P_2}\|_{(Q_1:q,P_1:p,Q_2:q,P_2:p)}  \\&= \|X\|_{(Q_1:q,P_1:p)}\cdot\|Y\|_{(Q_2:q,P_2:p)}.
    \end{align}
    On the other hand, via \Cref{thm:multi.quasi.relational.consistency} and \Cref{cor.multi.tensor.multiplicativity} we have
    \begin{align}
        \|X_{Q_1P_1}\otimes Y_{Q_2P_2}\|_{(Q_1:q,Q_2:q,P_1:p,P_2:p)} &=\inf_{A,B\geq0}\|A\|_{2r}\|B\|_{2r}\|A_{Q_1Q_2}^{-1}(X_{Q_1P_1}\otimes Y_{Q_2P_2})B^{-1}_{Q_1Q_2}\|_p \\
        &\leq \inf_{A^{}, A^{\prime}\geq0}\|A^{}\|_{2r}\|A^{\prime}\|_{2r}\inf_{B^{}, B^{\prime}\geq0}\|B^{}\|_{2r}\|B^{\prime}\|_{2r}\\ &\quad\times\|(A^{\,-1}_{Q_1}\otimes A_{Q_2}^{\prime\,-1})(X_{Q_1P_1}\otimes Y_{Q_2P_2})(B^{\,-1}_{Q_1}\otimes B_{Q_2}^{\prime\,-1})\|_p \\
        &=\inf_{A^{}, A^{\prime}\geq0}\|A^{}\|_{2r}\|A^{\prime}\|_{2r}\inf_{B^{}, B^{\prime}\geq0}\|B^{}\|_{2r}\|B^{\prime}\|_{2r}\\ &\quad\times\|A^{\,-1}_{Q_1}X_{Q_1P_1}B^{\,-1}_{Q_1}\otimes  A_{Q_2}^{\prime\,-1}Y_{Q_2P_2}B_{Q_2}^{\prime\,-1}\|_p \\
        &= \|X_{Q_1P_1}\|_{(Q_1:q,P_1:p)}\|Y_{Q_2P_2}\|_{(Q_2:q,P_2:p)},
    \end{align} where $\frac{1}{r}=|\frac{1}{q}-\frac{1}{p}|$ and in the last line we used the unitary invariance and multiplicativity under tensor products of the Schatten-$p$-quasi-norms. We dropped writing out the support conditions in the first line $\Pi_A=L_{X\otimes Y}, \Pi_B=R_{X\otimes Y}$. They imply in the second line that $\Pi_A=L_X, \Pi_{A^\prime}=L_Y$ and  $\Pi_B=R_X, \Pi_{B^\prime}=R_Y$.
    The other case goes similarly. Assume $q\geq p$, then on the one hand via \Cref{thm:QuasiNormMinkowskiInequality} and \Cref{cor.multi.tensor.multiplicativity} we have
    \begin{align}
         \|X_{Q_1P_1}\otimes Y_{Q_2P_2}\|_{(Q_1:q,Q_2:q,P_1:p,P_2:p)} &\leq \|X_{Q_1P_1}\otimes Y_{Q_2P_2}\|_{(Q_1:q,P_1:p,Q_2:q,P_2:p)}  \\&= \|X\|_{(Q_1:q,P_1:p)}\cdot\|Y\|_{(Q_2:q,P_2:p)}.
    \end{align}
    On the other hand, via \Cref{thm:multi.quasi.relational.consistency} and \Cref{cor.multi.tensor.multiplicativity} we have
    \begin{align}
        \|X_{Q_1P_1}\otimes Y_{Q_2P_2}\|_{(Q_1:q,Q_2:q,P_1:p,P_2:p)} \!&=\!\sup_{A,B\geq0}\!\|A\|^{-1}_{2r}\|B\|^{-1}_{2r}\|A_{Q_1Q_2}^{}(X_{Q_1P_1}\otimes Y_{Q_2P_2})B^{}_{Q_1Q_2}\|_p \\
        &\geq\! \sup_{A^{}, A^{\prime}\geq0}\|A^{}\|^{-1}_{2r}\|A^{\prime}\|^{-1}_{2r}\sup_{B^{}, B^{\prime}\geq0}\|B^{}\|^{-1}_{2r}\|B^{\prime}\|^{-1}_{2r}\\ &\quad\times\|(A^{}_{Q_1}\otimes A_{Q_2}^{\prime})(X_{Q_1P_1}\otimes Y_{Q_2P_2})(B^{}_{Q_1}\otimes B_{Q_2}^{\prime})\|_p \\
        &=\!\sup_{A^{}, A^{\prime}\geq0}\|A^{}\|^{-1}_{2r}\|A^{\prime}\|^{-1}_{2r}\sup_{B^{}, B^{\prime}\geq0}\|B^{}\|^{-1}_{2r}\|B^{\prime}\|^{-1}_{2r}\\ &\quad\times\|A^{}_{Q_1}X_{Q_1P_1}B^{}_{Q_1}\otimes  A_{Q_2}^{\prime}Y_{Q_2P_2}B_{Q_2}^{\prime}\|_p \\
        &=\! \|X_{Q_1P_1}\|_{(Q_1:q,P_1:p)}\|Y_{Q_2P_2}\|_{(Q_2:q,P_2:p)}.
    \end{align}
\end{proof}

As a direct corollary we obtain the additivity under tensor-product states of the sandwiched umlaut information introduced in \cite{Kochanowski.2026}.
\begin{corollary}[Claim of \cite{Kochanowski.2026}, Proposition 4.3]\label{prop.add.sand.umlaut}
Let $\alpha\in(0,1)$ and $\rho_{A_1B_1}, \sigma_{A_2B_2}$ be quantum states, then 
\begin{align}
    \tilde{U}_\alpha(A_1A_2;B_1B_2)_{\rho\otimes\sigma}= \tilde{U}_\alpha(A_1;B_1)_{\rho} + \tilde{U}_\alpha(A_2;B_2)_{\sigma}.
\end{align}
\end{corollary}
\begin{proof}
By \cite[Lemma 4.2]{Kochanowski.2026} we have $\tilde{U}_\alpha(A;B)_\rho=-\beta\log\|\rho_A^\frac{1}{2}\rho_{BA}^{\frac{1}{\beta}}\rho_A^\frac{1}{2}\|_{(B:\beta,A:\alpha)}$ where $\beta=\frac{\alpha}{1-\alpha}$, so the claimed statement reduces to showing
\begin{align}
   \|\rho_{A_1}^\frac{1}{2}\rho_{B_1A_1}^{\frac{1}{\beta}}\rho_{A_1}^\frac{1}{2}\otimes \sigma_{A_2}^\frac{1}{2}\sigma_{B_2A_2}^{\frac{1}{\beta}}\sigma_{A_2}^\frac{1}{2}\|_{(B_1:\beta, B_2:\beta,A_1:\alpha,A_2:\alpha)} = \|\rho_{A_1}^\frac{1}{2}\rho_{B_1A_1}^{\frac{1}{\beta}}\rho_{A_1}^\frac{1}{2}\|_{(B_1:\beta,A_1:\alpha)}\|\sigma_{A_2}^\frac{1}{2}\sigma_{B_2A_2}^{\frac{1}{\beta}}\sigma_{A_2}^\frac{1}{2}\|_{(B_2:\beta,A_2:\alpha)}.
\end{align}
This, however, is nothing but \cref{lem:basic.4.index.multiplicativity}.
\end{proof}

\subsubsection{Multiplicativity of Completely Bounded (Reverse) Hypercontractivity}\label{sec:Mult.cb.reverse.hyper}

In \cite[Theorem 1.6]{Kochanowski.2026} we established super-multiplicativity of the completely bounded reverse hypercontractive co-quasi-norm. Here we will extend this result also to negative indices and prove the achievability using the above \cref{lem:basic.4.index.multiplicativity}. We will also extend the completely bounded hypercontractivity of \cite[Theorem 11]{Devetak.2006} to arbitrary indices and give a correct proof of its achievability.

These mixed quasi-norms and moreover the notions of quantum (complete) (reverse) hypercontractivity \cite{Bonami.1970, Beckner.1975, NELSON.1973, Mossel.2013, Olkiewicz.1999, King.2014, Temme.2014} find uses in the study of Markovian dynamics of open quantum systems \cite{Beigi.2016, Bardet.2022, Cubitt.2015, Beigi.2020, Kastoryano.2013, Gross.1975.1} and extremal output Rényi entropies. See also \cite{Watrous.2004, Delgosha.2014}. 

Let us start with the completely bounded mixed Schatten norm $\|\cdot\|_{\cb,q\to p}$ of linear maps. When $1\leq q\leq p$ then a linear map $\Phi$ is called \emph{completely hypercontractive} \cite{Beigi.2016} whenever $\|\Phi\|_{\cb,q\to p}\leq1$.
In the seminal \cite[Theorem 11]{Devetak.2006} it was shown that exactly these ($1\leq q\leq p$) hypercontractive mixed norms are multiplicative under tensor products of CP maps. This has found several applications in the study of dissipative quantum Markov semigroups \cite{Beigi.2016, Bardet.2022}.

Using our multi-indexed Schatten quasi-norms we extend this result from $1\leq q\leq p$ to any pair of compatible non-zero indices $q\leq p$.
\begin{lemma}[Theorem 4.2 of \cite{Kochanowski.2026} and Theorem 11 of \cite{Devetak.2006} extended]\label{thm:mixed.up.cb.additivity}
Let $\Phi:Q_1\to P_1,\Psi:Q_2\to P_2$ be (not necessarily CP) linear maps, then for compatible $q\leq p$ it holds that
\begin{align}\label{eq:multiplicativitynorms}
    \|\Phi\otimes\Psi\|_{\cb,(Q_1Q_2:q)\to (P_1P_2:p)} = \|\Phi\|_{\cb,(Q_1:q)\to (P_1:p)}\cdot \|\Psi\|_{\cb,(Q_2:q)\to (P_2:p)}.
\end{align}
When $q\leq p<0$ we assume $\Phi,\Psi$ to be CP.
\end{lemma}

\begin{remark}
    This result holds explicitly for linear, not necessarily CP, maps and any compatible $q\leq p$, i.e.\ both $0<q\leq p$ and $q\leq p<0$, and thus extends \cite[Theorem 11]{Devetak.2006}, which proved this result for CP maps and $1\leq q\leq p$. When $q\leq p<0$ the maps need to be CP for the quantity to be well defined. 
\end{remark}

\begin{proof}[Proof of \Cref{thm:mixed.up.cb.additivity}]
We first prove the upper bound using \Cref{thm:QuasiNormMinkowskiInequality}, as in \cite[Theorem 11]{Devetak.2006}. Fix some size of the system $E$, then by straightforward submultiplicativity of the ($\cb$) quasi-norm under compositions we get
\begin{align}
    \|\id_E&\otimes\Phi\otimes\Psi\|_{(E:p,Q_1:q,Q_2:q)\to (E:p,P_1:p,P_2:p)}  \\ &\leq \|\id_E\otimes\Phi\otimes\id_{Q_2}\|_{(E:p,Q_1:q,Q_2:q)\to (E:p,P_1:p,Q_2:q)}\\&\qquad\times\|\id_E\otimes\id_{P_1}\otimes\Psi\|_{(E:p,P_1:p,Q_2:q)\to (E:p,P_1:p,P_2:p)} \\ 
    &\overset{(1)}{\leq} \|\id_E\otimes\id_{Q_2}\otimes\Phi\|_{(E:p,Q_2:q,Q_1:q)\to (E:p,Q_2:q,P_1:p)}\\&\qquad\times\|\id_E\otimes\id_{P_1}\otimes\Psi\|_{(E:p,P_1:p,Q_2:q)\to (E:p,P_1:p,P_2:p)} \\ 
    &\overset{(2)}{\leq} \|\Phi\|_{\cb,(Q_1:q)\to (P_1:p)}\|\Psi\|_{\cb,(Q_2:q)\to (P_2:p)}.
\end{align} Taking now $\sup_E$ yields the statement of the lemma thanks to \Cref{lem:cb.norm.via.t}.
Here in $(1)$ we used \Cref{thm:QuasiNormMinkowskiInequality} to swap $P_1$ and $Q_2$ as $p\geq q$ and in $(2)$ \Cref{lem:cb.norm.via.t}. This works equally when $0<q\leq p$ or $q\leq p<0$ as in the latter case we have also
\begin{align}
    \|Y\|_{(E:p,P_1:p,Q_2:q)}= \|Y^{-1}\|^{-1}_{(E:-p,P_1:-p,Q_2:-q)} \leq \|Y^{-1}\|^{-1}_{(E:-p,Q_2:-q,P_1:-p)} = \|Y\|_{(E:p,Q_2:q,P_1:p)}
\end{align}
via \Cref{thm:QuasiNormMinkowskiInequality} with $-q\geq -p$ and the order reversing property of $x\mapsto x^{-1}$.

For the achievability we show that sequences of product input states achieve the above bound. TThis is the presumably the proof idea of \cite{Devetak.2006}, made complete by \cref{lem:basic.4.index.multiplicativity}. Let $E=E_1E_2$ and take operators $X_{E_1Q_1}, Y_{E_2,Q_2}$ with $\|X\|_{(E_1Q_1:q)}=\|Y\|_{(E_2Q_2:q)}=1$. When $q,p<0$ then take $X,Y>0$. Then due to \Cref{lem:basic.4.index.multiplicativity} we have that for
\begin{align}
    (\id_{E_1}\otimes\id_{E_2}\otimes\Phi\otimes\Psi)(X_{E_1Q_1}\otimes Y_{E_2Q_2})\equiv Z_{E_1P_1}^{(1)}\otimes Z^{(2)}_{E_2P_2}
\end{align} it holds that
\begin{align}
    \|Z_{E_1P_1}^{(1)}\otimes Z^{(2)}_{E_2P_2}\|_{(E_1:q,E_2:q,P_1:p,P_2:p)} = \|Z^{(1)}_{E_1P_1}\|_{(E_1:q,P_1:p)}\|Z^{(2)}_{E_2P_2}\|_{(E_2:q,P_2:p)}.
\end{align} Thus we get
\begin{align}
    \|\Phi\otimes\Psi\|_{\cb,(Q_1Q_2:q)\to(P_1P_2:p)} &\geq \|(\id_{E_1}\otimes\id_{E_2}\otimes\Phi\otimes\Psi)(X_{E_1Q_1}\otimes Y_{E_2Q_2})\|_{(E_1E_2:q,P_1P_2:p)} \\
    &=\|Z^{(1)}_{E_1P_1}\|_{(E_1:q,P_1:p)}\|Z^{(2)}_{E_2P_2}\|_{(E_2:q,P_2:p)} \\
    &= \|(\id_{E_1}\otimes\Phi)(X_{E_1Q_1})\|_{(E_1:q,P_1:p)}\|(\id_{E_2}\otimes\Psi)(Y_{E_2Q_2})\|_{(E_2:q,P_2:p)}.
\end{align} Taking the supremum over $E_1,E_2$ and $X_{E_1Q_1}, Y_{E_2Q_2}$ completes the achievability proof.
\end{proof}

In a similar vein we also consider the reverse hypercontractive norms $\|\cdot\|^+_{\co,q\to p}$ for $q\geq p$. For $1\geq q\geq p$ these also find applications in the study of dissipation of open quantum systems, among others \cite{Cubitt.2015, Beigi.2020}. 

Analogously to the hypercontractive case above, the completely bounded reverse hypercontractive co-quasi-norm was shown to be, as expected, super-multiplicative under tensor products of CP maps in \cite[Theorem 1.6]{Kochanowski.2026}. While this super-multiplicativity is the operationally interesting inequality, here we prove the matching achievability. 

\begin{theorem}[Multiplicativity of complete reverse hypercontractivity, Theorem 1.6 of \cite{Kochanowski.2026} extended ]\label{thm:cb.reverse.hypercontractive.bound}
Let $q\geq p>0$ or $0>q\geq p$ be compatible indices and $\Phi:Q_1\to P_1,\,\Psi:Q_2\to P_2$ CP maps, then it holds that
\begin{align}\label{eqcocornmulti}
    \|\Phi\otimes\Psi\|^+_{\cb,\co,(Q_1Q_2:q)\to (P_1P_2:p)}= \|\Phi\|^+_{\cb,\co,(Q_1:q)\to (P_1:p)}\cdot \|\Psi\|^+_{\cb,\co, (Q_2:q)\to (P_2:p)}.
\end{align}
\end{theorem}

\begin{remark}
    The physically interesting `$\geq$' inequality of this statement is \cite[Theorem 1.6]{Kochanowski.2026}. The achievability, which was lacking there, will, as in the hypercontractive setting, be a consequence of \cref{lem:basic.4.index.multiplicativity}. For the operational details of this quantity we refer to \cite{Kochanowski.2026}.
\end{remark}

\begin{proof}[Proof of \Cref{thm:cb.reverse.hypercontractive.bound}]
Let $q\geq p$ be non-zero compatible indices, then
  \begin{align}
      \|&\id_E\otimes\Phi\otimes\Psi\|^+_{\co,(E:q,Q_1:q,Q_2:q)\to (E:q,P_1:p,P_2:p)}\\
      &\geq \|\id_E\otimes\Phi\otimes\id_{Q_2}\|^+_{\co,(E:q,Q_1:q,Q_2:q)\to (E:q,P_1:p,Q_2:q)}\\&\qquad\times \|\id_E\otimes\id_{P_1}\otimes\Psi\|^+_{\co,(E:q,P_1:p,Q_2:q)\to (E:q,P_1:p,P_2:p)} \\
      &\overset{(1)}{\geq} \|\id_E\otimes\id_{Q_2}\otimes\Phi\|^+_{\co,(E:q,Q_2:q,Q_1:q)\to (E:q,Q_2:q,P_1:p)}\cdot \|\Psi\|^+_{\cb,\co,(Q_2:q)\to (P_2:p)}
      \\ &\overset{(2)}{\geq} \|\Phi\|^+_{\cb,\co,(Q_1:q)\to (P_1:p)}\cdot \|\Psi\|^+_{\cb,\co,(Q_2:q)\to (P_2:p)},
  \end{align} 
  where $(1)$ follows from \Cref{thm:QuasiNormMinkowskiInequality} and \Cref{lem:cb.co.norm.via.t} for $\Psi$ on the combined system $(E:q,P_1:p)$, and in $(2)$ for $\Phi$ we used \Cref{lem:cb.co.norm.via.t} on the combined system $EQ_2$. Note that both of these work equally for positive or negative compatible indices.
  
The achievability follows again by means of \Cref{lem:basic.4.index.multiplicativity}. Let $X_{E_1Q_1}, Y_{E_2Q_2}\geq0$ be s.t. $\|X\|_{(E_1Q_1:q)}=\|Y\|_{(E_2Q_2:q)}=1$. Then
\begin{align}
    \|\Phi\otimes\Psi\|^+_{\cb,\co,(Q_1Q_2:q)\to (P_1P_2:p)} \!&\leq\! \|(\id_{E_1}\!\otimes\id_{E_2}\!\otimes\Phi\otimes\Psi)(X_{E_1Q_1}\otimes Y_{E_2Q_2})\|_{(E_1:q,E_2:q,P_1:p,P_2:p)} \\ &=\! \|(\id_{E_1}\!\otimes\Phi)(X_{E_1Q_1})\|_{(E_1:q,P_1:p)}\|(\id_{E_2}\!\otimes\Psi)(Y_{E_2Q_2})\|_{(E_2:q,P_2:p)}.
\end{align} Minimizing the right-hand side over $E_1,E_2$ and $X_{E_1Q_1}, Y_{E_2Q_2}$ yields the achievability claim.
\end{proof}

\section{Additivities and Chain Rules for Rényi Entropic Quantities}\label{sec:ApplictoQI}

Additivity questions for extremal output entropies have a long history. Additivity of the minimal output von Neumann entropy under tensor products of quantum channels was shown to be equivalent to additivity of the Holevo capacity, and hence to a single-letter formula for the classical capacity \cite{shor2004equivalence, Holevo.1998, SchumacherWestmoreland.1997}. Its Rényi analogue is the multiplicativity conjecture for maximal output $p$-norms \cite{AmosovHolevoWerner.2000}, and both were eventually disproved: first for large $p$ \cite{Werner.2002} and then for all $p>1$ \cite{Hayden.2008}, see also the explicit, non-random constructions \cite{Grudka.2010, Szczygielski.2024, Derksen.2026}, and for $p$ close to $0$ \cite{Cubitt.2008}, then for the von Neumann entropy \cite{Hastings.2009}; more recently for $p\in[0,\frac14)\cup(\frac34,1]$ \cite{Leung.2026} and, finally, simultaneously for all orders $p\in[0,\infty]$ \cite{Zhen.2026}. Additivity can be restored by allowing the input to be entangled with a reference system and conditioning on it. This is exactly what the completely bounded minimal output Rényi entropy is, and it was shown to be additive for all $p\geq1$ in \cite{Devetak.2006} and for all $p\geq\frac12$ in \cite{Kochanowski.2026, Li.2026}. In \Cref{sec:Additivity.Extremal.Output.Entropies} we establish analogous statements for conditional entropies in which the conditioning system $B$ is itself an output of the channel, possibly held by an adversary, and show that the maximal output conditional entropy is additive even without complete boundedness.

Chain rules are the second pillar, next to additivity, of how entropies behave under composition. For the von Neumann entropy the chain rule $H(A_1A_2|B)_\omega=H(A_1|B)_{\tr_{A_2}[\omega]}+H(A_2|A_1B)_\omega$ holds with equality, and it is what allows one to decompose the uncertainty of a many-round process into contributions of the individual rounds. For Rényi entropies no such identity is available. Chain rules then only hold as inequalities, typically at the price of changing the Rényi orders of the individual terms \cite{Dupuis.2015}, and analogous statements for the smooth min- and max-entropies hold up to correction terms \cite{Vitanov.2013}. Nevertheless, such chain rules are at the heart of modern security proofs in quantum cryptography. The entropy accumulation theorem \cite{Dupuis.2020} combines a Rényi chain rule with a continuity argument to bound the smooth min- or max-entropy produced by a sequential process whose rounds satisfy a Markov condition. This made device-independent protocols for key distribution \cite{Arnon.2018} and entanglement certification \cite{Arnon-Friedman.2019} amenable to finite-size analysis against general attacks. Its generalization \cite{Metger.2024} replaces the Markov condition by a non-signaling condition between rounds, which covers a wider range of protocols \cite{Metger.2023}. More recently, working with Rényi entropies directly, rather than passing through smooth entropies, has led to markedly better finite-size bounds \cite{ArqandHahnTan.2025, Arqand.2025, Himbeeck.2025, Fawzi.2026}. This has in turn raised the question which chain rules hold for which conditional Rényi entropy and under which structural assumptions \cite{Marwah.2025, Wooltorton.2026}. In particular, natural strengthenings in the device-independent setting are known to fail \cite[Section 4]{Wooltorton.2026}. The chain rules we derive below in \cref{sec:DI.ent.cert} concern the optimized conditional entropy $\tilde H^\uparrow_\alpha$ with $\alpha\in[\frac12,1)$, the regime relevant for upper bounds on conditional entropies; at the endpoint $\alpha=\frac12$ it coincides with the max-entropy, which links it to entanglement-related quantities. They hold for compositions of arbitrary, possibly signaling, channels without any Markov or non-signaling assumption. In this sense they complement the chain rules for the non-optimized conditional Rényi entropies underlying entropy accumulation \cite{Dupuis.2020, Metger.2024}, which require Markov or non-signaling conditions on the rounds. Operationally, conditional entropies of order $\alpha<1$ enter through \emph{upper} bounds: the max-entropy at $\alpha=\frac12$ governs the cost of information reconciliation \cite{Renner.2006}, its negativity certifies, and lower-bounds, distillable entanglement \cite{Arnon-Friedman.2019}, and sandwiched Rényi divergences of order $\alpha\in[\frac12,1)$ are strong converse exponents \cite{Li.2022}. Moreover, by the duality of conditional Rényi entropies, statements of order $\alpha<1$ are equivalent to statements of order $\beta=\frac{\alpha}{2\alpha-1}>1$ about a system purifying the process, see \Cref{rem:duality}.

Before we start discussing our results, let us recap the connection between sandwiched Rényi entropic quantities and two- and multi-indexed Schatten quasi-norms.
\subsection{Preliminaries: Entropic quantities via quasi-norms}\label{sec:Entopic.Norms}
Let $\cD(\cH)$ denote the set of quantum states on $\cH$, i.e.\ $\rho\geq0$ with $\Tr[\rho]=1$. For a state on labeled systems we write $\rho_{AB}$ and $\rho_A=\tr_B[\rho_{AB}]$. The Umegaki relative entropy is $D(\rho\|\sigma)=\Tr[\rho(\log\rho-\log\sigma)]$ if $\Pi_\rho\leq\Pi_\sigma$ and $+\infty$ otherwise. For $\alpha\in[\frac12,1)\cup(1,\infty]$ the \emph{sandwiched Rényi relative entropy} \cite{Lennert.2013, Wilde.2014} is
\begin{align}
    \tilde D_\alpha(\rho\|\sigma):=\frac{\alpha}{\alpha-1}\log\big\|\sigma^{\frac{1-\alpha}{2\alpha}}\rho\,\sigma^{\frac{1-\alpha}{2\alpha}}\big\|_\alpha
\end{align}
whenever $\Pi_\rho\leq\Pi_\sigma$, or $\alpha<1$ and $\Pi_\rho\Pi_\sigma\neq0$, and $+\infty$ otherwise. It induces the Rényi entropy 
\begin{align}
    H_\alpha(A)_\rho=-\tilde D_\alpha(\rho_A\|\1_A)=\frac{\alpha}{1-\alpha}\log\|\rho_A\|_\alpha
\end{align} and the \emph{optimized sandwiched Rényi conditional entropy}
\begin{align}
    \tilde H^\uparrow_\alpha(A|B)_\rho:=-\inf_{\sigma_B\in\cD(\cH_B)}\tilde D_\alpha(\rho_{AB}\|\1_A\otimes\sigma_B).
\end{align}
For $\alpha\to1$ these converge to the von Neumann (conditional) entropies \cite{book:Tomamichel.2016} and the endpoints $\alpha=\frac12$ and $\alpha=\infty$ are, respectively, the \emph{max-} and \emph{min-entropies} \cite{Koenig.2009}
\begin{align}
   \tilde H^\uparrow_{\frac{1}{2}}(A|B)_\rho = H_{\max}(A|B)_\rho, \\
   \tilde H^\uparrow_{\infty}(A|B)_\rho = H_{\min}(A|B)_\rho. 
\end{align}
This conditional entropy is both classically \cite{Arimoto.1977} and quantumly \cite{Beigi.2023, Kochanowski.2026} given as a two-indexed quasi-norm via
\begin{align}
    \tilde H_\alpha^\uparrow(A|B)_\rho &= \frac{\alpha}{1-\alpha}\log\|\rho_{BA}\|_{(B:1,A:\alpha)}, \\
    H(A|B)_\rho &= \lim_{\alpha\to 1}\frac{\alpha}{1-\alpha}\log\|\rho_{BA}\|_{(B:1,A:\alpha)} = -\frac{\mathrm{d}}{\mathrm{d}\alpha}\bigg|_{\alpha=1}\log\|\rho_{BA}\|^\alpha_{(B:1,A:\alpha)}.
\end{align}
For $\alpha\geq1$ this was effectively observed in \cite{Devetak.2006}, before the sandwiched relative entropy had been introduced. For $\alpha\in[\frac{1}{2},1)$ it is \cite[Corollary 4.1]{Kochanowski.2026}. Both follow from \Cref{thm:multi.quasi.relational.consistency} with $q=1$, $\cQ=(\alpha)$ and $a=b\geq0$, where the optimization over $a$ becomes the optimization over $\sigma_B$ in the definition of $\tilde H^\uparrow_\alpha$. This identification has been used to give simple proofs of basic properties of this conditional entropy \cite{Beigi.2023, Kochanowski.2026}, decoupling results \cite{Cheng.2024}, and chain rules \cite{Fawzi.2026}. By \Cref{cor:partial.trace.1} the Rényi entropy of a marginal is for any $\alpha\in[\frac{1}{2},1)\cup(1,\infty]$ 
\begin{align}
    H_\alpha(A)_{\tr_B[\rho_{AB}]}=\frac{\alpha}{1-\alpha}\log\|\rho_{AB}\|_{(A:\alpha,B:1)},
\end{align} which for pure states is the Rényi entanglement entropy \cite{Wang.2016}. Likewise, the optimized conditional Rényi entropy of the marginal $\tr_C[\rho_{ABC}]$ can be written, for any $\alpha\in[\frac{1}{2},1)\cup(1,\infty]$, as
\begin{align}
     \tilde H^\uparrow_\alpha(A|B)_{\tr_C[\rho_{ABC}]}=\frac{\alpha}{1-\alpha}\log\|\rho_{BAC}\|_{(B:1,A:\alpha, C:1)}.
\end{align}

These correspondences between quasi-norms and entropic quantities extend straightforwardly to minimal and maximal output entropic quantities. The following equalities are easy to check.

\begin{proposition}\label{prop:Extremal:Entropies.Mixed.Norms}
In the following let $\Phi:Q\to AB$ be a quantum channel, then it holds that 
\begin{itemize}
    \item Minimal output conditional Rényi-entropy: 
    \begin{align}
        \inf_{\rho_Q} \tilde{H}^\uparrow_\alpha(A|B)_{\Phi_{Q\to AB}(\rho)} = \frac{\alpha}{1-\alpha}\log\begin{cases}
            &\|\Phi\|^+_{1\to (B:1,A:\alpha)} \qquad \text{ for } \alpha\in(1,\infty], \\ 
            &\|\Phi\|^+_{\co,1\to (B:1,A:\alpha)} \quad \text{ for }  \alpha\in[\frac{1}{2},1).
        \end{cases}
    \end{align}
    \item Maximal output conditional Rényi-entropy:
    \begin{align}
        \sup_{\rho_Q} \tilde{H}^\uparrow_\alpha(A|B)_{\Phi_{Q\to AB}(\rho)} = \frac{\alpha}{1-\alpha}\log\begin{cases}
            &\|\Phi\|^+_{\co,1\to (B:1,A:\alpha)} \quad \text{ for } \alpha\in(1,\infty], \\
            &\|\Phi\|^+_{1\to (B:1,A:\alpha)} \qquad \text{ for } \alpha\in[\frac{1}{2},1).
        \end{cases}
    \end{align}
    \item Replacing these mixed (co)-quasi-norms with their completely bounded counterparts corresponds to an additional optimization over an environment $E$ and states $\rho_{EQ}$. That is
    \begin{align}
        \opt_{E,\rho_{EQ}}\tilde{H}_\alpha^\uparrow(A|BE)_{(\id_E\otimes\Phi)(\rho_{EQ})} = \frac{\alpha}{1-\alpha}\log\begin{cases}
            &\|\Phi\|^+_{\cb,1\to (B:1,A:\alpha)}, \\
            &\|\Phi\|^+_{\cb,\co, 1\to (B:1,A:\alpha)}.
        \end{cases} 
    \end{align} where we have the upper case when $\opt=\sup,\, \alpha\in[\frac{1}{2},1)$ and $\opt=\inf,\, \alpha\in(1,\infty]$ and the lower case when $\opt=\inf,\, \alpha\in[\frac{1}{2},1)$ and $\opt=\sup,\, \alpha\in(1,\infty]$.
\end{itemize}
\end{proposition}

\noindent Operationally, these quantities describe different scenarios. In the completely bounded minimal output conditional entropy $\inf_{E,\rho_{EQ}}\tilde H^\uparrow_\alpha(A|BE)_{(\id_E\otimes\Phi_{Q\to AB})(\rho_{EQ})}$ the input is prepared jointly with an arbitrary reference system $E$, which is then added to the conditioning system. It is thus the worst-case uncertainty about $A$ of an adversary who holds the output $B$ together with a quantum memory entangled with the input, as in the analysis of quantum key distribution with time-varying channels \cite{Fawzi.2026}. The maximal output conditional entropy $\sup_{\rho_Q}\tilde H^\uparrow_\alpha(A|B)_{\Phi_{A\to AB}(\rho_Q)}$, on the other hand, is the largest uncertainty about $A$ that a holder of $B$ can be left with. Since separable states have non-negative conditional entropy, a negative upper bound on it certifies that every output of the channel is entangled across $A$ and $B$ \cite{Arnon-Friedman.2019}. Additivity of either quantity means that, in these scenarios, using channels in parallel with arbitrarily entangled inputs gives no advantage over using them one at a time.

\noindent With this dictionary in mind we proceed to derive our extremal entropy additivities and chain rules. 

\subsection{Additivities of Extremal Output Conditional Entropies}\label{sec:Additivity.Extremal.Output.Entropies}

In the following we prove additivity of the completely bounded minimal and maximal output sandwiched Rényi conditional entropies under tensor products of quantum channels for all $\alpha\in[\frac{1}{2},1)\cup(1,\infty]$. In the mixed-norm picture, each of these statements becomes a multiplicativity statement of a $(\cb),1\to(1,p)$ mixed norm or a $(\cb),\co,1\to(1,p)$ mixed co-quasi-norm, which we will prove as short applications of the tools of the previous \Cref{sec:LinerMaps}. The unconditional output entropies, for which the results are known \cite{Devetak.2006, Kochanowski.2026, Li.2026},
are contained as the special case of trivial conditioning systems.

Let us start with the completely bounded minimal output case.

\subsubsection{Completely Bounded Minimal Output Conditional Entropies are Additive}

\begin{theorem}[Additivity of cb min conditional entropy]\label{thm:cb.min.cond.entropy.add}
   Let $p\in[\frac{1}{2},1)\cup(1,\infty]$ then for two quantum channels $\Phi:Q_1\to A_1B_1$, $\Psi:Q_2\to A_2B_2$ it holds that
\begin{align}
\inf_{E,\rho_{EQ_1Q_2}}\tilde{H}^\uparrow_p(B_1B_2|A_1A_2E)_{(\id_E\otimes\Phi\otimes\Psi)(\rho_{EQ_1Q_2})} &= \inf_{E,\rho_{EQ_1}}\tilde{H}^\uparrow_p(B_1|A_1E)_{(\id_E\otimes\Phi)(\rho_{EQ_1})}\\&\qquad+\inf_{E,\rho_{EQ_2}}\tilde{H}^\uparrow_p(B_2|A_2E)_{(\id_E\otimes\Psi)(\rho_{EQ_2})}.
\end{align} 
\end{theorem}

\begin{remark}
    For $p\in[1,\infty]$ this result was published in \cite{Fawzi.2026} and finds applications in security proofs for time-varying quantum key distribution protocols. Hence the above extends the additivity to the whole range of admissible $p$.
\end{remark}

It hence suffices to prove the case $p\in[\frac{1}{2},1)$. Via \Cref{prop:Extremal:Entropies.Mixed.Norms} its proof will be a direct consequence of the following multiplicativity results.

\begin{lemma}\label{cb.co.1.to.1.p.additivity}
Let $p\in[\frac{1}{2},1]$ and $\Phi:Q_1\to A_1B_1,\,\Psi:Q_2\to A_2B_2$ be CP maps, then
    \begin{align}
        \|\Phi\otimes\Psi\|^+_{\cb,\co,(Q_1Q_2:1)\to (A_1A_2:1,B_1B_2:p)} \!=\! \|\Phi\|^+_{\cb,\co,(Q_1:1)\to (A_1:1,B_1:p)} \|\Psi\|^+_{\cb,\co,(Q_2:1)\to (A_2:1,B_2:p)}
    \end{align}
\end{lemma}

\begin{proof}[Proof of \cref{cb.co.1.to.1.p.additivity}]
    For the non-trivial inequality see that
    \begin{align}
       \|\Phi\otimes\Psi\|^+&_{\cb,\co,(Q_1:1,Q_2:1)\to (A_1:1,A_2:1,B_1:p,B_2:p)}  \\&\geq \|\id_{Q_1}\otimes\Psi\|^+_{\cb,\co,(Q_1:1,Q_2:1)\to (Q_1:1,A_2:1,B_2:p)} \\&\qquad\times \|\SWAP_{Q_1\leftrightarrow A_2}\otimes\id_{B_2}\|^+_{\cb,\co,(Q_1:1,A_2:1,B_2:p)\to (A_2:1,Q_1:1,B_2:p)} \\ &\qquad \times \|\id_{A_2}\otimes\Phi\otimes\id_{B_2}\|^+_{\cb,\co,(A_2:1,Q_1:1,B_2:p)\to (A_2:1,A_1:1,B_1:p,B_2:p)}\\&\qquad\times\|\SWAP_{A_2\leftrightarrow A_1}\otimes\id_{B_1B_2}\|^+_{\cb,\co, (A_2:1,A_1:1,B_1:p,B_2:p)\to (A_1:1,A_2:1,B_1:p,B_2:p)} \\
       &\overset{(1)}{=} \|\id_{Q_1}\otimes\Psi\|^+_{\cb,\co,(Q_1:1,Q_2:1)\to (Q_1:1,A_2:1,B_2:p)} \\&\qquad\times\|\id_{A_2}\otimes\Phi\otimes\id_{B_2}\|^+_{\cb,\co,(A_2:1,Q_1:1,B_2:p)\to (A_2:1,A_1:1,B_1:p,B_2:p)} \\
       &\overset{(2)}{\geq} \|\Psi\|^+_{\cb,\co,(Q_2:1)\to (A_2:1,B_2:p)} \|\Phi\|^+_{\cb,\co,(Q_1:1)\to (A_1:1,B_1:p)}
    \end{align}
where in $(1)$ we used that swapping adjacent equal indices does not affect the underlying quasi-norms, \Cref{cor:equal.adjacient.indices}, and in $(2)$ we used \Cref{lem:cb.co.norm.via.t} and \Cref{cor:cb.id.right.quasi}.

For the achievability take $X_{E_1Q_1}, Y_{E_2Q_2}\geq0$ with $\|X\|_{(E_1:1,Q_1:1)}=\|Y\|_{(E_2:1,Q_2:1)}=1$ and see that
\begin{align}
    \|\Phi\otimes\Psi\|^+&_{\cb,\co,(Q_1:1,Q_2:1)\to (A_1:1,A_2:1,B_1:p,B_2:p)} \\&\leq \|(\id_{E_1E_2}\otimes\Phi\otimes\Psi)(X_{E_1Q_1}\otimes Y_{E_2Q_2})\|_{(E_1:1,E_2:1,A_1:1,A_2:1,B_1:p,B_2:p)} \\ &\leq \|(\id_{E_1}\otimes\Phi)(X)\otimes(\id_{E_2}\otimes\Psi)(Y)\|_{(E_1:1,A_1:1,B_1:p,E_2:1,A_2:1,B_2:p)} \\ &= \|(\id_{E_1}\otimes\Phi)(X)\|_{(E_1:1,A_1:1,B_1:p)}\|(\id_{E_2}\otimes\Psi)(Y)\|_{(E_2:1,A_2:1,B_2:p)},
\end{align} where the second inequality follows from the non-commutative Minkowski inequality \Cref{thm:QuasiNormMinkowskiInequality} with $p\in[\frac{1}{2},1]$ and the last equality from \Cref{cor.multi.tensor.multiplicativity}. Minimizing over such $X,Y$ proves the achievability.
\end{proof}

\begin{proof}[Proof of \cref{thm:cb.min.cond.entropy.add}]
This follows from \Cref{prop:Extremal:Entropies.Mixed.Norms} combined with \Cref{cb.co.1.to.1.p.additivity} for $p\in[\frac{1}{2},1)$ and \cite[Theorem 4.11]{Fawzi.2026} for $p\in(1,\infty]$. Technically the latter gives the non-trivial inequality, while the achievability follows from product inputs and \Cref{lem:basic.4.index.multiplicativity} or \cite[Corollary 5.9]{book:Tomamichel.2016} as in the achievability part of the proof of \Cref{cb.co.1.to.1.p.additivity}.
\end{proof}

\subsubsection{Maximal Output Conditional Entropies are Additive}

Considering additivity of the maximal output sandwiched Rényi conditional entropies we have the following result.

\begin{theorem}[Additivity of maximal output conditional entropies] \label{thm:max.cond.entropy.add}
Let $p\in[\frac{1}{2},1)\cup(1,\infty]$ then for two quantum channels $\Phi:Q_1\to A_1B_1$, $\Psi:Q_2\to A_2B_2$ it holds that
\begin{align}
    \sup_{\rho_{Q_1Q_2}}\tilde{H}^\uparrow_p(B_1B_2|A_1A_2)_{(\Phi\otimes\Psi)(\rho_{Q_1Q_2})} = \sup_{\rho_{Q_1}}\tilde{H}^\uparrow_p(B_1|A_1)_{\Phi(\rho_{Q_1})}+\sup_{\rho_{Q_2}}\tilde H^\uparrow_p(B_2|A_2)_{\Psi(\rho_{Q_2})}.
\end{align}
\end{theorem}

\begin{remark}
Notice that here, in contrast to the minimal output case, additivity holds without complete boundedness. For the non-completely bounded Rényi entropies this was already known \cite{Devetak.2006, Li.2026, Kochanowski.2026}. To the best of our knowledge this additivity of maximal output optimized conditional Rényi-entropies has not appeared explicitly in the quantum information literature before.
\end{remark}

Via \Cref{prop:Extremal:Entropies.Mixed.Norms} its proof will be a consequence of the following two multiplicativity results. 

\begin{lemma}\label{thm:p.leq.1.cond.multiplicativity}
Let $p\in[\frac{1}{2},1]$ and $\Phi:Q_1\to A_1B_1$ and $\Psi:Q_2\to A_2B_2$ be CP maps, then
    \begin{align}
        \|\Phi\otimes\Psi\|^+_{\cb,(Q_1:1,Q_2:1)\to (A_1:1,A_2:1,B_1:p,B_2:p)} = \|\Phi\|^+_{(Q_1:1)\to (A_1:1,B_1:p)} \|\Psi\|^+_{(Q_2:1)\to (A_2:1,B_2:p)}.
    \end{align}
\end{lemma}

\begin{remark}
   Note that the result implicitly contains the collapse for CP maps
   \begin{align}
       \|\Phi\|^{(+)}_{\cb,(Q:1)\to (A:1,B:p)}  = \|\Phi\|^{(+)}_{(Q:1)\to (A:1,B:p)}.
   \end{align}
   \noindent Also by \cref{cor:mixed.positivity.normed} we have that this result is equally true without the restriction of the input states to be positive, i.e. \begin{align}
        \|\Phi\otimes\Psi\|_{\cb,(Q_1:1,Q_2:1)\to (A_1:1,A_2:1,B_1:p,B_2:p)} = \|\Phi\|_{(Q_1:1)\to (A_1:1,B_1:p)} \|\Psi\|_{(Q_2:1)\to (A_2:1,B_2:p)}
    \end{align} holds for CP maps and is equal to the above.
\end{remark}

\begin{proof}[Proof of \cref{thm:p.leq.1.cond.multiplicativity}]
    The non-trivial inequality follows rather directly using the complete contractiveness of the SWAP \Cref{thm:QuasiNormMinkowskiInequality} with $p\leq 1$ in 
    \begin{align}
       \|\Phi\otimes\Psi\|^+&_{\cb,(Q_1:1,Q_2:1)\to (A_1:1,A_2:1,B_1:p,B_2:p)} \\&\leq \|\Phi\otimes\Psi\|^+_{\cb,(Q_1:1,Q_2:1)\to (A_1:1,B_1:p,A_2:1,B_2:p)} \\ &\leq \|\Phi\otimes\id_{Q_2}\|^+_{\cb,(Q_1:1,Q_2:1)\to (A_1:1,B_1:p,Q_2:1)}\\&\quad\times\|\id_{A_1B_1}\otimes \Psi\|^+_{\cb,(A_1:1,B_1:p,Q_2:1)\to (A_1:1,B_1:p,A_2:1,B_2:p)} \\
       &=\|\Phi\|^+_{\cb,(Q_1:1)\to (A_1:1,B_1:p)}\|\Psi\|^+_{\cb,(Q_2:1)\to (A_2:1,B_2:p)}
    \end{align} where in the last line we implicitly used \Cref{cor:partial.trace.1} since $\Phi$ is CP and \Cref{lem:cb.norm.via.t}. Explicitly this is
    \begin{align}
        \|\Phi\otimes\id_{Q_2}\|^+_{\cb,(Q_1:1,Q_2:1)\to (A_1:1,B_1:p,Q_2:1)} &= \sup_{E,X_{EQ_1Q_2}\geq0}\frac{\|\tr_{Q_2}[(\id_E\otimes\Phi\otimes\id_{Q_2})(X)]\|_{(E:1,A_1:1,B_1:p)}}{\|\tr_{Q_2}[X_{EQ_1Q_2}]\|_{(EQ_1:1)}} \\&= \sup_{E,X_{EQ_1Q_2}\geq0}\frac{\|(\id_E\otimes\Phi)(\tr_{Q_2}[X_{EQ_1Q_2}])\|_{(E:1,A_1:1,B_1:p)}}{\|\tr_{Q_2}[X_{EQ_1Q_2}]\|_{(EQ_1:1)}} \\&=\sup_{E,X_{EQ_1}\geq0}\frac{\|(\id_E\otimes\Phi)(X_{EQ_1})\|_{(E:1,A_1:1,B_1:p)}}{\|X_{EQ_1}\|_{(EQ_1:1)}} \\&=\|\Phi\|^+_{\cb,(Q_1:1)\to (A_1:1,B_1:p)}.
    \end{align}
    Further we have for $\Phi,\Psi$,
    \begin{align}
        \|\Phi\|^{(+)}_{\cb,(Q_1:1)\to (A_1:1,B_1:p)} &= \sup_E\|\id_E\otimes\Phi\|^{(+)}_{(E:1,Q_1:1)\to (E:1,A_1:1,B_1:p)}  \\ &\leq \sup_E\|\Phi\otimes\id_E\|^{(+)}_{(Q_1:1,E:1)\to (A_1:1,B_1:p,E:1)} \\ &= \|\Phi\|^{(+)}_{(Q_1:1)\to (A_1:1,B_1:p)}
    \end{align} by \Cref{thm:QuasiNormMinkowskiInequality}, since $p\in[\frac{1}{2},1]$, and \Cref{cor:partial.trace.1}.

    The achievability follows equally via product state inputs $X_{Q_1},Y_{Q_2}\geq0$ with $\|X\|_1=\|Y\|_1=1$.
    \begin{align}
        \|\Phi\otimes\Psi\|^+_{\cb,(Q_1:1,Q_2:1)\to (A_1:1,A_2:1,B_1:p,B_2:p)} &\geq \|\Phi\otimes\Psi\|^+_{(Q_1:1,Q_2:1)\to (A_1:1,A_2:1,B_1:p,B_2:p)} \\ &\geq \|\Phi(X)\otimes\Psi(Y)\|_{(A_1:1,A_2:1,B_1:p,B_2:p)} \\ &= \|\Phi(X)\|_{(A_1:1,B_1:p)}\|\Psi(Y)\|_{(A_2:1,B_2:p)},
    \end{align} where the last equality follows from  \Cref{lem:basic.4.index.multiplicativity}. Minimizing over $X,Y$ finishes the claim.
\end{proof}

The other multiplicativity statement we require for \cref{thm:max.cond.entropy.add} is the following.
\begin{lemma}\label{co.1.to.1.p.multiplicativity}
Let $p\in[1,\infty]$ and $\Phi:Q_1\to A_1B_1,\,\Psi:Q_2\to A_2B_2$ be CP maps, then
    \begin{align}
        \|\Phi\otimes\Psi\|^+_{\cb,\co,(Q_1:1,Q_2:1)\to (A_1:1,A_2:1,B_1:p,B_2:p)} \!=\! \|\Phi\|^+_{\co,(Q_1:1)\to (A_1:1,B_1:p)} \|\Psi\|^+_{\co,(Q_2:1)\to (A_2:1,B_2:p)}.
    \end{align}
\end{lemma}

\begin{remark}
    Just as above we have here a collapse of the completely bounded quantity for CP maps
    \begin{align}
        \|\Phi\|^+_{\cb,\co,(Q:1)\to (A:1,B:p)}  = \|\Phi\|^+_{\co,(Q:1)\to (A:1,B:p)} 
    \end{align} implicit in this statement.
\end{remark}

\begin{proof}
    For the non-trivial bound observe that
    \begin{align}
        \|\Phi\otimes\Psi\|^+&_{\cb,\co,(Q_1:1,Q_2:1)\to (A_1:1,A_2:1,B_1:p,B_2:p)} \\&\geq \|\Phi\otimes\id_{Q_2}\|^+_{\cb,\co,(Q_1:1,Q_2:1)\to(A_1:1,B_1:p,Q_2:1)}  \\&\qquad \times\|\id_{A_1B_1}\otimes\Psi\|^+_{\cb,\co,(A_1:1,B_1:p,Q_2:1)\to (A_1:1,B_1:p,A_2:1,B_2:p)} \\&\qquad\times\|\id_{A_1}\otimes\SWAP\otimes\id_{B_2}\|^+_{\cb,\co,(A_1:1,B_1:p,A_2:1,B_2:p)\to (A_1:1,A_2:1,B_1:p,B_2:p)} \\ 
        &\geq \|\Phi\|^+_{\cb,\co,(Q_1:1)\to(A_1:1,B_1:p)}  \|\Psi\|^+_{\cb,\co,(Q_2:1)\to (A_2:1,B_2:p)} 
    \end{align} where last inequality follows from the complete expansivity of the SWAP due to \Cref{thm:QuasiNormMinkowskiInequality}, \Cref{lem:cb.co.norm.via.t}, and $p\geq 1$ and implicitly \Cref{cor:partial.trace.1} with complete positivity of $\Phi$ as
    \begin{align}
       \|\Phi\otimes\id_{Q_2}\|^+_{\cb,\co,(Q_1:1,Q_2:1)\to(A_1:1,B_1:p,Q_2:1)}  &= \inf_{E,\underset{\|X\|_1=1}{X_{EQ_1Q_2}\geq0}}\|(\id_E\otimes\Phi\otimes\id_{Q_2})(X_{EQ_1Q_2})\|_{(EA_1:1,B_1:p,Q_2:1)} \\&=\inf_{E,\underset{\|X\|_1=1}{X_{EQ_1Q_2}\geq0}}\|(\id_E\otimes\Phi)(\tr_{Q_2}[X_{EQ_1Q_2}])\|_{(EA_1:1,B_1:p)} \\&=\inf_{E,\underset{\|X\|_1=1}{X_{EQ_1}\geq0}}\|(\id_E\otimes\Phi)(X_{EQ_1})\|_{(EA_1:1,B_1:p)} \\&= \|\Phi\|^+_{\cb,\co,(Q_1:1)\to(A_1:1,B_1:p)}.
    \end{align}

    Further, similarly to the proof of the previous lemma, we have via \Cref{thm:QuasiNormMinkowskiInequality} and \Cref{cor:partial.trace.1} that the cb quantity is equal to the non-cb quantity: 
    \begin{align}
        \|\Phi\|^+_{\cb,\co,(Q_1:1)\to(A_1:1,B_1:p)}  &= \inf_{E} \|\id_E\otimes\Phi\|^+_{\co,(E:1,Q_1:1)\to(E:1,A_1:1,B_1:p)} \\& \geq \inf_{E} \|\Phi\otimes\id_E\|^+_{\co,(Q_1:1,E:1)\to(A_1:1,B_1:p,E:1)} \\ &\geq \|\Phi\|^+_{\co,(Q_1:1)\to(A_1:1,B_1:p)}.
    \end{align}
    
    The achievability follows via product state inputs $X_{Q_1}, Y_{Q_2}$ and \Cref{lem:basic.4.index.multiplicativity}.
\end{proof}

\begin{proof}[Proof of \cref{thm:max.cond.entropy.add}]
    This is a direct combination of \Cref{prop:Extremal:Entropies.Mixed.Norms} with \Cref{thm:p.leq.1.cond.multiplicativity} for $p\in[\frac{1}{2},1]$ and \Cref{co.1.to.1.p.multiplicativity} for $p\in[1,\infty]$. 
\end{proof}

\begin{remark}[Duality]\label{rem:duality}
Let $V:Q\to ABC$ be a Stinespring isometry of $\Phi:Q\to AB$ and $\Phi^\prime:=\tr_B[V\cdot V^*]:Q\to AC$. Since conditioning on a purifying system can only decrease $\tilde H^\uparrow_\alpha$, the infimum defining the completely bounded minimal output conditional entropy is attained on pure inputs $\psi_{EQ}$, and the duality relation $\tilde H^\uparrow_\alpha(A|BE)=-\tilde H^\uparrow_\beta(A|C)$ for pure states with $\frac1\alpha+\frac1\beta=2$ \cite{Lennert.2013, Beigi.2013}, yields
\begin{align}
    \inf_{E,\rho_{EQ}}\tilde H^\uparrow_\alpha(A|BE)_{(\id_E\otimes\Phi_{Q\to AB})(\rho)}=-\sup_{\rho_Q}\tilde H^\uparrow_\beta(A|C)_{\Phi_{Q\to AB}^\prime(\rho_Q)}.
\end{align}
As $(\Phi_1\otimes\Phi_2)^\prime=\Phi_1^\prime\otimes\Phi_2^\prime$, \Cref{thm:cb.min.cond.entropy.add} for $\alpha\in[\frac12,1)$ and \Cref{thm:max.cond.entropy.add} for $\beta\in(1,\infty]$ are equivalent, and likewise \Cref{thm:max.cond.entropy.add} for $\alpha\in[\frac12,1)$ is equivalent to \cite[Theorem 4.11]{Fawzi.2026}. Operationally, the completely bounded minimal output conditional entropy of $\Phi$ thus measures how much an adversary holding the environment $C$ can learn about $A$ in the worst case.
\end{remark}

\subsection{Additivity of Completely Bounded Rényi Coherent Information}\label{sec:Renyi.coherent.info}

Another related entropic quantity which is induced by the above conditional entropy is the $\alpha$-Rényi coherent information of a quantum channel \cite{Gao.2018.2, Beigi.2023, Berta.2026.1, Li.2026, Shor.2002, Leditzky.2016}. For any $\alpha\in[\frac{1}{2},1)\cup(1,\infty]$ the $\alpha$-coherent information of some bipartite quantum state $\rho_{AB}$ is given by
\begin{align}
    I_\alpha(A\rangle B)_\rho:=-\tilde H^\uparrow_\alpha(A|B)_\rho
\end{align}
and the $\alpha$-\emph{Rényi coherent information} of the quantum channel $\Phi:Q\to P$ is the maximal output coherent information of the channel $\Phi_{Q\to P}\otimes\id_R$, where $R\simeq Q$ are isomorphic. That is
\begin{align}
    I^{\text{coh}}_\alpha(\Phi):=\max_{\rho_{QR}}I_\alpha(R\rangle P)_{(\Phi\otimes\id_R)(\rho_{QR})}=\max_{\rho_{QR}}-\tilde H^\uparrow_\alpha(R| P)_{(\Phi\otimes\id_R)(\rho_{QR})},
\end{align} where the maximization is over quantum states $\rho_{QR}\in\cD(\cH_{QR})$.
Consequently, the following is easy to show.
\begin{proposition}
Let $\alpha\in[\frac{1}{2},1)\cup(1,\infty]$ and $\Phi:Q\to P$ be a quantum channel, then
\begin{align}
    I^{\text{coh}}_\alpha(\Phi) = \frac{\alpha}{\alpha-1}\log\begin{cases}
        &\|\Phi\otimes\id_R\|_{(QR:1)\to (P:1,R:\alpha)} \qquad \text{ for } \alpha\in (1,\infty], \\
        &\|\Phi\otimes\id_R\|^+_{\co,(QR:1)\to (P:1,R:\alpha)} \quad \text{ for } \alpha\in[\frac{1}{2},1). 
    \end{cases}
\end{align}
\end{proposition}

\begin{remark}
    For $\alpha>1$ this was shown in \cite[Proposition 5.4]{Beigi.2023}, while for $\alpha<1$ it is new.
\end{remark}

Unlike the above quantity, which is not expected to be additive, since already the coherent information ($\alpha\to1$) is not \cite{DiVincenzo.1998, SmithYard.2008}, we show here that its completely bounded version is additive for any $\alpha\in[\frac{1}{2},1)\cup(1,\infty]$.

The completely bounded version is, naturally, given by
\begin{align}
    I_\alpha^{\cb,\text{coh}}(\Phi)&:=\max_{E,\rho_{EQR}}-\tilde H_\alpha^\uparrow(R|EP)_{(\id_{E}\otimes\Phi_{Q\to P}\otimes\id_R)(\rho_{EQR})} \\&= \frac{\alpha}{\alpha-1}\log\begin{cases}
        \|\Phi\otimes\id_R\|_{\cb,(Q:1,R:1)\to(P:1,R:\alpha)} \qquad \text{ for } \alpha\in(1,\infty], \\
        \|\Phi\otimes\id_R\|^+_{\cb,\co,(Q:1,R:1)\to(P:1,R:\alpha)} \quad \text{ for } \alpha\in[\frac{1}{2},1).
    \end{cases}
\end{align}

\begin{theorem}[Additivity of the completely bounded Rényi coherent information]\label{thm:cb.coherent.add}
Let $\alpha\in[\frac{1}{2},1)\cup (1,\infty]$ then for any two quantum channels $\Phi:Q_1\to P_1,\,\Psi:Q_2\to P_2$ it holds that
\begin{align}
   I_\alpha^{\cb,\textup{coh}}(\Phi\otimes\Psi) = I_\alpha^{\cb,\textup{coh}}(\Phi) + I_\alpha^{\cb,\text{coh}}(\Psi). 
\end{align}
The same holds for $I^{\cb,\textup{coh}}(\Phi):=\lim_{\alpha\to 1}I_\alpha^{\cb,\textup{coh}}(\Phi)$.
\end{theorem}

\begin{proof}
    This follows directly from the established additivity of the completely bounded minimal output conditional sandwiched Rényi entropy \Cref{thm:cb.min.cond.entropy.add} (or in the mixed-norm picture the underlying \Cref{cb.co.1.to.1.p.additivity} and \cite[Theorem 4.11]{Fawzi.2026}) applied to the channels $\Phi_{Q_1\to P_1}\otimes\id_{R_1}: Q_1R_1\to P_1R_1$ and $\Psi_{Q_2\to P_2}\otimes\id_{R_2}:Q_2R_2\to P_2R_2$, where $Q_1\simeq R_1, Q_2\simeq R_2$.
\end{proof}

For completeness we record that the \emph{channel Rényi mutual information} is a similar quantity which has received some attention in the literature. It is a mutual information, rather than conditional entropy, version of the coherent information of \Cref{sec:Renyi.coherent.info}. It is also effectively a completely bounded quantity and its additivity was established for $\alpha\in(1,\infty]$ in \cite{Devetak.2006, Gupta.2015} and for $\alpha\in[\frac{1}{2},1)$ in \cite{Li.2026}.

Finally we go beyond the tensor-product setting.

\subsection{A Signaling Chain Rule for DI Entanglement Certification}\label{sec:DI.ent.cert}

In this subsection we derive a chain rule for $\alpha<1$ conditional sandwiched Rényi entropies and discuss its operational relevance in DI entanglement certification. We leave it, however, to future work \cite{Hahn.in.prep} to take care of the practical details and derive actual numerical improvements in terms of certification rates.

In certain quantum settings the relevant protocols are well-modeled using product channels. This is the case, e.g., in device-dependent quantum key distribution protocols \cite{Fawzi.2026, Wooltorton.2026}, device-dependent entanglement certification protocols \cite{Friis.2019, Pallister.2018}, or dissipation of open quantum systems under independent Markovian noise \cite{Beigi.2016, Bardet.2022, Cubitt.2015, Beigi.2020}.
However, this is not the case in general and in particular in device-independent (DI) settings, where tensor-product assumptions are too strong, see e.g. \cite{Wooltorton.2026, Metger.2024, Dupuis.2020, Arnon.2018}
for QKD and \cite{Arnon-Friedman.2019} for entanglement certification.

\subsubsection{Device Independent Entanglement Distillation Certification}

Since entanglement is a key property differentiating quantum from classical protocols \cite{Horodecki.2009, JozsaLinden.2003, RaussendorfBriegel.2001, Vidal.2003, BennettWiesner.1992, Bennett.1993, Ekert.1991, Bennett.2002} it is a relevant question how to certify whether a given device or protocol actually produces entangled output states. The most general setting for this is the device-independent setting \cite{Arnon-Friedman.2019} in which both the source and the measurement devices are treated as black boxes. Device-independent entanglement certification studied in \cite{Arnon-Friedman.2019} bounds the smooth max-entropy \cite{Koenig.2009}.
While the smooth max-entropy has a similar operational interpretation to Rényi entropies for $\alpha<1$, Rényi entropies have led to tighter finite-size bounds than smooth entropies in related tasks, e.g.\ in privacy amplification and decoupling with $\alpha>1$ \cite{Dupuis.2022}, and one may hope for a similar improvement for bounds on the distillable entanglement. Hence here we consider the $\alpha<1$ conditional optimized sandwiched entropies and more specifically, we consider the maximal output conditional entropy as a measure of how much entanglement a given quantum channel can produce. That is because a negative conditional entropy certifies entanglement, as separable states have non-negative conditional entropy, and hence the maximal output conditional entropy is a worst-case measure of entangled outputs and can be used to bound the global distillable entanglement in the worst case \cite{Arnon-Friedman.2019, book:KhatriWilde.2024}. We will bound the maximal output conditional entropy after $n$ rounds of a certification protocol in terms of the individual rounds.

The protocol and channel structure of device-independent entanglement certification of \cite{Arnon-Friedman.2019}, which is also the one we consider here, is the following. We have a \emph{source} which sequentially outputs possibly entangled quantum states into the system $Q$, acting as a quantum register. During an honest implementation $Q$ should be thought of as consisting internally of $n$ systems $A_i, B_i$ such that the source emits pairs of (possibly) entangled systems $A_iB_i$ at time steps $i\in[n]$. The certification procedure proceeds by deciding at any given time step $j$ whether to `store' systems $A_jB_j$ or whether to extract them and perform some ``test'' measurement on them yielding classical information. We model both by the application of a quantum channel $\Phi_j$ and ignore the exact details of these measurements as they are not relevant in our consideration here. The remaining, non-extracted and potentially altered, quantum system is stored in the
register $Q_j$. Moreover, similarly to \cite{Arnon-Friedman.2019}, we allow $Q_j$ to additionally hold a copy of all classical information produced during previous measurement rounds, i.e., it typically contains the choice of measurement settings and the
corresponding classical outputs for all previous rounds. This allows us to additionally consider channels $\Phi_i$ that are adaptable, in the sense that they may act differently, depending on the classical measurement data that is stored in $Q_{i-1}$. For a visual representation of this see \cref{fig:sequential.rounds}.

\begin{figure}[h]
\centering
\begin{tikzpicture}[
    >={Stealth[length=2mm]},
    box/.style  = {draw, rounded corners=1.5pt, fill=black!5,
                   minimum width=9mm, minimum height=9mm,
                   inner sep=2pt, font=\normalsize},
    wire/.style = {thick},
    lab/.style  = {font=\footnotesize, inner sep=1.5pt},
    note/.style = {font=\footnotesize\itshape, inner sep=2pt},
]

\node[box] (P1) at (1.60,0) {$\Phi_1$};
\node[box] (P2) at (3.60,0) {$\Phi_2$};
\node      (dots) at (5.30,0) {$\cdots$};
\node[box] (Pn) at (7.00,0) {$\Phi_n$};

\draw[wire,->] (0.30,0) -- (P1.west) node[lab, midway, above] {$Q$};
\draw[wire,->] (P1.east) -- (P2.west) node[lab, midway, above] {$Q_1$};
\draw[wire,->] (P2.east) -- (4.95,0)  node[lab, midway, above] {$Q_2$};
\draw[wire,->] (5.65,0)  -- (Pn.west) node[lab, midway, above] {$Q_{n-1}$};
\draw[wire,->] (Pn.east) -- (8.30,0)  node[lab, midway, above] {$Q_n$};

\foreach \n/\i in {P1/1, P2/2, Pn/n} {
  \draw[wire,->] (\n.north) -- ++(0,1.05) node[lab, above] {$A_{\i}$};
  \draw[wire,->] (\n.south) -- ++(0,-1.05) node[lab, below] {$B_{\i}$};
}
\node[lab] at (5.30, 1.15) {$\cdots$};
\node[lab] at (5.30,-1.15) {$\cdots$};

\draw[decorate, decoration={brace, amplitude=4pt}]
      (0.95, 1.9) -- (7.65, 1.9)
      node[lab, midway, above=6pt] {$A_1\cdots A_n:\,\alpha$};
\draw[decorate, decoration={brace, amplitude=4pt, mirror}]
      (0.95,-1.9) -- (7.65,-1.9)
      node[lab, midway, below=6pt] {$B_1\cdots B_n:\,1$};

\end{tikzpicture}
\caption{The device-independent entanglement certification structure of \cite{Arnon-Friedman.2019}. A source feeds the quantum memory $Q$; in round $i$ the channel $\Phi_i:Q_{i-1}\to A_iB_iQ_i$ either stores the pair $A_iB_i$ or measures it, and passes the remaining register $Q_i$ --- which carries all classical measurement outcomes so far --- on to the next round. Both the choice of $\Phi_i$ and the state emitted by the source at step $i$ may therefore depend on every earlier outcome. The certified quantity is the maximal output conditional entropy $\tilde H^\uparrow_\alpha(A_1\cdots A_n\mid B_1\cdots B_n)$, whose negativity witnesses entanglement; the braces indicate the corresponding index assignment $(B_1\cdots B_n:1,\;A_1\cdots A_n:\alpha)$. Iterating \Cref{thm:signaling.sup.chain.rule} bounds it by the sum of the $n$ single-round quantities, for every $\alpha\in[\tfrac12,1)$.}
\label{fig:sequential.rounds}
\end{figure}
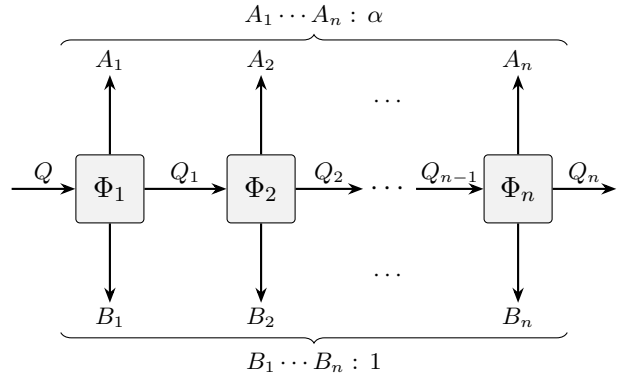

The chain rules below concern the following elementary step of such a sequential protocol. We consider two quantum channels $\Phi:Q\to A_1B_1Q_1, \Psi:Q_1\to A_2B_2$, where the input of the second channel is part of the output of the first. The global map is then
\begin{align}\label{equ:signaling.channel.structure}
    (\id_{A_1B_1}\otimes\Psi)\circ \Phi.
\end{align}

\begin{remark}
Specific choices in this rather generic setup let us recover the tensor-product setting as well as several non-signaling settings \cite{Beckman.2001, Chiribella.2009, Gutoski.2007}; see \Cref{fig:composed.channels} for a graphic representation.

The tensor-product setting is recovered when setting $Q=Q_0Q_1$ and $\Phi$ is replaced by $\Phi^\prime_{Q_0\to A_1B_1}\otimes\id_{Q_1}$.  Conversely, the channel structure of the non-signaling setting in \cite[Section 3.3, Lemma 3.6]{Metger.2024} is recovered using that non-signaling maps are semi-localizable \cite{Beckman.2001}, i.e.\ of the form $\mathcal{M}=(\mathcal{N}_{MR\to A^\prime R^\prime}\otimes\id_{E^\prime})\circ(\mathcal{V}_{E\to E^\prime M}\otimes\id_R)$ with a memory system $M$: setting $Q\to ER$, $\Phi\to\mathcal{V}_{E\to E^\prime M}\otimes\id_R$ with $B_1\to E^\prime$, $A_1$ trivial and $Q_1\to MR$, and $\Psi\to\tr_{R^\prime}\circ\mathcal{N}_{MR\to A^\prime R^\prime}$ with $A_2\to A^\prime$ and $B_2$ trivial.  
\end{remark}

\subsubsection{A Signaling Chain Rule for Maximal Output Conditional Entropies}

We now present a chain rule for the optimized conditional entropies $\tilde H^\uparrow_\alpha$, $\alpha\in[\frac{1}{2},1)$, applicable to this task. In \cite{Arnon-Friedman.2019} the smooth max-entropy is bounded via the entropy accumulation theorem \cite{Dupuis.2020}, which builds on chain rules for the non-optimized conditional Rényi entropies under a Markov condition. The chain rule below concerns the optimized entropies and requires no such condition. Since it holds for all $\alpha\in[\frac{1}{2},1)$, it in principle allows one to optimize the resulting statements and certification rates over $\alpha$. We leave this, however, to the work in preparation \cite{Hahn.in.prep}.

\begin{theorem}[Signaling Max-Output Chain Rule]\label{thm:signaling.sup.chain.rule}
For any $\alpha\in[1/2,1)$ and quantum channels $\Phi:Q\to Q_1A_1B_1, \Psi:Q_1\to A_2B_2$, we have that
\begin{align}
   \sup_{\rho\in\cD(\cH_Q)} \tilde H_\alpha^\uparrow(A_1A_2|B_1B_2)_{(\id_{A_1B_1}\otimes\Psi)\circ\Phi(\rho)} &\leq  \sup_{\rho\in\cD(\cH_Q)} \tilde H_\alpha^\uparrow(A_1|B_1)_{\tr_{Q_1}[\Phi(\rho)]} +\sup_{\rho\in\cD(\cH_{Q_1})} \tilde H_\alpha^\uparrow(A_2|B_2)_{\Psi(\rho)}.
\end{align}
In the mixed norm picture this is equivalent to 
\begin{align}
    \|(\id_{A_1B_1}\otimes\Psi)\circ\Phi \|^+_{(Q:1)\to (B_1B_2:1,A_1A_2:\alpha)}\leq \|\Phi\|^+_{(Q:1)\to (B_1:1,A_1:\alpha,Q_1:1)} \|\Psi\|^+_{(Q_1:1)\to (B_2:1,A_2:\alpha)}
\end{align} which holds true for any $\alpha\in[1/2,1]$.
\end{theorem}
By iteration, this directly extends to arbitrarily many channels $\Phi_i:Q_{i-1}\to Q_iA_{i}B_{i}$ with $Q_0=Q$ as in \cref{fig:sequential.rounds} yielding
\begin{align}
    \sup_{\rho\in\cD(\cH_Q)} \tilde H_\alpha^\uparrow(A_1\dots A_n|B_1\dots B_n)_{\rho_{A_1B_1\dots A_nB_n}}\leq \sum_{i=1}^n\sup_{\rho\in\cD(\cH_{Q_{i-1}})} \tilde H_\alpha^\uparrow(A_i|B_i)_{\tr_{Q_i}[\Phi_i(\rho)]} 
\end{align} where $\rho_{A_1B_1\dots A_nB_n}$ denotes the output after $n$ rounds with $Q_n$ traced out.

\begin{figure}[h]
\centering
\begin{tikzpicture}[
    >={Stealth[length=2mm]},
    box/.style   = {draw, rounded corners=1.5pt, fill=black!5,
                    inner sep=2pt, font=\normalsize},
    wire/.style  = {thick},
    lab/.style   = {font=\footnotesize, inner sep=1.5pt},
    note/.style  = {font=\footnotesize\itshape, inner sep=2pt},
]

\node[box, minimum width=9mm, minimum height=28mm] (Phi) at (1.65, 0)    {$\Phi$};
\node[box, minimum width=9mm, minimum height=15mm] (Psi) at (5.05,-0.80) {$\Psi$};

\draw[wire,->] (0.20,0) -- (1.20,0) node[lab, midway, above] {$Q$};

\draw[wire,->] (2.10,-0.80) -- (4.60,-0.80) node[lab, midway, above] {$Q_1$};

\draw[wire,->] (2.10, 0.90) -- (7.30, 0.90) to[out=0,in=180] (8.60, 0.90);
\draw[wire,->] (2.10, 0.20) -- (7.30, 0.20) to[out=0,in=180] (8.60,-0.35);
\draw[wire,->] (5.50,-0.35) -- (7.30,-0.35) to[out=0,in=180] (8.60, 0.25);
\draw[wire,->] (5.50,-1.25) -- (7.30,-1.25) to[out=0,in=180] (8.60,-1.25);

\node[lab, above] at (2.50, 0.90) {$B_1$};
\node[lab, above] at (2.50, 0.20) {$A_1$};
\node[lab, above] at (5.90,-0.35) {$B_2$};
\node[lab, above] at (5.90,-1.25) {$A_2$};

\node[note] at (5.05, 0.55) {$\id_{A_1B_1}$};

\draw[dashed, rounded corners=3pt, black!55]
      (0.95,-1.75) rectangle (5.65, 1.45);
\node[note, black!55] at (3.30, 1.70) {$(\id_{A_1B_1}\otimes\Psi)\circ\Phi$};

\draw[decorate, decoration={brace, amplitude=4pt}]
      (8.95, 1.15) -- (8.95, 0.35)
      node[lab, midway, right=6pt] {$B_1B_2:\,1$};
\draw[decorate, decoration={brace, amplitude=4pt}]
      (8.95,-0.35) -- (8.95,-1.15)
      node[lab, midway, right=6pt] {$A_1A_2:\,\alpha$};

\end{tikzpicture}
\caption{The channel structure of 
\Cref{thm:signaling.sup.chain.rule}. The first channel
$\Phi:Q\to A_1B_1Q_1$ acts on the input register and passes the residual register $Q_1$ on to $\Psi:Q_1\to A_2B_2$, while the systems $A_1B_1$ produced in the first step are left
untouched; the global map is $(\id_{A_1B_1}\otimes\Psi)\circ\Phi$, indicated by the dashed
box. Later channels may therefore depend on everything the earlier ones produced, so the structure can genuinely be signaling. The single crossing on the right is the reordering step $(2)$ in the proof: the composition naturally delivers the output in the order $(B_1,A_1,B_2,A_2)$, whereas the norm on the left-hand side of the chain rule groups it as $(B_1B_2:1,\,A_1A_2:\alpha)$ which is important for the interpretation as conditional entropies. 
Setting $Q=Q^\prime Q_1$ and $\Phi=\Phi^\prime_{Q^\prime\to A_1B_1} \otimes\id_{Q_1}$ recovers the tensor-product setting of \Cref{thm:p.leq.1.cond.multiplicativity}.} \label{fig:composed.channels}
\end{figure}
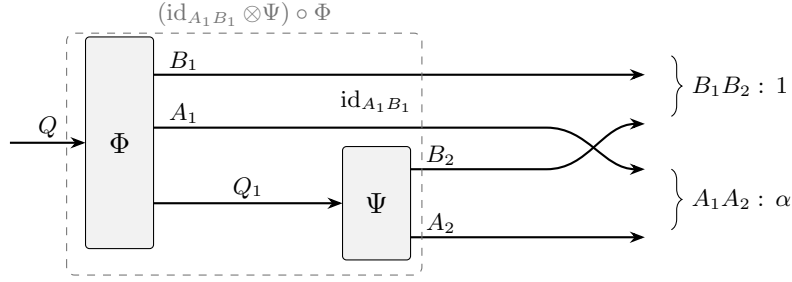

Before giving its proof we discuss two points.
\begin{remark}\label{rem:optimized.vs.unpotimized}
Note that the regime of $\alpha$ for which this mixed-quasi-norm statement holds is maximal in the sense that this mixed norm statement breaks down for $\alpha>1$: already for tensor-product channels with trivial conditioning systems $B_1,B_2$ it reduces to the multiplicativity of the maximal output $1\to\alpha$ norm, which fails for all $\alpha>1$ \cite{Werner.2002, Hayden.2008}; natural strengthenings of chain rules for $\alpha>1$ also fail in the device-independent setting \cite[Section 4]{Wooltorton.2026}. I.e.\ the $\alpha>1$ signaling chain rule with minimal output conditional sandwiched-$\alpha$ entropies, as compared to the maximal and $\alpha<1$ above, does not hold. Interestingly, when replacing the optimized sandwiched conditional entropies by their unoptimized version
\begin{align}
    \tilde H_\alpha^\downarrow(A|B)_\rho=-\tilde D_\alpha(\rho_{AB}\|\1_A\otimes \rho_B),
\end{align} a chain rule of this type can be recovered for $\alpha\in(1,2)$ under a non-signaling condition \cite[Lemma 3.6]{Metger.2024}. See also \cite[Lemma 5.2]{Arqand.2024}.
\end{remark}

\begin{remark}
    Note further that for tensor product channels, we recover the (non-trivial) submultiplicativity part of \Cref{thm:p.leq.1.cond.multiplicativity} by setting $Q=Q^\prime Q_1$ and $\Phi_{Q^\prime Q_1\to A_1B_1Q_1}=\Phi_{Q^\prime\to A_1B_1}^\prime\otimes\id_{Q_1}$ and that by \cref{cor:mixed.positivity.normed} the mixed norm statement is equal to itself with the $^+$ superscripts dropped, i.e. when we optimize over arbitrary non-positive trace-normalized inputs.
\end{remark}

\begin{proof}[Proof of \Cref{thm:signaling.sup.chain.rule}]
We have
\begin{align}
     \|(\id_{B_1A_1}\otimes\Psi)\circ\Phi \|^+_{(Q:1)\to (B_1B_2:1,A_1A_2:\alpha)} &\overset{(1)}{\leq} \|\Phi\|^+_{(Q:1)\to (B_1:1,A_1:\alpha,Q_1:1)}\\&\qquad\cdot \|\id_{B_1A_1}\otimes\Psi\|^+_{(B_1:1,A_1:\alpha,Q_1:1)\to (B_1:1,B_2:1,A_1:\alpha,A_2:\alpha)} \\&\overset{(2)}{\leq}  \|\Phi\|^+_{(Q:1)\to (B_1:1,A_1:\alpha,Q_1:1)}\\ &\qquad\cdot \|\id_{B_1A_1}\otimes\Psi\|^+_{(B_1:1,A_1:\alpha,Q_1:1)\to (B_1:1,A_1:\alpha,B_2:1,A_2:\alpha)} \\
     &\overset{(3)}{\leq} \|\Phi\|^+_{(Q:1)\to (B_1:1,A_1:\alpha,Q_1:1)}\\&\qquad\cdot \|\Psi\|^+_{\cb,(Q_1:1)\to (B_2:1,A_2:\alpha)} \\
     &\overset{(4)}{=} \|\Phi\|^+_{(Q:1)\to (B_1:1,A_1:\alpha,Q_1:1)} \|\Psi\|^+_{(Q_1:1)\to (B_2:1,A_2:\alpha)},
\end{align} where the first inequality $(1)$ is just the submultiplicativity under composition, the second inequality $(2)$ is a consequence of the complete contractiveness under the SWAP with $1\geq \alpha\geq\frac{1}{2}$ from \Cref{thm:QuasiNormMinkowskiInequality}. $(3)$ is the definition of the completely bounded norm with \Cref{lem:cb.norm.via.t} and $(4)$ is as in the proof of \Cref{thm:p.leq.1.cond.multiplicativity}.  
\end{proof}

\noindent Thus what we have shown is that
\begin{align}
   \|(\id_{A_1B_1}\otimes\Psi)\circ\Phi \|_{(Q:1)\to (B_1B_2:1,A_1A_2:\alpha)}\leq \|\Phi\|_{(Q:1)\to (B_1:1,A_1:\alpha,Q_1:1)} \|\Psi\|_{(Q_1:1)\to (B_2:1,A_2:\alpha)} 
\end{align} holds for CP maps $\Phi,\Psi$.
This is a somewhat more general result than what is technically required for the DI entanglement certification setup. Operationally, iterating \Cref{thm:signaling.sup.chain.rule} reduces the analysis of an $n$-round, adaptive certification protocol to $n$ single-round optimization problems, each of which only involves the channel of a single round, while the order $\alpha\in[\frac12,1)$ remains free and informally speaking can be optimized at the end. We leave deriving numerical certification rates and hence quantifying the resulting improvement over the smooth max-entropy approach of \cite{Arnon-Friedman.2019} to the work in preparation \cite{Hahn.in.prep}.

\subsubsection{A Signaling Chain Rule for Minimal Output Conditional Entropies}

For completeness, although not of immediate use in device-independent entanglement certification, we also show a similar chain rule statement for the minimal, instead of maximal, sandwiched conditional output entropies. Note the slight difference, in that the minimal conditional entropy of the output of $\Phi$ is conditioned on the register $Q_1$, instead of it being traced out, and the minimal conditional entropy of the output of $\Psi$ is considered within environments of arbitrary sizes.

\begin{theorem}\label{thm:signaling.inf.chain.rule}
    Let  $\alpha\in[\frac{1}{2},1)$ then for quantum channels $\Phi:Q\to Q_1A_1B_1, \Psi:Q_1\to A_2B_2$ it holds that
\begin{align}
   \inf_{\rho\in\cD(\cH_Q)} \!\!\tilde H_\alpha^\uparrow(A_1A_2|B_1B_2)_{(\id_{A_1\!B_1}\!\otimes\Psi)\circ\Phi(\rho)} &\!\geq \!\! \inf_{\rho\in\cD(\cH_Q)} \!\!\tilde H_\alpha^\uparrow(A_1|B_1Q_1)_{\Phi(\rho)}+\!\!\!\inf_{E,\rho\in\cD(\cH_{EQ_1})}\!\!\!\!\!\!\tilde H_\alpha^\uparrow(A_2|B_2E)_{(\id_E\!\otimes\Psi)(\rho)}.
\end{align}
In the mixed-co-quasi-norm picture this is 
\begin{align}
    \|(\id_{A_1B_1}\otimes\Psi)\circ\Phi \|^+&_{\co,(Q:1)\to (B_1B_2:1,A_1A_2:\alpha)}\geq \|\Phi\|^+_{\co,(Q:1)\to (B_1:1,Q_1:1,A_1:\alpha)}\cdot \|\Psi\|^+_{\cb,\co,(Q_1:1)\to (B_2:1,A_2:\alpha)}.
\end{align}
\end{theorem}

\begin{remark}
    Analogously to above, \Cref{cb.co.1.to.1.p.additivity} can be viewed as a special case of \Cref{thm:signaling.inf.chain.rule} for tensor-product channels, where we additionally prove achievability via \Cref{lem:basic.4.index.multiplicativity}. This implies the additivity of the minimal output completely bounded conditional $\alpha<1$ entropies.
\end{remark}

\begin{proof}[Proof of \Cref{thm:signaling.inf.chain.rule}]
Just as in the proof above we begin with
\begin{align}
    \|(\id_{A_1B_1}\otimes\Psi)\circ\Phi\|^+&_{\co,(Q:1)\to (B_1B_2:1,A_1A_2:\alpha)} \\&\overset{(1)}{\geq} 
    \|\Phi\|^+_{\co,(Q:1)\to (B_1:1,Q_1:1,A_1:\alpha)}\\&\qquad\times\|\id_{B_1}\otimes\Psi\otimes\id_{A_1}\|^+_{\co,(B_1:1,Q_1:1,A_1:\alpha)\to (B_1:1,B_2:1,A_2:\alpha,A_1:\alpha)} \\ 
    &\overset{(2)}{\geq} \|\Phi\|^+_{\co,(Q:1)\to (B_1:1,Q_1:1,A_1:\alpha)}\\&\qquad\times\|\Psi\otimes\id_{A_1}\|^+_{\cb,\co,(Q_1:1,A_1:\alpha)\to (B_2:1,A_2:\alpha,A_1:\alpha)} \\ &\overset{(3)}{=}\|\Phi\|^+_{\co,(Q:1)\to (B_1:1,Q_1:1,A_1:\alpha)}\|\Psi\|^+_{\cb,\co,(Q_1:1)\to (B_2:1,A_2:\alpha)},
\end{align} where the first inequality $(1)$ is the standard supermultiplicativity under composition of the co-quasi-norm, \Cref{prop:standard.sub.and.sup.mult}, together with the fact that we can swap $A_1,A_2$ as they carry the same index. The second inequality $(2)$ follows from the definition of the completely bounded co-quasi-norm, \Cref{def:cb.co.quasi.norm}, and the last equality follows from \Cref{cor:cb.id.right.quasi}, which requires $\alpha\in[\frac{1}{2},1]$.
\end{proof}

\noindent Note that the same proof, applied to $\cb$ versions of the above mixed-co-norms yields also 
\begin{align}
    \|(\id_{A_1B_1}\otimes\Psi)\circ \Phi\|^+_{\cb,\co,(Q:1)\to (B_1B_2:1,A_1A_2:\alpha)} \geq \|\Phi\|^+_{\cb,\co,(Q:1)\to (B_1:1,Q_1:1,A_1:\alpha)}\|\Psi\|^+_{\cb,\co,(Q_1:1)\to (B_2:1,A_2:\alpha)} 
\end{align} for any $\alpha\in[\frac{1}{2},1]$ which in terms of entropies implies
\begin{align}
    \inf_{E,\rho} \tilde H_\alpha^\uparrow(A_1A_2|B_1B_2E)_{(\id_E\otimes(\id_{A_1B_1}\otimes\Psi)\circ\Phi)(\rho)} &\geq \inf_{E,\rho} \tilde H_\alpha^\uparrow(A_1|B_1Q_1E)_{(\id_E\otimes\Phi)(\rho)}\\&\quad+\inf_{E,\rho} \tilde H_\alpha^\uparrow(A_2|B_2E)_{(\id_E\otimes\Psi)(\rho)}
\end{align} for any two quantum channels $\Phi,\Psi$.

\section{Observations and Outlook}\label{sec:out.multi.indexed}

In this paper we combined recent results from matrix analysis with some simple operator space results in order to enrich the available non-commutative toolbox, available to mathematicians and information theorists alike, by multi-indexed Schatten spaces, and thereby narrowing the gap between the commutative and the non-commutative multi-indexed spaces. We further derived several important tools for handling linear and CP maps between such spaces and demonstrated several use cases within quantum information theory in form of additivity statements and chain rules for Rényi conditional entropies under channels. In the following we want to discuss a few observations and open directions. 

\vspace{-0.1cm}
\paragraph{Multi-Indexed Theory without Operator Space Theory} Our construction of multi-indexed Schatten quasi-norms is matrix analytic, yet the proof of the factorization formula \Cref{thm:Multi.Thm.Pisier.like} still uses a small part of Pisier's theory: the operator-valued factorization formula \Cref{thm:VariationalSpOperatorspace} and the concavity statement \Cref{lem.app.op.concavity}, which relies on the complete isometric embedding of operator spaces into some $\cB(\cK)$. In the two-indexed case no such input was necessary \cite{Kochanowski.2026}. It would be interesting to find a purely matrix-analytic proof of \Cref{thm:Multi.Thm.Pisier.like}.

\vspace{-0.2cm}
\paragraph{Expanding the Framework}
One theme which recurred throughout this work and its predecessors is the question of how far the framework can ultimately be pushed. After having extended it here to, among others, arbitrary numbers of indices, the most immediate direction is that of the indices themselves: We required the indices $q_i$ to be of the same sign, either all positive or all negative. This was to simplify the situation and because it was not necessary for the applications to linear maps and entropic chain rules presented here. However, the condition actually used in the proofs is $|\frac{1}{q_i}-\frac{1}{q_j}|\leq1$, which can also be satisfied by indices of mixed sign, e.g.\ $(2,-2)$, leaving their interpretation and meaning open. Part of what makes this case interesting is that mixed signs are the one place where an otherwise very ordered picture breaks. For $q,p\geq1$ we have norms, duality and convexity, whereas for $q,p\leq1$ we demonstrated quasi- or anti-norms, reverse duality and concavity; the two halves mirror each other closely. For indices of mixed sign we would not expect either picture to apply, and in particular neither a (reverse) duality nor a (reverse) triangle inequality to hold. Such mixed-sign indices should, however, allow one to extend the reach of completely bounded quantities to input and output tuples which admit no common reference index of a single sign, and to further entropic quantities --- possibly a $-\infty$-conditional Rényi entropy, or an umlaut information of order $\alpha\in[1,\infty]$. Not only are most of the tools required for their analysis therefore presumably already in place, there is also an operational motivation for wanting them.

\vspace{-0.2cm}
\paragraph{More General Entropic Additivities} Further, our multiplicativity statements are proved for output tuples $(\cP)=(1,p)$, which is what the entropic applications in \Cref{sec:ApplictoQI} required. The framework of \Cref{sec:LinerMaps} defines the mixed norms and the tools in dealing with them, however, for arbitrary compatible (in- and) output tuples, and multiplicativity does not stop at two output indices. More general $\cb,1\to(\cP)$ multiplicativity statements would be natural additions to \Cref{sec:ApplictoQI}, currently hindered mainly by the lack of an operational interpretation of the corresponding multi-indexed quantities as entropic quantities.

\vspace{-0.2cm}
\paragraph{Beyond Finite Dimensions} Finally, for simplicity we restricted ourselves to the setting of finite-dimensional Hilbert spaces in this work. Pisier's operator-valued Schatten spaces, though, are defined for general (semi-)finite von Neumann algebras, and the Araki--Masuda realization of the sandwiched Rényi divergences \cite{Berta.2018} is likewise infinite dimensional. We believe that with essentially no modifications our results extend to separable settings as well, and that extensions to type II and type III von Neumann algebras are quite likely for weighted versions. The technical difficulty in proving all our results in separable Hilbert spaces is that we have to be more careful when invoking Sion's minimax \Cref{thm:Sion}. It requires one of the sets over which we optimize to be compact. In this work we could always choose this set to be the set of subnormalized quantum states, which for separable Hilbert spaces, however, is only weak-$^*$ compact by Banach-Alaoglu, implying that one effectively only needs to prove weak-$^*$ upper semi-continuity in \Cref{lem.app.op.concavity} to push our results beyond finite dimensions. Establishing the multi-indexed quasi-norms in the von Neumann algebraic setting would make them available in quantum field theory and in the study of infinite spin systems, where the corresponding entropic quantities are of independent interest.

\noindent\hrulefill
\medskip

\noindent \textbf{AI Statement}
The author used AI tools for proofreading purposes and help with language editing and organization. The author reviewed and verified all resulting text and takes full responsibility for the content of the manuscript.
\newline 

\noindent \textbf{Acknowledgments}
JK would like to thank Thomas Hahn and Aby Philip for discussions related to device-independent entanglement certification and the relevant chain rules. JK acknowledges support from the Program QuanTEdu-France n° ANR-22-CMAS-0001 France 2030.

\newpage

\bibliographystyle{abbrv} 
\bibliography{lit}  
\hfill 
\newpage  

\appendix

\section{Useful Matrix-Analytical Tools}\label{app:matrix.tools}

The proofs of the technical main results of this work require dealing with optimizations of multivariate trace-formulas. These involve sometimes suprema and infima at the same time and Sion's seminal minimax theorem gives sufficient criteria for their exchange.

\begin{theorem}[Ky Fan's Sion's minimax theorem, \cite{KyFan.53} Theorem 2]\label{thm:Sion} Let $X$ be a compact convex subset of a Hausdorff topological vector space and $Y$ a subset of a linear space. Let $f:X\times Y\to \mathbb{R}\cup\{-\infty\}$ be upper semi-continuous on $X$ for fixed $y\in Y$ and concave in $x$ and convex in $y$, then
\begin{align}
    \sup_{x\in X}\inf_{y\in Y}f(x,y)= \inf_{y\in Y}\sup_{x\in X}f(x,y).
\end{align}
\end{theorem}
We note here, that when minimizing and maximizing over multiple variables each, then the function needs to be jointly concave, respectively convex in them.

Two further important tools we require are rewritings of optimizations over a product of suitably homogeneous functions as a sum of suitable powers of these functions. In particular, we recall \cite[Lemmas B.3 and B.4]{Kochanowski.2026}. Such rewritings are relatively straightforward consequence of (reverse) AM-GM inequalities, see also \cite{Rubboli.2025}.

For a supremum of a product of three functions we recall the following.
\begin{lemma}[Supremum of products as sum; Lemma B.3 of \cite{Kochanowski.2026}]\label{lem:Sup.Prod.to.Sum}
Given three positive functions $g,h:\cB^+(\cH)\to [0,\infty],\, f:\cB^+(\cH\otimes\cH)\to [0,\infty]$ and numbers $r_2<0<r_1$ such that $\frac{1}{r_0}:=\frac{2}{r_1}+\frac{1}{r_2}<0$ and
\begin{align}
    g(\lambda C)=\lambda^{r_1}g(C)\quad h(\mu D)=\mu^{r_1}h(D)\quad f(\lambda C,\mu D)= (\lambda\mu)^{-r_2}f(C,D),
\end{align} holds for any $\lambda,\mu\geq0$.
Then 
\begin{align}
    \sup_{C,D\geq0}\left\{g(C)^{\frac{r_0}{r_1}}h(D)^{\frac{r_0}{r_1}}f(C,D)^{\frac{r_0}{r_2}}\right\} = \sup_{C,D\geq0}\left\{\frac{r_0}{r_1}g(C)+\frac{r_0}{r_1}h(D)+\frac{r_0}{r_2}f(C,D)\right\}.
\end{align}
\end{lemma}

Similarly, for an infimum of a product of three functions we recall:
\begin{lemma}[Infimum of products as sum; Lemma B.4 of \cite{Kochanowski.2026}]\label{lem:Inf.Prod.to.Sum}
Given three positive functions $g,h:\cB^+(\cH)\to [0,\infty],\, f:\cB^+(\cH\otimes\cH)\to [0,\infty]$ and numbers $0<r_1,r_2$ s.t. 
\begin{align}
   g(\lambda C)=\lambda^{r_1}g(C), \quad  h(\mu D)=\mu^{r_1}h(D), \quad f(\lambda C,\mu D)=(\lambda\mu)^{-r_2}f(C,D)
\end{align} 
holds for any $\lambda,\mu\geq0$. 
Define $\frac{1}{r_0}:=\frac{2}{r_1}+\frac{1}{r_2}$, then it holds that
\begin{align}
    \inf_{C,D\geq0}g(C)^{\frac{r_0}{r_1}}h(D)^{\frac{r_0}{r_1}}f(C,D)^{\frac{r_0}{r_2}} = \inf_{C,D\geq0}\left\{\frac{r_0}{r_1}g(C)+\frac{r_0}{r_1}h(D)+\frac{r_0}{r_2}f(C,D)\right\}.
\end{align}
\end{lemma}

Finally we will frequently use the following factorization formulas for Schatten spaces. 

\begin{proposition}[Factorization formulas for $\cS_q$; Proposition 1.1 of \cite{Kochanowski.2026}]\label{prop:SpVariationalFormulas}
Let $q,p\neq0$ and set $\frac{1}{r}=\big|\frac{1}{q}-\frac{1}{p}\big|$, then for any $X\in\cS_q(\cH)$ it holds that
\begin{align}
    \|X\|_q=\begin{cases}
        \underset{\substack{X=\Pi_aX\Pi_b\\a,b\in\cS^+_{2r}(\cH)}}{\inf}\|a\|_{2r}\|b\|_{2r}\|a^{-1}Xb^{-1}\|_p  &\text{ if } q\leq p, \\
       \underset{\substack{X=\Pi_aX\Pi_b\\a,b\in\cS_{2r}^+(\cH)}}{\sup}\|a\|^{-1}_{2r}\|b\|^{-1}_{2r}\|aXb\|_p  &\text{ if } q\geq p.
    \end{cases}
\end{align} Moreover, when $X$ is normal\footnote{An operator $X$ is normal when $XX^*=X^*X$.} one may choose $a=b$.
\end{proposition}

These are essentially just consequences of the generalized Hölder's inequality for Schatten spaces. A well-known special case for $p=1$ of these is \cite[Lemma 3.4]{book:Tomamichel.2016}. We note here, without proving, that these also hold for negative indices when restricting to $a=b^*>0$ and $X>0$.

\section{Operator Spaces}\label{App:OpSpaces}

We briefly recall what is required of operator space theory and Pisier's operator valued Schatten spaces and refer, for a proper exposition of operator spaces to \cite{book:Pisier.2003, book:Paulsen.2003}, operator-valued Schatten spaces to \cite{book:Pisier.1998}, and for an overview aimed at quantum information theory to \cite{Beigi.2023}. 

A (concrete) \emph{operator space} is a closed subspace $\cX\subseteq\cB(\cK)$. The linear spaces
\begin{align}
    M_{m,n}(\mathcal{X}):=\Big\{[X_{ij}]_{i\in[m],j\in[n]}|X_{ij}\in\mathcal{X}\Big\}, \quad M_{n}(\mathcal{X})\equiv M_{n,n}(\mathcal{X})
\end{align} of $\mathcal{X}$-valued $m\times n$ matrices can be endowed with a natural norm by viewing elements of $M_{m,n}(\mathcal{X})$ as linear maps in $\mathcal{B}(\mathcal{K}^n,\mathcal{K}^m)$, where $\mathcal{K}^n:=\bigoplus_{i=1}^n\mathcal{K}=\mathbb{C}^n\otimes\mathcal{K}$, via the identification $M_{m,n}(\mathcal{X})\subset M_{m,n}(\mathcal{B}(\mathcal{K}))\simeq \mathcal{B}(\mathcal{K}^n,\mathcal{K}^m)$:
\begin{align}
    X=[X_{ij}]\in M_{m,n}(\mathcal{X}) \leftrightarrow X:\mathcal{K}^n \to \mathcal{K}^m;\quad  v^n=[v_1,\dots, v_n]^\intercal\mapsto \Big[\sum_j X_{ij}v_j\Big]^\intercal\,.
\end{align}
Hence the space $M_{m,n}(\mathcal{B}(\mathcal{K}))$ is naturally equipped with the norms induced by $\mathcal{B}(\mathcal{K}^n,\mathcal{K}^m)$, i.e.,
\begin{align} 
    \|X\|_{m,n}\equiv\|X\|_{M_{m,n}(\mathcal{X})}&:=\sup\{\|Xv^n\|_{\mathcal{K}^m}|v^n\in\mathcal{K}^n, \|v^n\|_{\mathcal{K}^n}\leq 1\} \label{equ:def.mxn.norm} \\ &= \sup\left\{\left(\sum_{i=1}^m\bigg\|\sum_{j=1}^nX_{ij}v_j\bigg\|_{\mathcal{K}}^2\right)^{\frac{1}{2}}\Bigg|\sum_{j=1}^n\|v_j\|_{\mathcal{K}}^2=1\right\},
\end{align} where $\|\cdot\|_{\mathcal{K}}$ denotes the norm on $\mathcal{K}$ induced by its inner product. In the case $m=n$, we write $\|\cdot\|_n:=\|\cdot\|_{n,n}$. For simplicity we will also denote $\|X\|_{\mathcal{X}}\equiv \|X\|_{M_1(\mathcal{X})}$.

\begin{proposition} \label{prop:central.properties.OSnorms}
The family of norms defined above satisfies the following two main properties, namely for any $m,n\in\mathbb{N}$,
\begin{align}
\textup{(i)} \quad & \|FXG\|_{m} \leq \|F\|_{m,n}\|X\|_{n}\|G\|_{n,m} \quad \forall F,G^*\in M_{m,n}(\mathbb{C}), X\in M_n(\mathcal{X}), \\
\textup{(ii)} \quad & \|X\oplus Y\|_{m+n} = \max\{\|X\|_{m},\|Y\|_{n}\} \quad \forall X\in M_n(\mathcal{X}), Y\in M_m(\mathcal{X}),
\end{align} where 
$X\oplus Y= \begin{pmatrix}
    X & 0 \\
    0 & Y
\end{pmatrix}$, and $FXG\equiv (F\otimes \1_\mathcal{K})X(G\otimes \1_\mathcal{K})$.
\end{proposition}
More generally,
\begin{definition} A linear space $\mathcal{X}$ with a family of norms $\|\cdot \|_{m,n}$ on $M_{m,n}(\mathcal
X)$ that satisfy the above properties (i) and (ii) is called an \textit{(abstract) operator space}.
\end{definition}

\noindent By Ruan's theorem \cite[Ruan's Theorem]{book:Pisier.2003} it turns out that the only linear spaces $\mathcal{X}$ with endowed norms $\|\cdot \|_{m,n}$ on $M_{m,n}(\mathcal{X})$ satisfying properties (i) and (ii) are closed linear subspaces of some $\mathcal{B}(\mathcal{K})$, where $\mathcal{K}$ is a (possibly infinite-dimensional) Hilbert space.  Importantly, this inclusion is a completely isometric embedding \cite{book:Pisier.2003}, meaning that for any $X\in\cX, Y\in M_n[\cX]$ it holds that
\begin{align}
    \|X\|_{\cX}=\|X\|_{\cB(\cK)} = \|X\|_\infty, \text{ moreover, }  \|Y\|_{M_n(\cX)}=\|Y\|_{\cB(\mathbb{C}^n\otimes\cK)} = \|Y\|_\infty
\end{align} where the second norm is the norm of $X\subset\cB(\cK)$ via this embedding, which is nothing but the operator or Schatten-$\infty$-norm on $\cK$. Similarly for $Y$.

\noindent Given some map $\phi:\cX\to \cY$ between operator spaces $\cX,\cY$ then we denote with $\phi_n=\id_n\otimes \phi$ the canonical dilation of $\phi$ acting on $M_n(\cX)\to M_n(\cY)$ \cite{book:Pisier.2003}.
This map $\phi$ is now called \emph{completely bounded (cb)} from $\cX\to \cY$ whenever its associated completely bounded norm
\begin{align}
    \|\phi\|_{\cb,\cX\to \cY}:=\sup_{n\in\mathbb{N}}\|\phi_n\|_{M_n(\cX)\to M_n(\cY)}<\infty
\end{align} is finite. We denote the set of all completely bounded maps from $\cX\to \cY$ with $\mathcal{CB}(\cX,\cY)$
Analogously we call a map \emph{completely contractive} if $\|\phi\|_{\cb,\cX\to \cY}\leq1$ and \emph{completely isometric} if $\phi_n$ is an isometry for each $n\in\mathbb{N}$.

\noindent Using completely bounded maps we can define \emph{dual} operator spaces. 
Given an operator space $\cX$, its Banach space dual is given as $\cX^*=\cB(\cX,\mathbb{C})= \mathcal{CB}(\cX,\mathbb{C})$, which is again a Banach space under the corresponding induced norm of $\cB(\cX,\mathbb{C})$. The equality of $\cB(\cX,\mathbb{C})= \mathcal{CB}(\cX,\mathbb{C})$ is a consequence of the commutative (as one-dimensional) nature of $\mathbb{C}$, effectively a consequence of Smith's lemma \cite[Proposition 1.12]{PisierXu.2003}.
It turns out that if we give it the structure
\begin{align}\label{equ:prelim.pisier.op.dual.norm}
    M_n(\cX^*)= \mathcal{CB}(\cX,M_n(\mathbb{C}))
\end{align} by associating elements of $M_n(\mathcal{CB}(\cX,\mathbb{C}))$ in the natural way with $\mathcal{CB
}(\cX,M_n(\mathbb{C}))$, that is, for any $\Phi=[\phi_{ij}]\in M_n(\mathcal{CB}(\cX,\mathbb{C}))$ we associate it to
\begin{align}
     &\tilde{\Phi}\in \mathcal{CB}(\cX,M_n(\mathbb{C})) \\
     \cX\ni X &\mapsto [\phi_{ij}(X)]\in M_n(\mathbb{C})
\end{align} \cite[Section 2.3]{book:Pisier.2003}.
The in this way constructed $\cX^*$ with the family of norms $\|\cdot\|_n$ via \eqref{equ:prelim.pisier.op.dual.norm} is called the \emph{operator space dual} of $\cX$.

\begin{remark}
    For readers who are less operator space and algebraically minded there is one important takeaway from this passage. Namely that every operator space embeds completely isometrically into some $\cB(\cK)$ and thus its norm is, through this embedding, `just' the regular operator ($p=\infty$-Schatten) norm over the Hilbert space $\cK$. Implicitly thus also always a convex and so never a `true' quasi-norm which is not a norm.
\end{remark}

\section{Proofs Omitted from Section \cref{sec:Multi.Pisier}}\label{app:Generalized.Variational.Expressions}

\subsection{Proof of \Cref{thm:VariationalSpOperatorspace}}

\begin{proof}[Proof of \Cref{thm:VariationalSpOperatorspace}]
We will prove (i) by separately showing upper and lower bounds and (ii) will then follow by duality. This proof is based on the proofs of \cite[Lemma 3.1]{Fawzi.2026}, which itself is inspired by the one of \cite[Theorem 4.5]{Beigi.2023}
Note that to keep notations short and as done before, we omit writing identity operators, e.g.\ for $a,b\in \mathcal{S}_{2r}(\mathcal{H})$ and $Y\in \mathcal{S}_\infty[\mathcal{H};\mathcal{X}]$ we will write $aYb\equiv(a\otimes\1)Y(b\otimes\1)$. 

    (i) Let $X=aYb\in \mathcal{S}_p[\mathcal{H};\mathcal{X}]$ with $a,b\in \mathcal{S}_{2r}(\mathcal{H}),Y\in \mathcal{S}_q[\mathcal{H};\mathcal{X}]$ and suppose $Y=cZd$ with $c,d\in \mathcal{S}_{2q}(\mathcal{H}), Z\in \mathcal{S}_\infty[\mathcal{H};{\mathcal{X}}]$, then $X=ac Z db$ and by H\"{o}lder's inequality $\|ac\|_{2p}\leq \|a\|_{2r}\|c\|_{2q}$, since $\frac{1}{p}=\frac{1}{r}+\frac{1}{q}$, we get $ac,db\in \mathcal{S}_{2p}(\mathcal{H})$. Hence
    \begin{align}
\|X\|_{\mathcal{S}_p[\mathcal{H};\mathcal{X}]}&\leq \|ac\|_{2p}\|Z\|_{\mathcal{S}_\infty[\mathcal{H};\mathcal{X}]}\|db\|_{2p} \\&\leq \|a\|_{2r}\|c\|_{2q}\|Z\|_{\mathcal{S}_\infty[\mathcal{H};\mathcal{X}]}\|d\|_{2q}\|b\|_{2r}\,.
    \end{align} 
     After minimization over $c,d,Z$, we obtain
    \begin{align}
\|X\|_{\mathcal{S}_p[\mathcal{H};\mathcal{X}]} \leq \|a\|_{2r}\|Y\|_{\mathcal{S}_q[\mathcal{H};\mathcal{X}]}\|b\|_{2r} \implies  \|X\|_{\mathcal{S}_p[\mathcal{H};\mathcal{X}]} \leq \inf_{\substack{a,b\in \mathcal{S}_{2r}\\Y\in \mathcal{S}_q[\mathcal{H};\mathcal{X}]}}\|a\|_{2r}\|Y\|_{\mathcal{S}_q[\mathcal{H};\mathcal{X}]}\|b\|_{2r}\,.
    \end{align}

    On the other hand, let $\epsilon>0$. Then there exists $Y\in \mathcal{S}_\infty[\mathcal{H};\mathcal{X}]$ and $a,b\in \mathcal{S}_{2p}(\mathcal{H})$ such that $X=aYb$ and 
    \begin{align}
\|X\|_{\mathcal{S}_p[\mathcal{H};\mathcal{X}]}+\epsilon \geq \|a\|_{2p}\|Y\|_{\mathcal{S}_\infty[\mathcal{H};\mathcal{X}]}\|b\|_{2p}\,.
    \end{align}
    By performing a polar decomposition of $a$ and $b$ and absorbing the unitaries into $Y$, which does not change the norm, we may assume $a,b$ to be positive semidefinite and thus $a^{\frac{p}{q}},b^{\frac{p}{q}}\in \mathcal{S}_{2q}(\mathcal{H})$ and $a^{\frac{p}{r}},b^{\frac{p}{r}}\in \mathcal{S}_{2r}(\mathcal{H})$. As a result, $X = a^{\frac{p}{r}} a^{\frac{p}{q}} Y b^{\frac{p}{q}} b^{\frac{p}{r}}$. Hence

    \begin{align}
        \inf_{\substack{X=cWd\\c,d\in \mathcal{S}_{2r}(\mathcal{H}), W\in \mathcal{S}_q[\mathcal{H};\mathcal{X}], }}\|c\|_{2r}\|W\|_{\mathcal{S}_q[\mathcal{H};\mathcal{X}]}\|d\|_{2r} &\leq \|a^{\frac{p}{r}}\|_{2r}\|a^{\frac{p}{q}}Yb^{\frac{p}{q}}\|_{\mathcal{S}_q[\mathcal{H};\mathcal{X}]}\|b^{\frac{p}{r}}\|_{2r} \\
        &\leq \|a^{\frac{p}{r}}\|_{2r}\|a^{\frac{p}{q}}\|_{2q}\|Y\|_{\mathcal{S}_\infty[\mathcal{H};\mathcal{X}]}\|b^{\frac{p}{q}}\|_{2q}\|b^{\frac{p}{r}}\|_{2r} \\&=\|a\|_{2p}\|Y\|_{\mathcal{S}_\infty[\mathcal{H};\mathcal{X}]} \|b\|_{2p} = \|X\|_{\mathcal{S}_p[\mathcal{H};\mathcal{X}]}+\epsilon.
    \end{align}
    Since $\epsilon>0$ is arbitrary the claim follows.

    \medskip 

   (ii) by the duality $\mathcal{S}_{q}[\mathcal{H};\mathcal{X}]^*=\mathcal{S}_{q^\prime}[\mathcal{H};\mathcal{X}^*]$, i.e.\ there is a complete isometry between these two operator spaces, where $1/q+1/q^\prime=1$ (see \cite[Proposition 4.3]{Beigi.2023}). For $1/p + 1/p^\prime = 1$, we have $1 \leq q^\prime \leq p^\prime$ and also $\frac{1}{r} = \frac{1}{q^\prime} - \frac{1}{p^\prime}$, we thus obtain
    \begin{align}
    \|X\|_{\mathcal{S}_q[\mathcal{H};\mathcal{X}]}
        &=\sup_Y\frac{|\Tr[Y^*X]|}{\|Y\|_{\mathcal{S}_{q^\prime}[\mathcal{H};\mathcal{X}^*]}}\\
        &\overset{\operatorname{(i)}}{=} \sup_Y\sup_{ Y = aZb}\frac{|\Tr[Y^*X]|}{\|a\|_{2r}\|b\|_{2r}\|Z\|_{\mathcal{S}_{p^\prime}[\mathcal{H};\mathcal{X}^*]}} \\ &=\sup_{a,b,Z}\frac{|\Tr[b^*Z^*a^*X]|}{\|a\|_{2r}\|b\|_{2r}\|Z\|_{\mathcal{S}_{p^\prime}[\mathcal{H};\mathcal{X}^*]}}\\
        &=\sup_{a,b,Z}\frac{|\Tr[Z^*(a^*Xb^*)]|}{\|a\|_{2r}\|b\|_{2r}\|Z\|_{\mathcal{S}_{p^\prime}[\mathcal{H};\mathcal{X}^*]}}\\
&=\sup_{a,b}\frac{\|aXb\|_{\mathcal{S}_p[\mathcal{H};\mathcal{X}]}}{\|a\|_{2r}\|b\|_{2r}}.
    \end{align}

\medskip

To derive the second lines of both of these equations we will argue first for positive semidefiniteness of the optimizers. This follows via polar decomposition and unitary invariance of both the Schatten-$2r$-norms and local unitary invariance of the Pisier norms $\cS_q[\cH;\cX]$. 
Explicitly,
\begin{align}
    \|X\|_{\cS_p[\cH;\cX]} &= \inf_{\substack{a,b\in \mathcal{S}_{2r}(\mathcal{H}), Y\in\mathcal{S}_q[\mathcal{H}; \mathcal{X}] \\X=aYb}}\|a\|_{2r}\|b\|_{2r}\|Y\|_{\cS_q[\cH;\cX]} \\ & = \inf_{\substack{a,b\in \mathcal{S}_{2r}(\mathcal{H}), Y\in\mathcal{S}_q[\mathcal{H}; \mathcal{X}] \\X=|a|UYV|b^*|}}\||a|\|_{2r}\||b|\|_{2r}\|U_1YV_1\|_{\cS_q[\cH;\cX]} \\ & = \inf_{\substack{a,b\in \mathcal{S}_{2r}(\mathcal{H}), Y\in\mathcal{S}_q[\mathcal{H}; \mathcal{X}] \\X=|a|UYV|b^*|}}\||a|\|_{2r}\||b|\|_{2r}\|Y\|_{\cS_q[\cH;\cX]} 
\end{align} and likewise
\begin{align}
    \|X\|_{\mathcal{S}_q[\mathcal{H};\mathcal{X}]} &= \sup_{a,b\in \mathcal{S}_{2r}(\mathcal{H})}\|a\|_{2r}^{-1}\|b\|_{2r}^{-1}\|aXb\|_{\mathcal{S}_p[\mathcal{H};\mathcal{X}]} \\
    &= \sup_{a,b\in \mathcal{S}_{2r}(\mathcal{H})}\||a|\|_{2r}^{-1}\||b|\|_{2r}^{-1}\|U_1|a^*|X|b|V_1\|_{\mathcal{S}_p[\mathcal{H};\mathcal{X}]} 
    \\
    &= \sup_{a,b\in \mathcal{S}_{2r}(\mathcal{H})}\||a|\|_{2r}^{-1}\||b|\|_{2r}^{-1}\||a^*|X|b|\|_{\mathcal{S}_p[\mathcal{H};\mathcal{X}]}.
\end{align}
The local unitary invariance of the Pisier norms of $\cS_q[\cH;\cX]$ for any $q\in[1,\infty]$ follows via \Cref{thm:PisierOriginal}:
\begin{align}
    \|U_1XV_1\|_{\cS_q[\cH;\cX]} &= \inf_{\underset{U_1XV_1=a_1Yb_1}{a,b\in\cS_{2q}(\cH), Y\in\cS_\infty[\cH;\cX]}}\|a\|_{2q}\|b\|_{2q}\|Y\|_{\cS_\infty[\cH;\cX]} \\
    &= \inf_{\underset{X=U_1^*a_1Yb_1V_1^*}{a,b\in\cS_{2q}(\cH), Y\in\cS_\infty[\cH;\cX]}}\|a\|_{2q}\|b\|_{2q}\|Y\|_{\cS_\infty[\cH;\cX]}\\
    &= \inf_{\underset{X=a^\prime_1Yb^\prime_1}{a^\prime,b^\prime\in\cS_{2q}(\cH), Y\in\cS_\infty[\cH;\cX]}}\|Ua^\prime\|_{2q}\|b^\prime V\|_{2q}\|Y\|_{\cS_\infty[\cH;\cX]} \\
    &= \inf_{\underset{X=a^\prime_1Yb^\prime_1}{a^\prime,b^\prime\in\cS_{2q}(\cH), Y\in\cS_\infty[\cH;\cX]}}\|a^\prime\|_{2q}\|b^\prime \|_{2q}\|Y\|_{\cS_\infty[\cH;\cX]} = \|X\|_{\cS_q[\cH;\cX]}.
\end{align}

\medskip

To prove that we can assume the operators $a,b$ in the optimization to have support governed by the left-, respectively right-marginal support of $X$ we can argue exactly as in \cite[Lemma 3.2]{Kochanowski.2026}. We hence don't explicitly repeat the argument here, but mention that the generalized Hölder inequality for the Schatten-$p$-quasi-norms used in the proofs of \cite[Section 3.1]{Kochanowski.2026} needs to be replaced by \Cref{thm:PisierOriginal}.

\medskip

Finally concerning self-adjoint operators $X=X^*$ we will first show that, already in Pisier's formula \Cref{thm:PisierOriginal}, the optimizers $a,b$ can be chosen equal $a=b$ when $X=X^*$. This will then via $(i)$ and $(ii)$ directly imply the claim for this theorem. The proof we present below is inspired by the one of \cite[Proposition 5.2 (iii)]{Bardet.2024}, which proves that claim for the special case $\cX=\cS_\infty(\cH_2)$. 
For completeness we mention that the proof of \cite[Lemma 3.2]{Kochanowski.2026}, which relies only on matrix analytic arguments, would suffice to prove this.
Let $X=X^*\in\cS_q[\cH;\cX]$ and
WLOG assume $\|X\|_{(q;\cX)}=1$. By Pisier's formula \Cref{thm:PisierOriginal} there exists positive $a,b\in \cS^+_{2q}(\cH)$ and $Y\in \cS_\infty[\cH;\cX]$ s.t. $X=aYb$ and $\|a\|_{2q}=\|b\|_{2q}=1, \|Y\|_{(\infty;\cX)}\leq 1+\epsilon$.    
Defining $Q:= \frac{1}{\sqrt{2}}(a^2+b^2)^\frac{1}{2}$, denote its generalized Moore-Penrose inverse with $Q^{-1}$, then we can rewrite
\begin{align}
    X=aYb=X^*=bY^*a\equiv QZQ,
\end{align} where 
\begin{align}
Z&:= \frac{1}{2}Q^{-1}(aYb+bY^*a)Q^{-1}=\frac{1}{2}\begin{pmatrix}
    Q^{-1}a & Q^{-1}b    
    \end{pmatrix}\begin{pmatrix}
        Y & 0 \\
        0 & Y^*
    \end{pmatrix}
    \begin{pmatrix}
        bQ^{-1} \\ aQ^{-1}
    \end{pmatrix},
\end{align} since the support projections of both $a$ and $b$ are in positive operator order dominated by the projection $QQ^{-1}=Q^{-1}Q$.
Direct computation yields
\begin{align}
    \|Q\|_{2q} = \left\|\frac{1}{2}a^2+\frac{1}{2}b^2\right\|_q^{\frac{1}{2}}\leq \left(\frac{1}{2}\|a\|^2_{2q}+\frac{1}{2}\|b\|^2_{2q}\right)^\frac{1}{2}\leq 1
\end{align} by triangle inequality and assumption. Further we get, using the $C^*$-identity
\begin{align}
    \|\begin{pmatrix}
        Q^{-1}a & Q^{-1}b
    \end{pmatrix}\|_\infty &= \|Q^{-1}(a^2+b^2)Q^{-1}\|^{\frac{1}{2}}_\infty\leq \sqrt{2}, \\
    \left\|\begin{pmatrix}
        bQ^{-1} \\ aQ^{-1}
    \end{pmatrix}\right\|_\infty &= \|Q^{-1}(a^2+b^2)Q^{-1}\|^{\frac{1}{2}}_\infty\leq \sqrt{2}. 
\end{align} Using the above and Pisier's formula we get
\begin{align}
    \|Z\|_{(\infty;\cX)}\leq \frac{1}{2}\left\|\begin{pmatrix}
        Q^{-1}a & Q^{-1}b
    \end{pmatrix}\right\|_\infty \left\|\begin{pmatrix}
        Y & 0 \\ 0 & Y^*
    \end{pmatrix} \right\|_{(\infty,\infty;\cX)} \left\|\begin{pmatrix}
        bQ^{-1} \\ aQ^{-1}
    \end{pmatrix} \right\|_\infty \leq \|Y\|_{(\infty;\cX)}\leq 1+\epsilon.
\end{align}
So in total we have shown that
\begin{align}
\left\|\left(\frac{Q}{\|Q\|_{2q}}\right)^{-1}X\left(\frac{Q}{\|Q\|_{2q}}\right)^{-1}\right\|_{(\infty;\cX)} = \|Q\|_{2q}^2\|Z\|_{(\infty;\cX)}\leq 1+\epsilon,
\end{align} which since this works for any $\epsilon>0$ finishes the proof, as it shows that 
\begin{align}
    \inf_{a\geq 0,\, X=aYa} \|a\|^2_{2q}\|Y\|_{(\infty;\cX)}=1=\|X\|_{(q;\cX)}.
\end{align}
Now given some $X=X^*$ repeating $(i)$ with this symmetric Pisier's formula when $X=X^*$ yields the infimum formula of the theorem with $a=b$. With this the duality argument of $(ii)$ finally yields the supremum formula with $a=b$.
\end{proof}

\subsection{Proof of \Cref{thm:main.multi.pisier}}\label{app:main.multi.pisier}

In the following we will prove \Cref{thm:main.multi.pisier}.
Recall that the statement we wish to prove is that for a normed multi-indexed Schatten space $\mathcal{Q}$ and an operator space $\mathcal{X}$ then it holds that
\begin{align}
    \|X\|_{(Q:\mathcal{Q};\cX)} &=\inf_{\underset{X=(A\otimes\1)Y(B\otimes\1)}{A,B,Y}}\|AA^*\|_{(Q:\mathcal{Q})}^{\frac{1}{2}}\|B^*B\|_{(Q:\mathcal{Q})}^{\frac{1}{2}}\|Y\|_{(Q:\infty;\cX)} \\ &=\inf_{\underset{\Pi_A=L_X, \Pi_B=R_X}{A,B\geq0}}\|A\|_{(Q:\mathcal{Q})}^{\frac{1}{2}}\|B\|_{(Q:\mathcal{Q})}^{\frac{1}{2}}\|(A^{-\frac{1}{2}}\otimes \1)X(B^{-\frac{1}{2}}\otimes\1)\|_{(Q:\infty;\cX)}.
\end{align} and that, moreover, when $X=X^*$ is self-adjoint, then we can restrict to $A=B^*$

\begin{proof}[Proof of \Cref{thm:main.multi.pisier}]
    We will prove this inductively in the size $|\mathcal{Q}|=n$. The case $n=1$ is just Pisier's theorem \cite[Theorem 1.5]{book:Pisier.1998}, see also \Cref{thm:VariationalSpOperatorspace}. For the induction step $n\to n+1$ take $\mathcal{Q}_{n+1}:=(q_1,...,q_{n+1})\equiv (\mathcal{Q}_n,\,q_{n+1})$. Then we have
\begin{align}
&\|X\|_{(Q:\mathcal{Q}_{n+1};\cX)}\\&=\inf_{A,B}\|AA^*\|_{\mathcal{Q}_n}^\frac{1}{2}\|B^*B\|_{\mathcal{Q}_n}^\frac{1}{2}\|A^{-1}XB^{-1}\|_{(Q:\infty, q_{n+1};\cX)} \\&=\inf_{A,B}\|AA^*\|_{\mathcal{Q}_n}^\frac{1}{2}\|B^*B\|_{\mathcal{Q}_n}^\frac{1}{2}\sup_{a,b\geq0}\|a\|^{-1}_{2q_{n+1}}\|b\|^{-1}_{2q_{n+1}}\|aA^{-1}XB^{-1}b\|_{(Q:q_{n+1};\cX)} \\&=\inf_{A,B}\|AA^*\|_{\mathcal{Q}_n}^\frac{1}{2}\|B^*B\|_{\mathcal{Q}_n}^\frac{1}{2}\sup_{a,b\geq0}\inf_{c,d\geq0}\|a\|^{-1}_{2q_{n+1}}\|b\|^{-1}_{2q_{n+1}}\|c\|_{2q_{n+1}}\|d\|_{2q_{n+1}}\|c^{-1}aA^{-1}XB^{-1}bd^{-1}\|_{( Q:\infty;\cX)} \\ &=\inf_{A,B}\|AA^*\|_{\mathcal{Q}_n}^\frac{1}{2}\|B^*B\|_{\mathcal{Q}_n}^\frac{1}{2}\sup_{a,b\geq0}\inf_{C,D}\|a\|^{-1}_{2q_{n+1}}\|b\|^{-1}_{2q_{n+1}}\|aA^{-1}CC^*A^{*-1}a\|^\frac{1}{2}_{q_{n+1}}\|bB^{*-1}D^*DB^{-1}b\|^\frac{1}{2}_{q_{n+1}}\\ &\hspace{5cm}\times\|C^{-1}XD^{-1}\|_{(Q:\infty;\cX)} \\ &=\inf_{A,B}\|AA^*\|_{\mathcal{Q}_n}^\frac{1}{2}\|B^*B\|_{\mathcal{Q}_n}^\frac{1}{2}\sup_{a,b\geq0}\inf_{C,D\geq0}\|a^{\frac{1}{2q_{n+1}}}A^{-1}CA^{*-1}a^{\frac{1}{2q_{n+1}}}\|^\frac{1}{2}_{q_{n+1}}\|b^{\frac{1}{2q_{n+1}}}B^{*-1}DB^{-1}b^{\frac{1}{2q_{n+1}}}\|^\frac{1}{2}_{q_{n+1}}\\ &\hspace{5cm}\times\|C^{-\frac{1}{2}}XD^{-\frac{1}{2}}\|_{(Q:\infty;\cX)},
\end{align} where we repeatedly used \Cref{thm:VariationalSpOperatorspace}. In the fourth line we set $C:=Aa^{-1}c$ and $D:=db^{-1}B$ and used the standard identity $\|x\|^2_{2q}=\|xx^*\|_q=\|x^*x\|_q$ for Schatten norms. In the last equality we substituted $\widetilde{C}=CC^*, \widetilde{D}=D^*D$ and renamed them to $C,D$ while using that the operator-valued Schatten norms are invariant under local unitaries on the outer index.
Next we will argue how we can swap the infimum and supremum. 
First raising everything to the power $\frac{q_{n+1}}{1+q_{n+1}}$ and then applying \Cref{lem:Inf.Prod.to.Sum} to the infimum yields 
\begin{align}
   &\left(\inf_{C,D\geq0}\|a^{\frac{1}{2q_{n+1}}}A^{-1}CA^{*-1}a^{\frac{1}{2q_{n+1}}}\|^\frac{1}{2}_{q_{n+1}}\|b^{\frac{1}{2q_{n+1}}}B^{*-1}DB^{-1}b^{\frac{1}{2q_{n+1}}}\|^\frac{1}{2}_{q_{n+1}}\|C^{-\frac{1}{2}}XD^{-\frac{1}{2}}\|_{(Q:\infty;\cX)}\right)^\frac{q_{n+1}}{q_{n+1}+1} \\ &= \inf_{C,D\geq0}\frac{1}{2q_{n+1}+2}\|a^{\frac{1}{2q_{n+1}}}A^{-1}CA^{*-1}a^{\frac{1}{2q_{n+1}}}\|^{q_{n+1}}_{q_{n+1}}+\frac{1}{2q_{n+1}+2}\|b^{\frac{1}{2q_{n+1}}}B^{*-1}DB^{-1}b^{\frac{1}{2q_{n+1}}}\|^{q_{n+1}}_{q_{n+1}}\\  &\qquad +\frac{q_{n+1}}{q_{n+1}+1}\|C^{-\frac{1}{2}}XD^{-\frac{1}{2}}\|_{(Q:\infty;\cX)}.
\end{align}
We claim now that the RHS of this expression is jointly convex in $(C,D)$ and concave in $a$ and $b$ each, as well as upper- and lower semi-continuous in the correct sense to apply Sion's minimax theorem \Cref{thm:Sion}. First we have that
\begin{align}
\|a^{\frac{1}{2q_{n+1}}}A^{-1}CA^{*-1}a^{\frac{1}{2q_{n+1}}}\|^{q_{n+1}}_{q_{n+1}} &= \Tr[(a^{\frac{1}{2q_{n+1}}}A^{-1}CA^{*-1}a^{\frac{1}{2q_{n+1}}})^{q_{n+1}}] \\ &=\Tr[(a^{\frac{1}{2q_{n+1}}}A^{-1}CA^{*-1}a^{\frac{1}{q_{n+1}}}A^{*-1}CA^{-1}a^{\frac{1}{2q_{n+1}}})^{\frac{q_{n+1}}{2}}]
\end{align} is continuous and concave in $a\geq0$ since $q_{n+1}\geq 1$ by \cite[Theorem 1.1 (1)]{Zhang.2020} and continuous and convex in $C\geq0$ since $q_{n+1}\geq 1$ by \cite[Theorem 1.1]{Carlen.2008}. The same clearly applies to the second summand. For the third summand we use that the operator space $\cS_\infty[\cH;\cX]\hookrightarrow\cB(\cH\otimes\cK)$ is completely isometrically embedded in $\cB(\cH\otimes\cK)$, when $\cX\subset\cB(\cK)$ \cite{book:Pisier.1998}. We have by \cite[Theorem 5.1 (2)]{Hiai.2013} that
\begin{align}
    \|(C^{-\frac{1}{2}}\otimes\1)X(D^{-\frac{1}{2}}\otimes\1)\|_{(\infty;\cX)}&= \|(C^{-\frac{1}{2}}\otimes\1)X(D^{-\frac{1}{2}}\otimes\1)\|_{\cB(\cH\otimes\cK)}\\&=\big\|\left((C^{-\frac{1}{2}}\otimes\1)X(D^{-1}\otimes\1))X^*(C^{-\frac{1}{2}}\otimes\1)\right)^\frac{1}{2}\big\|_{\cB(\cH\otimes\cK)}
\end{align} is jointly convex in $(C,D)$, which on $\cS_\infty[\cH;\cX]$ reduces to our desired statement.
Continuity in $(C,D)$ follows when restricting to $C,D$ with $\Pi_C=L_X, \Pi_D=R_X$.

Hence \Cref{thm:Sion} applies and we may exchange the supremum over $a,b$ with the infimum over $C,D$. After the exchange, the supremum over $a$ only acts on the first summand, and by \Cref{thm:VariationalSpOperatorspace}~(ii), with symmetric optimizers for the positive operator $A^{-1}CA^{*-1}$ and after undoing the substitution $a\to a^{\frac{1}{2q_{n+1}}}$, and using
\begin{align}
    \sup_{a\geq0,\,\|a\|_1\leq1}\|a^{\frac{1}{2q_{n+1}}}A^{-1}CA^{*-1}a^{\frac{1}{2q_{n+1}}}\|^{q_{n+1}}_{q_{n+1}}=\|A^{-1}CA^{*-1}\|^{q_{n+1}}_{(Q_1\cdots Q_n:\infty,\,Q_{n+1}:q_{n+1})},
\end{align}
and likewise for the second summand. Since these expressions are homogeneous of degree $q_{n+1}$ in $C$, respectively $D$, \Cref{lem:Inf.Prod.to.Sum} can be applied in reverse, and we obtain
\begin{align}
    \|X\|_{(Q:\cQ_{n+1};\cX)} &= \inf_{A,B}\|AA^*\|_{\cQ_n}^{\frac12}\|B^*B\|_{\cQ_n}^{\frac12}\inf_{C,D\geq0}\|A^{-1}CA^{*-1}\|^{\frac12}_{(\infty,q_{n+1})}\\&\hspace{3cm}\times\|B^{*-1}DB^{-1}\|^{\frac12}_{(\infty,q_{n+1})}\|C^{-\frac12}XD^{-\frac12}\|_{(Q:\infty;\cX)} \\
    &= \inf_{C,D\geq0}\Big(\inf_{A}\|AA^*\|_{\cQ_n}\|A^{-1}CA^{*-1}\|_{(\infty,q_{n+1})}\Big)^{\frac12}\\&\hspace{3cm}\times\Big(\inf_{B}\|B^*B\|_{\cQ_n}\|B^{*-1}DB^{-1}\|_{(\infty,q_{n+1})}\Big)^{\frac12}\|C^{-\frac12}XD^{-\frac12}\|_{(Q:\infty;\cX)} \\
    &= \inf_{C,D\geq0}\|C\|^{\frac12}_{(Q:\cQ_{n+1})}\|D\|^{\frac12}_{(Q:\cQ_{n+1})}\|C^{-\frac12}XD^{-\frac12}\|_{(Q:\infty;\cX)},
\end{align}
where in the last line we used the induction hypothesis for the positive operators $C$ and $D$ with the operator space $\cS_{q_{n+1}}(\cH_{Q_{n+1}})$ and symmetric optimizers, together with $\|\cdot\|_{(\cQ_n;\cS_{q_{n+1}})}=\|\cdot\|_{(\cQ_{n+1})}$. This is the second expression in \eqref{equ:main.multi.pisier}; the first one follows from it by polar decomposition.

In the case where $X=X^*$ repeat the analogous argument in the proof of \Cref{thm:VariationalSpOperatorspace} or use that statement in the induction step above to derive the claim.
\end{proof}

\end{document}